\documentclass[11pt]{article}

\usepackage[preprint]{acl}

\usepackage{times}
\usepackage{latexsym}
    
\usepackage[T1]{fontenc}

\usepackage[utf8]{inputenc}

\usepackage{microtype}

\usepackage{inconsolata}

\usepackage[utf8]{inputenc} 
\usepackage[T1]{fontenc}    
\usepackage{hyperref}       
\usepackage{url}            
\usepackage{booktabs}       
\usepackage{amsfonts}       
\usepackage{nicefrac}       
\usepackage{microtype}      
\usepackage{xcolor}         
\usepackage{amsmath}     
\usepackage{amssymb}     
\usepackage[table]{xcolor}
\usepackage{tabularx}
\usepackage{enumitem}

\definecolor{echoshade}{HTML}{FFF2CC}
\definecolor{grposhade}{HTML}{E8F1FF}   
\definecolor{oracleshade}{HTML}{EAF7EA} 
\usepackage{bm}
\newcommand{\ms}[2]{#1{\scriptsize $_{\pm #2}$}}
\newcommand{\EE}{\mathbb{E}}
\usepackage{bm}          
\usepackage{bbm}         
\usepackage{cleveref}
 \usepackage{wrapfig}
\usepackage{graphicx}    
\usepackage{subcaption}
\usepackage{multirow}    
\usepackage{array}       
\definecolor{darkgold}{RGB}{184, 134, 11} 
\newcommand{\an}[1]{\textcolor{black}{#1}}
\newcommand{\nk}[1]{\textcolor{black}{#1}}
\newcommand{\sa}[1]{\textcolor{black}{#1}}
\newcommand{\todo}[1]{\textcolor{red}{TODO: #1}}

\usepackage{amsthm}
\usepackage{arydshln}
\theoremstyle{plain}
\newtheorem{theorem}{Theorem}[section]
\newtheorem{proposition}[theorem]{Proposition}
\newtheorem{lemma}[theorem]{Lemma}
\newtheorem{corollary}[theorem]{Corollary}

\theoremstyle{definition}
\newtheorem{definition}[theorem]{Definition}
\newtheorem{assumption}[theorem]{Assumption}

\theoremstyle{remark}
\newtheorem{remark}[theorem]{Remark}
\usepackage[most]{tcolorbox}
\usepackage{pdflscape}
\newtcolorbox{theorybox}{
    enhanced,
    breakable=true,
      colback=white,
    colframe=blue!50!black,
    boxrule=0pt,
    leftrule=3pt,
    arc=1.5mm,
    left=1.2mm,
    right=1.2mm,
    top=1mm,
    bottom=1mm,
    before skip=0.8em,
    after skip=0.8em
}
\usepackage[edges]{forest} \usepackage{xcolor} \definecolor{echoshade}{HTML}{FFF2CC} \definecolor{lightblue}{HTML}{EAF3FF} \definecolor{lightgreen}{HTML}{EAF7EA} \definecolor{lightgraybox}{HTML}{F3F4F6}

\usepackage{tcolorbox}
\usepackage{listings}
 
\tcbuselibrary{skins, breakable}
 
\definecolor{promptgray}{RGB}{245,245,245}
\definecolor{promptborder}{RGB}{180,180,180}
\definecolor{headergray}{RGB}{210,210,210}
\definecolor{keywordblue}{RGB}{0,70,150}
 
\tcbuselibrary{skins, breakable}
\newtcolorbox{promptbox}[2][]{%
  enhanced, breakable,
  colback=promptgray, colframe=promptborder,
  coltitle=black, fonttitle=\small\bfseries\sffamily,
  title={#2},
  attach boxed title to top left={yshift=-2mm, xshift=4mm},
  boxed title style={colback=headergray, colframe=promptborder, sharp corners},
  sharp corners,
  top=4mm, left=4mm, right=4mm, bottom=3mm,
  #1
}

\usepackage{graphicx}
\usepackage{xcolor}

\title{\textsc{ECHO}: Learning Epistemically Adaptive Language Agents with Turn-Level Credit
}

\author{%
  Abhijnan Nath$^{1,2}$\thanks{Work done while at Colorado State University.} \quad
  Sifatul Anindho$^{1}$ \quad
  Nikhil Krishnaswamy$^{1}$ \\[0.3em]
  \small $^{1}$Situated Grounding and Natural Language (SIGNAL) Lab \\[-0.10em]
  \small Department of Computer Science, Colorado State University, Fort Collins, CO 80523 USA \\[-0.10em]
  \small $^{2}$Adobe Inc., Seattle, WA USA \\[-0.10em]
  \small \texttt{\{abhijnan.nath,sifatul.anindho,nkrishna\}@colostate.edu}
}

\begin{document}
\maketitle
\begin{abstract}
Effective multi-turn agents must decide what information to seek, how to use new evidence, and when they are certain enough to act.
We introduce \textbf{Epistemic Decision Processes} (EDPs), a belief-state formulation of multi-turn information seeking in which actions produce external observations that update the agent's posterior over a latent task variable. EDPs make epistemic adaptivity explicit: good policies choose actions that are useful under the current belief, not merely those that correlate with eventual success. 
We then introduce \textsc{ECHO} (\textbf{E}pistemic \textbf{C}redit for \textbf{H}istory-Conditioned \textbf{O}ptimization), a practical clipped policy-gradient objective that assigns turn-level credit using posterior-sensitive rewards. In the \textbf{Clue Selector Game}, a novel controlled evidence-seeking benchmark, and in \textbf{CRAFT}, a challenging multiagent coordination game, we show that \textsc{ECHO} substantially improves resolution, information gain, and efficiency over trajectory-level GRPO, and matches or exceeds frontier baselines on epistemic metrics such as grounding, recovery, and calibration. Code and data: \url{https://github.com/csu-signal/echo-edp}
\end{abstract}

\section{Introduction}
\vspace*{-2mm}


In multi-step problem solving,
effective agents must decide what evidence to gather next; they do not merely act upon incomplete information~\citep{howard1966information,pauker1980threshold,pirolli1999information,settles2009active}. \an{This skill is central to human learning and question asking~\citep{chouinard2007children,geva2021didaristotleuselaptop,meder2019stepwise,rothe2018people}, and increasingly so to LLM agents that search the web~\citep{deng2023mind2web,yao2022webshop,Zhou2023WebArenaAR}, call tools~\citep{schick2023toolformer}, inspect documents~\citep{zhang-etal-2025-belle}, and revise actions as observations arrive~\citep{ferrag2025llm,yao2022react}.} In such settings, the value of an \nk{information-seeking} question depends on the current belief state: what is already known, what uncertainty remains, and whether the answer would change future action~\citep{kaelbling1998planning}. Thus, adaptivity is fundamentally \textit{epistemic}~\citep{meder2019stepwise,rothe2018people}---good agents choose actions that reduce \nk{not just global uncertainty, but the uncertainty} that \nk{matters most in the current context}.

\begin{figure*}
    \centering
    \begin{subfigure}[t]{0.5\textwidth}
        \centering
        \includegraphics[width=\linewidth]{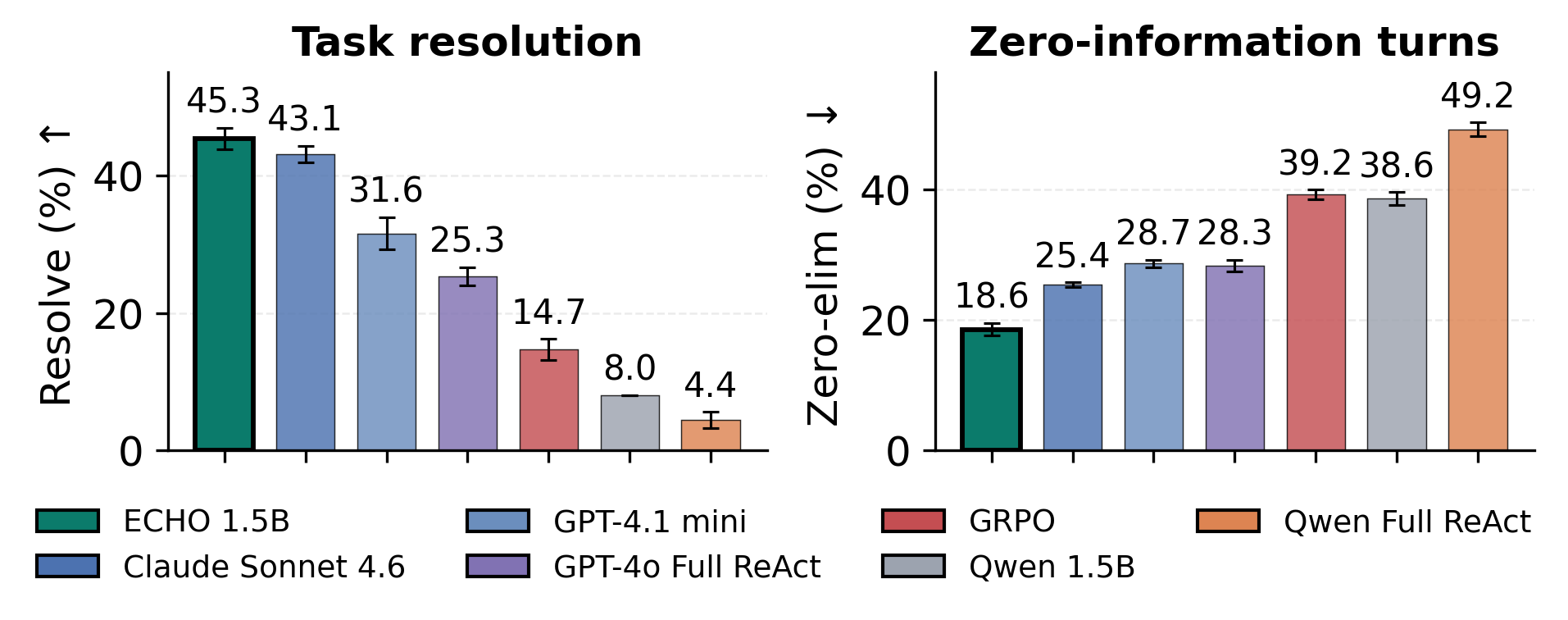}
    \end{subfigure}
    \begin{subfigure}[t]{0.28\textwidth}
        \centering
        \includegraphics[width=\linewidth]{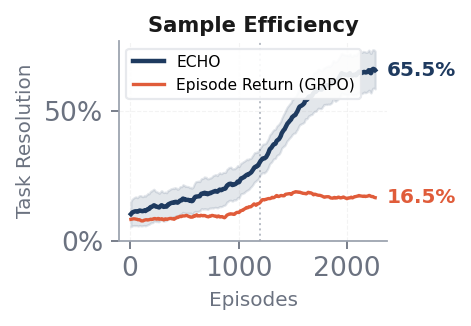}
    \end{subfigure}

    \vspace*{-3mm}
  \caption{
Results on multi-turn information seeking in the \textbf{Clue Selector Game (CSG)}. \textbf{A 1.5B \textsc{ECHO}-trained Qwen2.5 policy reaches frontier-level task resolution, by asking questions that more reliably update the candidate belief state.} Compared with GRPO and prompting baselines, \textsc{ECHO} produces fewer zero-elimination turns and learns substantially \textit{faster}, illustrating the central idea of epistemic credit: useful actions should be rewarded when they reduce uncertainty under the current belief state, not only when the final trajectory succeeds.
}
    \label{fig:intro_teaser}
    \vspace*{-3mm}
\end{figure*}

\nk{In multi-turn tasks, intermediate rewards are often unavailable, or must \sa{be} learned approximately through process reward models~\citep{choudhury2025processrewardmodelsllm,hu2026playing20questiongame,lee2025prints}. 
The agent's uncertainty is often organized around a hidden variable of interest, such as a goal or sub-goal, that must be inferred through interaction rather than revealed directly~\citep{abdulhai2023lmrl,mazzaccara-etal-2024-learning,zhang-etal-2024-probing,zhang2024the}. While the agent may observe the full history of its own queries and the resulting feedback, the value of the next action still depends on how that history changes its belief. This creates a \textit{credit-assignment} problem that existing methods do not fully address~\citep{bertolazzi2023chatgpt,laban2026llms, patil2025berkeley}.}

\an{In this paper, we \textbf{first} argue that information-seeking agents need \textit{epistemic} credit assignment, \nk{reflecting the utility of evidence elicited by an action given the agent's current belief, not simply appearance in a successful trajectory.}
We formalize this through \emph{Epistemic Decision Processes} (EDPs), a belief-state view\footnote{
EDPs are a task-level specialization of belief-state decision processes: unlike BAMDPs~\citep{duff2002optimal,martin1967bayesian}, which typically model uncertainty over transition/reward parameters or latent MDP identities, EDPs focus on uncertainty over a hidden task variable under a fixed interaction interface; see Sec.~\ref{sec:edp_theory} and Appendix~\ref{app:proofs} for details.
} of multi-turn information seeking in which actions produce external observations that update a posterior over a latent task variable. \textbf{Second}, we show there exist long-horizon information-seeking tasks where belief-sensitive policies are necessary for strong performance, and trajectory-level returns are too coarse to identify which turns actually reduced uncertainty.}
 
\nk{\textbf{Finally},} we introduce \textsc{ECHO}, a clipped policy-gradient objective that assigns turn-level credit using posterior-sensitive rewards. \nk{\textsc{ECHO} trains} agents to ask questions that do not just correlate with eventual success, but are useful under the current belief state. \textsc{ECHO} beats or competes with proprietary model baselines, and outperforms GRPO and prompting baselines in resolution, information gain, and sample efficiency, while showing strong grounding, recovery, and calibration.
We evaluate ECHO in the \emph{Clue Selector Game} (CSG), \nk{a deduction task} where 
the hypothesis space, candidate posterior, and belief updates are exactly measurable. 
\textbf{Empirically, \textsc{ECHO} trains a 1.5B policy to reach frontier-level task resolution, reduce zero-elimination turns, and learn substantially faster than standard RL methods.}
\nk{In \textit{CRAFT} \cite{nath2026craft}, a substantially more complex multiagent coordination scenario, \textsc{ECHO}-trained Qwen2.5-7B likewise reaches competitiveness with frontier models, with higher game completion rates in fewer turns.} 

\vspace*{-2mm}
\section{Related Work}

\vspace*{-2mm}
\paragraph{Temporal credit assignment and in-context learning}
Standard temporal credit assignment~\citep{barto2003recent,schulman2017proximalpolicyoptimizationalgorithms,shao2024deepseekmath,sutton1984temporal} can propagate delayed rewards to earlier turns in a task; it does not automatically identify if an action reduced uncertainty under the belief state at that moment.
LLM agents' full query--observation history can be placed \textit{in-context}~\citep{brown2020language,olsson2022context} for the model to update its belief and ask an informative question. However, recent work suggests that in-context updating is sensitive to attention biases, context-length, and implicit priors shaped by pretraining~\citep{kossen2024incontext,liu-etal-2024-lost, falck2024context}. Bayesian methods, e.g., BED-LLM, compute informative queries at inference time using expected information gain~\citep{choudhury2025bed,echarghaoui2026balar,hu2026playing20questiongame,mazzaccara-etal-2024-learning,zhang-etal-2024-probing}. These methods highlight why belief-sensitive learning is difficult in realistic language tasks: the hypothesis space's pretraining and implicit priors make the agent's posterior difficult to inspect or credit at each step. Moreover, inference-time strategies~\citep{jeong2025reflectthenplan,yao2022react} must recompute query utility every turn; \nk{they do not bake in adaptive querying to the policy.}

\vspace*{-2mm}
\paragraph{Bayesian learning and belief-conditioned policies}
\an{
Bayes-Adaptive MDPs formalize epistemic adaptivity by augmenting state with a posterior over environment parameters or latent MDP identities~\citep{bellman1956problem,duff2002optimal,Ghavamzadeh_2015}, with deep RL extensions using approximate inference~\citep{Zintgraf2020VariBAD}. Language-based information seeking \textit{also} requires crediting actions for changing belief; otherwise, outcome-based RL can lock agents into low-information interaction patterns~\citep{zou2026informationselflockingreinforcementlearning}. 
Training-time Bayes-adaptive methods such as BARL encourage reflective exploration within generated reasoning traces~\citep{zhang2025markovianreflectiveexplorationbayesadaptive}, but recent work cautions that visible reasoning traces are not necessarily faithful evidence of reasoning competence~\citep{palod2025performative,samineni2025localcoherenceglobalvalidity}.}



\vspace*{-2mm}
\paragraph{Controlled information-seeking environments}
 
\an{Common question-asking settings study how agents infer hidden targets through interaction~\citep{andukuri2024stargateteachinglanguagemodels,kobalczyk2025activetaskdisambiguationllms,lialfa,zhang-etal-2024-probing,zhou2024archer}. However, non-uniform LLM priors affect question selection and evaluation~\citep{mazzaccara-etal-2024-learning}; intermediate progress is usually tracked via hand-crafted knowledge-bases~\citep{zhao2016towards} or approximated through learned reward networks, process reward models, or inference-time information-gain estimates~\citep{choudhury2025bed,choudhury2025processrewardmodelsllm,hu2026playing20questiongame,lee2025prints}. Final accuracy can miss process-level failures such as information self-locking, poor grounding, or failure to recognize knowledge gaps~\citep{shao2025llm,zou2026informationselflockingreinforcementlearning}. 
Our controlled \textbf{Clue Selector Game} (CSG) preserves a language-mediated multi-turn loop---action, response, belief update, action---while making the hypothesis space explicit, the posterior exactly measurable, and process-level metrics such as grounding, redundancy, and recovery, directly computable.} We use CRAFT \cite{nath2026craft} as our secondary experimental setting. Where the CRAFT paper examines collaborative failure in off-the-shelf models, this paper proposes a novel alignment method to improve coordination.

\vspace*{-2mm}
\section{Epistemic Decision Processes}
\label{sec:edp_theory}
\vspace*{-2mm}

\begin{figure}[ht!]
    \vspace*{-2mm}
    \centering
    \includegraphics[
        width=.85\linewidth,
        trim=0.3in 0.25in 0.3in 0.45in, clip,
        clip
    ]{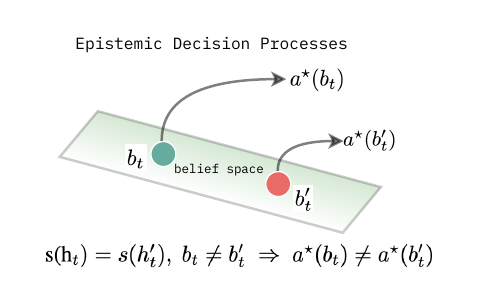}
    \caption{
EDP intuition. Two different histories can result in the same observable task state \(s(h_t)=s(h'_t)\) while inducing different beliefs \(b_t\neq b'_t\) over the latent task variable. An \textit{epistemically adaptive} policy conditions on the belief state at each turn $t$ to choose different useful actions, so \(a^\star(b_t)\neq a^\star(b'_t)\). A belief-agnostic policy that depends only on \(s(h_t)\) and \(t\), but not \(b_t\), cannot make this distinction.
}
  
    \label{fig:edp}
    \vspace*{-4mm}
\end{figure}




We formalize multi-turn information seeking as a decision problem over beliefs. An agent acts under uncertainty about a latent task variable \(z^\star\), chooses actions that elicit external observations, and updates its belief \nk{from those observations} before acting again. 
The main dense reward is \textit{posterior-sensitive}: the same question can be informative under one candidate set but not another.



When the policy interacts with an external environment over multiple turns: each action produces an observation which changes the posterior. The next action should be chosen under the updated belief~\citep{abdulhai2023lmrl}. 
\nk{Intuitively, optimal policies in this setting are}
\emph{epistemically adaptive}: they choose actions because they are useful under the current belief, not merely because they are plausible \nk{given the observable} surface history.



We formalize this as an \emph{Epistemic Decision Process} (EDP), a finite-horizon interaction defined by:
\vspace*{-4mm}
\[
    \mathcal{E}=\left(\mathcal{Z},\rho_0,\mathcal{A},\Omega,O,u,T\right).
\]
\vspace*{-6mm}

Here, \(z^\star \in \mathcal{Z}\) denotes the latent task variable sampled from prior \(\rho_0\)---for example, a target answer or goal state---while \(\mathcal{A}\), \(\Omega\), \(O\), \(u\), and \(T\) denote the action space, observation space, observation kernel, evidential utility, and horizon. At turn \(t\), the interaction history \(h_t=(a_0,o_0,\ldots,a_{t-1},o_{t-1})\) induces a posterior belief \(b_t(z):=\mathbb{P}(z^\star=z \mid h_t)\).\footnote{In LLM implementations, the serialized prompt may include a textual state summary, such as remaining hypotheses. We treat this as an encoding of the epistemic component \(b_t\).} The agent samples an action \(a_t \sim \pi_\theta(\cdot \mid h_t)\), receives an external observation \(o_t \sim O(\cdot \mid z^\star,h_t,a_t)\), and \an{updates its posterior via a Bayesian belief-update operator \(b_{t+1}=\tau(b_t,h_t,a_t,o_t)\), where \(\tau\) is induced by \(O\). }Fig.~\ref{fig:edp} shows a conceptualization of the EDP.

\vspace*{-2mm}
\paragraph{LLM Instantiation}
An LLM instantiation of an EDP is an autoregressive generation policy that gets a serialized history \(h_t\), including the task instruction, current state summary, and previous action--observation pairs. The model generates a structured output that is parsed as the action or generated query \(a_t\). The observation \(o_t\) from the environment 
\nk{is used to} filter the hypotheses still consistent with the evidence. Thus, the transition is not merely deterministic text concatenation~\citep{samineni2026rlonlyanalyzingstructural}: the next prompt appends \((a_t,o_t)\) but the observation itself is produced by an external observation kernel \(o_t\sim O(\cdot\mid z^\star,h_t,a_t)\), the distribution from which the environment draws its response $o_t$. Each turn is therefore a genuine environment interaction. The decision-relevant state is the posterior uncertainty induced by the history, i.e., the augmented epistemic state \((s(h_t),b_t,t)\).

\vspace*{-2mm}
\paragraph{Bayesian Value and Advantage}
\label{ssec:bayesian_value_advantage}
Training an epistemically adaptive policy requires crediting actions under the agent's current posterior belief, rather than only under the final episode outcome. For a fixed latent task \(z\), let
\(Q_z^\pi(h_t,a_t)=\EE_{\pi,O}[\sum_{s=t}^{T-1}u(b_s,a_s,o_s)+R_T(h_T,z)\mid h_t,a_t,z^\star=z]\)
denote the value of taking action \(a_t\) after history \(h_t\), where \(R_T\) is any terminal task reward. Since the agent does not observe \(z^\star\), the relevant value is posterior-weighted: \(Q_{\mathrm{B}}^\pi(h_t,a_t)=\EE_{z\sim b_t}[Q_z^\pi(h_t,a_t)]\). The corresponding advantage is \(A_{\mathrm{B}}^\pi(h_t,a_t)=Q_{\mathrm{B}}^\pi(h_t,a_t)-\EE_{a'\sim\pi(\cdot\mid h_t)}[Q_{\mathrm{B}}^\pi(h_t,a')]\). \an{
Since \(b_t=b(h_t)\), the dependence on the current belief is implicit in the history argument; equivalently, \(A_{\mathrm{B}}^\pi(h_t,a_t)\) is the advantage function of the induced belief-state MDP and evaluates the action under the posterior induced by \(h_t\).} Intuitively, this asks whether \(a_t\) is better than alternative actions available under the same posterior belief---we will call this the {\bf Bayesian advantage}. Estimating this exactly would require expensive multiturn branching actions from the same history \(h_t\)~\citep{abdulhai2023lmrl,guo2025segment,kazemnejad2025vinepporefiningcreditassignment,zhang2025markovianreflectiveexplorationbayesadaptive}. This raises two crucial questions: what policy class is the ``Bayesian advantage'' meant to favor, and how can we approximate its credit signal without exhaustive same-history branching? We first answer the policy question by distinguishing belief-agnostic policies from epistemically adaptive ones; we then introduce a tractable Monte Carlo policy-gradient estimator that uses turn-level epistemic rewards as a practical surrogate for the belief-conditioned advantage.

\vspace*{-2mm}
 \subsection{Distinguishing Epistemically Adaptive from Belief-Agnostic Policies}
\label{ssec:epistemically_adaptive_policies}
 \vspace*{-2mm}
 
\an{\nk{An {\it epistemically adaptive policy} uses the {\it belief} induced by the interaction history, not just the history itself.} Let \(s(h_t)\) denote the observable task state (Fig.~\ref{fig:edp}). These variables constrain what the agent can do at turn \(t\), but do not themselves encode the posterior over \(z^\star\). A policy is \emph{belief-agnostic} if its dependence on history factors only through this observable state, i.e., there exists a decision rule \(g\) such that \(\pi(a_t\mid h_t)=g(a_t\mid s(h_t))\). In contrast, a policy is \emph{epistemically adaptive} if its action distribution can also depend on the posterior belief induced by the history, i.e., there exists a decision rule \(f\) such that \(\pi(a_t\mid h_t)=f(a_t\mid s(h_t),b_t)\). We say the policy is \emph{strictly} epistemically adaptive if there exist histories \(h_t,h'_t\) resulting in the same observable \(s(h_t)=s(h'_t)\) but different beliefs \(b_t\neq b'_t\), such that \(\pi(\cdot\mid h_t)\neq\pi(\cdot\mid h'_t)\).}


\begin{theorybox}
\begin{proposition}[Belief-state sufficiency]
Assume the observation kernel and utility function depend on \(h_t\) only through the observable task state \(s(h_t)\) and belief \(b_t\). Then there exists an optimal policy $\pi^\star$ for the EDP that is Markovian in the augmented epistemic state \((s(h_t),b_t,t)\). Full proof is in Appendix~\ref{app:proofs}.
\label{proposition:belief_sufficiency_main}
\end{proposition}
\end{theorybox}
\vspace{-2mm}

Proposition~\ref{proposition:belief_sufficiency_main} defines the state space. Once the belief \(b_t\) is part of the state, optimal behavior can be written as a Markov policy over \((s(h_t),b_t,t)\). The important question is therefore \textbf{not} whether a Markov policy exists, but whether the training signal can teach the policy to use the belief as uncertainty is resolved~\citep{zhang2025markovianreflectiveexplorationbayesadaptive}. 
Theorem~\ref{thm:necessity}(i) states there are EDPs where an adaptive policy succeeds with probability \(1\), while any belief-agnostic policy succeeds with probability at most \(2^{1-d}\). Theorem~\ref{thm:necessity}(ii) shows why sparse trajectory-level returns reveal final success, but not the belief-conditioned credit of each intermediate action. Together, they motivates turn-level epistemic credit rather than trajectory-level credit.

\begin{theorybox}
\begin{theorem}[Adaptive policies require epistemic credit]
\label{thm:necessity}
Let $\mathcal{E}$ be an Epistemic Decision Process.
\begin{enumerate}[label=\roman*.] 
    \item \textbf{(Performance gap.)}
    For every depth $d \geq 1$, there exists an EDP $\mathcal{E}_d$ with
    horizon $d$ in which an epistemically adaptive policy achieves
    $\mathbb{P}(\mathrm{success}) = 1$, while every belief-agnostic policy
    satisfies $\mathbb{P}(\mathrm{success}) \leq 2^{1-d}$. Thus, the cost of ignoring belief compounds exponentially with the decision horizon.

    \item \textbf{(Estimator insufficiency.)}
    Let $G = \sum_{s=0}^{T-1} r_s$ be the total episode return and fix
    any turn $t$. For any trajectory-level \nk{advantage} estimator
    $\widehat{A}_t^{\mathrm{traj}} = \phi(G)$,
    \begin{align}
        &\notag\inf_{\phi}\,
        \mathbb{E}\!\left[
            \bigl(\phi(G) - A_{\mathrm{B}}^\pi(h_t,a_t)\bigr)^2
        \right]\\
        = &\notag\mathbb{E}\!\left[
            \operatorname{Var}\!\left(
                A_{\mathrm{B}}^\pi(h_t,a_t) \mid G
            \right)
        \right] > 0
    \end{align}
 whenever the same total return can arise from belief-distinct
decisions with different Bayesian advantages.
\end{enumerate}

Together, (i) and (ii) show that, for a family of EDPs where a useful next action depends on the current posterior, epistemic adaptivity is necessary for strong long-horizon performance, and 
aggregate trajectory return alone aliases belief-distinct decisions that require different local credit. Full proofs are in Appendix~\ref{app:proofs}. 
 
\end{theorem}
\end{theorybox}

\vspace*{-3mm}
\section{Depth-wise Epistemic Credit Optimization (\textsc{ECHO})}
\label{sec:echo}
\vspace*{-3mm}
The above leads to our core insight—\textit{actions should receive credit based on whether they were useful under the belief state in which they were chosen, not only based on whether the entire episode eventually succeeded}. We therefore introduce \textsc{ECHO} (\textbf{E}pistemic \textbf{C}redit for \textbf{H}istory-Conditioned \textbf{O}ptimization), a practical turn-level surrogate for Bayesian advantage in multi-turn EDPs. Like GRPO~\citep{shao2024deepseekmath}, \textsc{ECHO} uses group-normalized advantages and a clipped policy-gradient objective~\citep{schulman2017proximalpolicyoptimizationalgorithms}; unlike trajectory-level GRPO, it normalizes posterior-sensitive turn rewards rather than full episode returns. As such, \textsc{ECHO} keeps the practical group-relative clipped optimization form of GRPO, but shifts the credit signal from trajectory-level success to turn-level “epistemic” progress. 


For each training instance, we sample \(G\) parallel rollouts from the old policy. At turn \(t\), rollout \(g\) has history \(h_t^{(g)}\), belief \(b_t^{(g)}=\mathbb{P}(z^\star\mid h_t^{(g)})\), sampled action \(a_t^{(g)}\sim\pi_{\theta_{\mathrm{old}}}(\cdot\mid h_t^{(g)})\), and external observation \(o_t^{(g)}\). We assign a posterior-sensitive reward, $r_t^{(g)} = u\!\left(b_t^{(g)},a_t^{(g)},o_t^{(g)}\right)$,
which measures local epistemic progress, e.g., information gain, reduction in uncertainty, or elimination of inconsistent hypotheses. 
At each decision depth \(t\), we normalize these rewards across the active rollouts: 

\vspace*{-8mm}
    \begin{equation}
    \widehat{A}_{t}^{(g),\mathrm{ECHO}}
    =
    \frac{
        r_t^{(g)}
        -
        \frac{1}{|\mathcal{G}_t|}
        \sum_{j\in\mathcal{G}_t} r_t^{(j)}
    }{
        \operatorname{std}_{j\in\mathcal{G}_t}(r_t^{(j)})+\epsilon
    }
    \label{eq:echo_advantage}
\end{equation}

\noindent
While trajectory-level GRPO broadcasts a single full-episode advantage across all actions in a rollout, \textsc{ECHO} (Eq.~\ref{eq:echo_advantage}) computes group-normalized advantages at each decision depth. This gives each structured turn action credit according to the epistemic progress it produced under the current posterior~\citep{jin2026dgpo, nath2026owen, samineni2026rlonlyanalyzingstructural}. Our approach is not an exact belief-local advantage after rollouts diverge, since different rollouts may reach different histories and beliefs at the same turn—given the exponential state-action spaces in language based RL using LLMs~\citep{dong2024rlhf, srivastava2024policy, zhou2024archer}. The key point that each reward is evaluated relative to its rollout's own posterior, so the update preserves belief-dependent variation in action quality while avoiding exhaustive same-history branching.
Let
\vspace*{-2mm}
\[
    \rho_t^{(g)}(\theta)
    =
    \frac{
        \pi_\theta(a_t^{(g)}\mid h_t^{(g)})
    }{
        \pi_{\theta_{\mathrm{old}}}(a_t^{(g)}\mid h_t^{(g)})
    },
\]
\vspace*{-2mm}

\noindent
where for LLM agents \(\pi_\theta(a_t^{(g)}\mid h_t^{(g)})\) is the product of token probabilities for the structured action. \textsc{ECHO} optimizes the clipped objective:
\vspace*{-2mm}
\begin{align}
    \mathcal{J}_{\mathrm{ECHO}}(\theta)
    &=
    \notag\mathbb{E}
    \Bigg[
    \frac{1}{\sum_t |\mathcal{G}_t|}
    \sum_{t=0}^{T-1}
    \sum_{g\in\mathcal{G}_t}
    \min\!\Big(
        \rho_t^{(g)}\widehat{A}_{t}^{(g)},\\
        &\operatorname{clip}\!\left(
            \rho_t^{(g)},
            1-\eta,
            1+\eta
        \right)
        \widehat{A}_{t}^{(g)}
    \Big)
    \Bigg]
    \label{eq:echo-objective}
\end{align}
\vspace*{-8mm}


\vspace*{-2mm}
\section{Experiments}
\label{sec:experiments}
\vspace*{-2mm}

We instantiate EDPs and evaluate in two tasks, each with an explicit EDP interpretation (Appendix~\ref{app:edp-interp}).

\vspace*{-2mm}
\subsection{Clue Selector Game Environment}
\label{ssec:clue_task}

\vspace*{-1mm}
\begin{figure*}[t]
    \centering
    \includegraphics[width=.9\textwidth]{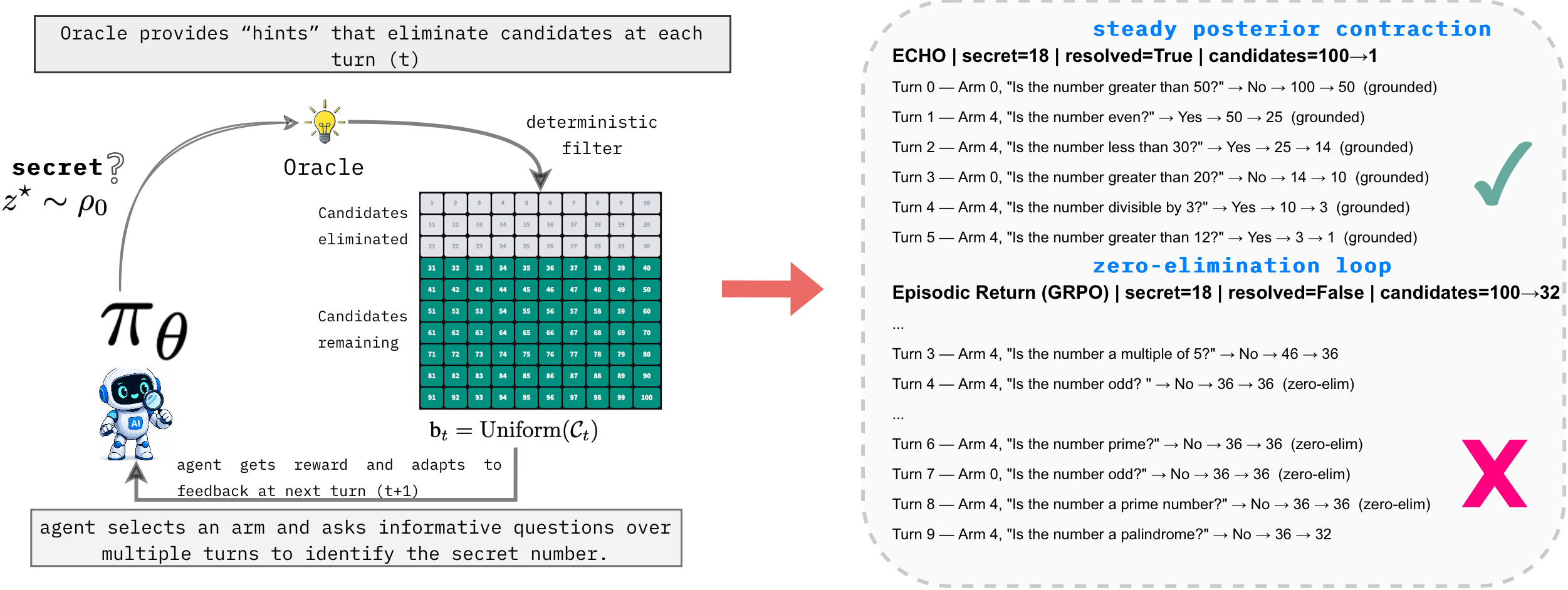}
\vspace*{-3mm}
    \caption{
    Clue Selector Game overview and example trajectory from real evaluation data: \textsc{ECHO} contracts the belief state, while episodic-return GRPO stalls with zero-elimination queries.
    }
    \label{fig:clue_overview}
\vspace*{-6mm}
\end{figure*}

The \emph{Clue Selector Game} (CSG) is a controlled information-seeking task inspired by prior work in multiturn RL~\citep{abdulhai2023lmrl}. In the CSG, the agent must identify a secret number by sequentially querying an oracle (GPT-4o-mini) that can only truthfully answer questions about the secret number's properties (e.g., about parity, divisibility, range) without revealing it directly. The agent acts through an \emph{arm--question interface}: at each turn it routes a natural-language question through one of several clue-holder arms, where each arm ``owns'' certain of the aforementioned domains of information. A successful agent identifies {\it what information is most useful} to seek (the arm), avoids repeatedly asking for information already obtained, and adapts its query strategy as its belief changes. 
Fig.~\ref{fig:clue_overview} shows a high-level overview of the CSG environment. Prompts are given in Fig.~\ref{fig:csg-prompts}. 

 The secret number is the latent task variable \(z^\star=n^\star\). A secret integer \(n^\star \in \{1,\ldots,N\}\) is sampled uniformly at random; we use \(N=100\) for training, while evaluation can include secrets up to \(N=200\) to test generalization. The dominant reward term is fractional candidate elimination, \(r_{\mathrm{info}}=(|\mathcal{C}_t|-|\mathcal{C}_{t+1}|)/|\mathcal{C}_t|\), which measures how much the action contracts the current posterior. Training also includes small shaping terms for valid arm--question matching and query diversity, a redundancy penalty for repetition, and sparse bonuses for resolving the task efficiently. 

\vspace*{-2mm}
\subsection{CRAFT Environment}
\vspace*{-2mm}

As a secondary evaluation setting, we use the {\it Collaborative Reasoning Agents For Construction Tasks} (CRAFT) platform \cite{nath2026craft}, released under the MIT License. CRAFT is an LLM simulator of the DPIP Lego Task \cite{zhu2026distributed}, in which 4 agents---3 Directors and 1 Builder---must build one structure which is described by 3 private side views that are distributed amongst the 3 Directors. They must collaboratively instruct the Builder to build the structure on a publicly-observable board while concealing their private information. This task is multiagent in nature and all agents exchange information under partially observable conditions, and so we use it to test ECHO's generalizability to a setting with additional stochastic dynamics. The information-seeking agent is the Builder, which must infer a move from 3 Director descriptions
and may ask clarifying questions when it cannot tell. CRAFT includes an Oracle that provides up to 8 verified legal moves to the Builder to choose from based on Director instructions. The Directors are GPT-4.1-mini instances, cf. \citet{vanderhoeven2026floutriskllmsstruggle}.

\sa{The dominant reward term is the change in task progress for an executed candidate move, which measures how much the action moves the board toward the target structure. The Builder's uncertainty is reduced by the Directors' responses, including reactions to its moves and answers to clarification questions. Training also includes small shaping terms for unparseable outputs, invalid moves, oracle candidate matching and clarification question, as well as sparse bonuses for completing the structure efficiently. We evaluate baselines on a held-out dataset of 20 structures. Games may start from an empty or partially correctly-completed board.}

\begin{table*}[t]
\centering
\scriptsize
\setlength{\tabcolsep}{4.2pt}
\renewcommand{\arraystretch}{1.08}
\begin{tabular}{lcccccc}
\toprule
\textbf{CSG} &
\textbf{Resolve$\uparrow$} &
\textbf{Zero$\downarrow$} &
\textbf{Qual.$\uparrow$} &
\textbf{Ground$\uparrow$} &
\textbf{GRecover$\uparrow$} &
\textbf{ResAfterZero$\uparrow$} \\
\midrule

\multicolumn{6}{l}{\textit{Frontier/proprietary baselines}} \\
\hdashline
Claude Sonnet 4.6
& \ms{43.1}{1.2}
& \ms{25.4}{0.4}
& \textbf{\ms{0.701}{0.005}}
& \textbf{\ms{0.738}{0.005}}
& \ms{0.185}{0.018}
& \ms{0.168}{0.012} \\

o3-mini
& \ms{27.1}{4.9}
& \ms{29.8}{1.9}
& \ms{0.610}{0.078}
& \ms{0.636}{0.082}
& \ms{0.248}{0.009}
& \ms{0.151}{0.027} \\

GPT-4.1-mini
& \ms{31.6}{2.4}
& \ms{28.7}{0.6}
& \ms{0.615}{0.002}
& \ms{0.640}{0.003}
& \ms{0.285}{0.003}
& \ms{0.159}{0.012} \\

GPT-4o
& \ms{25.8}{2.7}
& \ms{29.8}{1.2}
& \ms{0.625}{0.009}
& \ms{0.662}{0.007}
& \ms{0.286}{0.006}
& \ms{0.156}{0.017} \\

GPT-4o + Full ReAct
& \ms{25.3}{1.3}
& \ms{28.3}{0.9}
& \ms{0.618}{0.002}
& \ms{0.625}{0.001}
& \ms{0.231}{0.014}
& \ms{0.087}{0.010} \\

\midrule
\multicolumn{6}{l}{\textit{Base Qwen2.5 baselines}} \\
\hdashline
Qwen2.5-14B
& \ms{22.7}{2.0}
& \ms{27.3}{0.7}
& \ms{0.535}{0.007}
& \ms{0.537}{0.008}
& \ms{0.290}{0.014}
& \ms{0.125}{0.011} \\

Qwen2.5-7B
& \ms{20.9}{0.9}
& \ms{34.0}{0.5}
& \ms{0.502}{0.002}
& \ms{0.509}{0.003}
& \ms{0.200}{0.010}
& \ms{0.105}{0.024} \\

Qwen2.5-3B
& \ms{8.4}{1.2}
& \ms{36.0}{1.0}
& \ms{0.454}{0.006}
& \ms{0.454}{0.007}
& \ms{0.280}{0.008}
& \ms{0.072}{0.017} \\

Qwen2.5-1.5B
& \ms{8.0}{0.0}
& \ms{38.6}{1.0}
& \ms{0.352}{0.002}
& \ms{0.340}{0.005}
& \ms{0.229}{0.025}
& \ms{0.064}{0.004} \\

\midrule
\multicolumn{6}{l}{\textit{RL-trained Qwen2.5-1.5B baselines}} \\
\hdashline
\rowcolor{grposhade}
Episodic Return (GRPO)
& \ms{14.7}{1.5}
& \ms{39.2}{0.8}
& \ms{0.445}{0.008}
& \ms{0.446}{0.008}
& \ms{0.339}{0.004}
& \ms{0.093}{0.021} \\


\rowcolor{echoshade}
\textbf{\textsc{ECHO} (Ours)}
& \textbf{\ms{45.3}{1.5}}
& \textbf{\ms{18.6}{1.0}}
& \ms{0.670}{0.008}
& \ms{0.718}{0.006}
& \textbf{\ms{0.406}{0.006}}
& \textbf{\ms{0.318}{0.018}} \\

\midrule
\rowcolor{oracleshade}
Greedy oracle (reference)
& \ms{86.2}{0.4}
& \ms{29.6}{0.6}
& \ms{0.700}{0.006}
& \ms{0.704}{0.006}
& \ms{0.031}{0.001}
& \ms{0.088}{0.003} \\

\bottomrule
\end{tabular}

\begin{tabular}{lccccccc}
\toprule
\textbf{CRAFT} &
\textbf{Final Prog. $\uparrow$} &
\textbf{Completion $\uparrow$} &
\textbf{Oracle Match $\uparrow$} &
\textbf{Invalid $\downarrow$} &
\textbf{Parse Fail. $\downarrow$} &
\textbf{Prog. Gained $\uparrow$} &
\textbf{Ep. Length $\downarrow$} \\
\midrule

\multicolumn{8}{l}{\textit{Frontier/proprietary baselines}} \\
\hdashline

Claude Sonnet 5
& \ms{80.9}{0.5}
& \ms{31.7}{3.3}
& \textbf{\ms{86.0}{2.7}}
& \ms{4.5}{1.5}
& \ms{0.3}{0.3}
& \ms{46.4}{1.4}
& \ms{16.9}{0.2} \\

Claude Sonnet 4.6
& \textbf{\ms{82.1}{0.5}}
& \ms{33.3}{6.0}
& \ms{83.6}{1.4}
& \ms{7.5}{0.6}
& \textbf{\ms{0.0}{0.0}}
& \textbf{\ms{47.6}{1.8}}
& \ms{16.4}{0.6} \\

GPT-5.4-mini
& \ms{72.8}{0.8}
& \ms{23.3}{1.7}
& \ms{60.8}{1.8}
& \ms{27.4}{3.7}
& \textbf{\ms{0.0}{0.0}}
& \ms{38.3}{0.5}
& \ms{17.9}{0.3} \\

GPT-4.1-mini
& \ms{78.3}{1.6}
& \ms{28.3}{6.0}
& \ms{64.6}{3.4}
& \ms{25.8}{3.1}
& \ms{0.8}{0.2}
& \ms{43.8}{2.8}
& \ms{17.5}{0.4} \\

\midrule
\multicolumn{8}{l}{\textit{Base Qwen2.5 baselines}} \\
\hdashline

Qwen2.5-72B
& \ms{63.1}{0.5}
& \ms{8.3}{1.7}
& \ms{66.3}{2.1}
& \ms{13.2}{2.0}
& \ms{28.8}{1.4}
& \ms{28.7}{1.8}
& \ms{19.4}{0.2} \\

Qwen2.5-32B
& \ms{67.5}{0.6}
& \ms{11.7}{4.4}
& \ms{47.0}{3.4}
& \ms{35.6}{4.8}
& \ms{2.4}{0.2}
& \ms{33.0}{1.9}
& \ms{19.1}{0.3} \\

Qwen2.5-14B
& \ms{67.6}{1.0}
& \ms{10.0}{2.9}
& \ms{43.0}{2.2}
& \ms{41.7}{2.0}
& \ms{4.1}{0.7}
& \ms{33.1}{1.5}
& \ms{19.1}{0.3} \\

Qwen2.5-7B
& \ms{51.4}{0.6}
& \ms{1.7}{1.7}
& \ms{49.5}{1.7}
& \ms{14.7}{2.0}
& \ms{1.0}{0.6}
& \ms{17.0}{1.4}
& \ms{20.0}{0.0} \\

\midrule
\multicolumn{8}{l}{\textit{RL-trained Qwen2.5-7B models}} \\
\hdashline

\rowcolor{grposhade}
Episodic Return (GRPO)
& \ms{74.8}{0.3}
& \ms{25.0}{5.0}
& \ms{44.9}{1.5}
& \ms{35.6}{1.2}
& \ms{2.0}{0.6}
& \ms{40.3}{1.5}
& \ms{17.8}{0.3} \\

\rowcolor{echoshade}
\textbf{\textsc{ECHO} (Ours)}
& \ms{80.2}{0.3}
& \textbf{\ms{35.0}{5.0}}
& \ms{71.9}{3.0}
& \textbf{\ms{3.2}{1.0}}
& \ms{0.4}{0.3}
& \ms{45.7}{1.6}
& \textbf{\ms{16.2}{0.6}} \\
\bottomrule
\end{tabular}
\vspace*{-2mm}
\caption{
\textbf{Main CSG (top) and CRAFT (bottom) results}. Values are mean with SEM shown in subscript across \textbf{3} independent runs on the evaluation set. Resolve, Zero, and all CRAFT metrics except Ep. Length (turns) are reported in percentage points; all others are normalized scalars. Bold marks the best learned/model result by metric. 
}
\vspace*{-4mm}
\label{tab:combined_results}
\end{table*}

\vspace*{-2mm}
\subsection{Evaluation Strategy}
\vspace*{-2mm}
Our metrics test whether a policy uses accumulated evidence to choose useful actions throughout the trajectory, rather than merely succeeding at the end. 
In these tasks optimal information seeking is not determined by turn number alone, but by the structure of the current posterior over candidates, i.e., the specific surviving candidate set $\mathcal{C}_t$ or possible target structures.
An \textit{epistemically adaptive} agent should therefore ask questions that reduce uncertainty under the current \(\mathcal{C}_t\), avoid redundant queries, recover after low-information observations, and not ask broad questions if the belief state is narrow. Accordingly, we evaluate not only \textbf{final task resolution}, but also task-relevant \textbf{process-level epistemic quality} metrics~\citep{shao2025llm} as a proxy for epistemic adaptability (see Sec.~\ref{sec:results}). 
For CSG, we report: \textbf{Zero}, rate of non-redundant turns that eliminate no candidates; 
\textbf{Qual.}, candidate-elimination quality normalized by an ideal split; 
\textbf{Ground}, rate of turns that are both non-redundant and informative under the current candidate set; 
\textbf{GRecover}, rate at which the next turn after a zero-elimination action is \textit{grounded} as defined above; and \textbf{ResAfterZero}, fraction of episodes that still resolve after at least one zero-elimination event.
For CRAFT, we report: \textbf{Final Progress}, amount of the structure completed at the final ($\leq$20) timestep; \textbf{Completion}, percentage of structures completed within the 20-turn budget; \textbf{Oracle Match}, percentage of executed moves which match an oracle recommendation; \textbf{Invalid} move attempt rate; \textbf{Parse Failures}, \sa{fraction of turns where output cannot be parsed into a move}; \textbf{Progress Gained}, \sa{difference between final and initial progress}; and \textbf{Episode Length}, average length of a game up to a 20 turn max budget.
\an{Further details on evaluation metrics are provided in Appendix~\ref{ssec:main_task_epistemic_metrics_full}.}
 
Empirically, these metrics are strongly tied to task success and inversely correlated with low-information behavior, supporting their use as indicators of epistemically adaptive behavior.

\vspace*{-2mm}
\paragraph{Baselines}
\textbf{Frontier or proprietary models} include Claude Sonnet 4.6, Claude Sonnet 5 (CRAFT only) OpenAI's o3-mini (CSG only), GPT-5.4-mini (CRAFT), GPT-4.1-mini, GPT-4o (CSG), GPT-4.1-nano, and GPT-4o-mini, evaluated zero-shot with the standard task prompt. These models represent strong off-the-shelf performance without task-specific training. \textbf{Open weight models} are Qwen2.5 from 1.5B--14B (CSG) and from 7B--32B (CRAFT) to provide a scaling reference without task-specific prompting or RL. We also include \textbf{GPT-4o w/ full ReAct}~\citep{yao2022react} using a complete Thought/Action/Observation loop with growing history context. The main RL baseline is \textbf{episode-return RL}, which broadcasts the total episode reward uniformly to all turns (and thereby, to tokens) consistent with standard GRPO extended to multi-turn settings. It trains Qwen2.5-1.5B Instruct (CSG)/Qwen2.5-7B Instruct (CRAFT) with the same environment and reward function as \textsc{ECHO}, differing only in the advantage estimator.

For CSG, a \textbf{greedy oracle} reference policy with access to the secret during action selection, that chooses the arm-question pair with maximal candidate elimination, serves as an “upper-bound”~\citep{analyzingnas} diagnostic. Additional baselines within the same families, across more model sizes and prompting variants, are given in Appendix~\ref{app:binary_search_echo}. 

\vspace*{-2mm}
\section{Results}
\label{sec:results}

\begin{figure*}[ht!]
\centering
\begin{subfigure}[t]{0.24\textwidth}
\centering
\includegraphics[width=\linewidth]{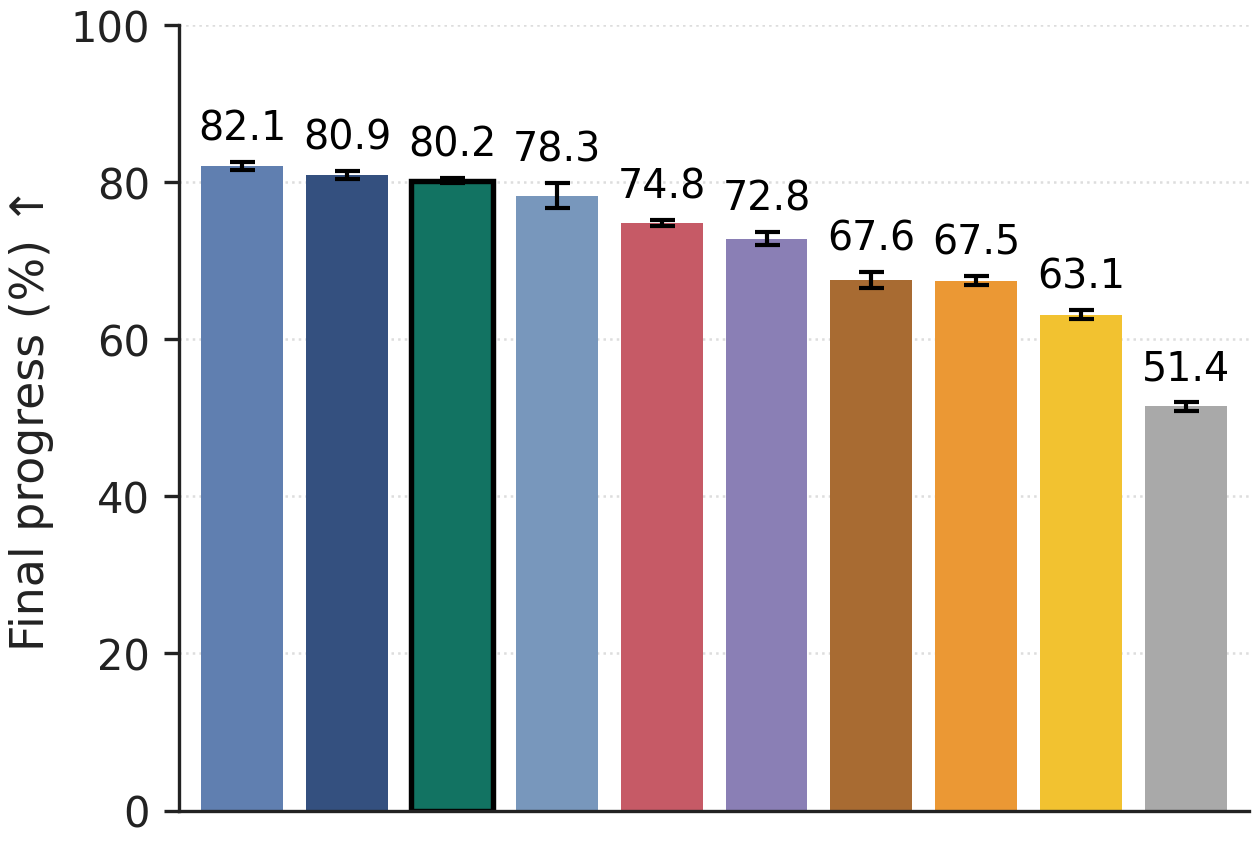}
\caption{Final progress (\%) $\uparrow$}
\label{fig:main_results:final_progress}
\end{subfigure}
\hfill
\begin{subfigure}[t]{0.24\textwidth}
\centering
\includegraphics[width=\linewidth]{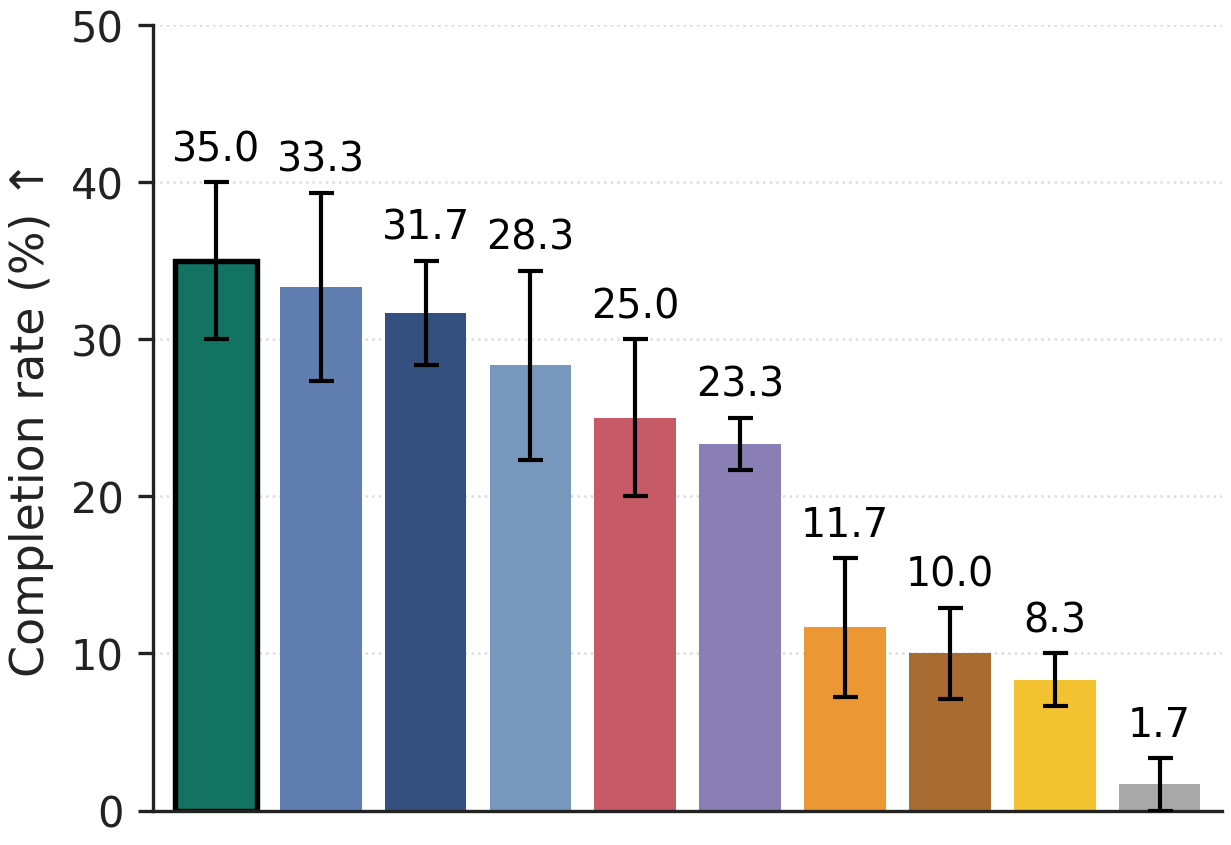}
\caption{Completion rate (\%) $\uparrow$}
\label{fig:main_results:completed}
\end{subfigure}
\hfill
\begin{subfigure}[t]{0.24\textwidth}
\centering
\includegraphics[width=\linewidth]{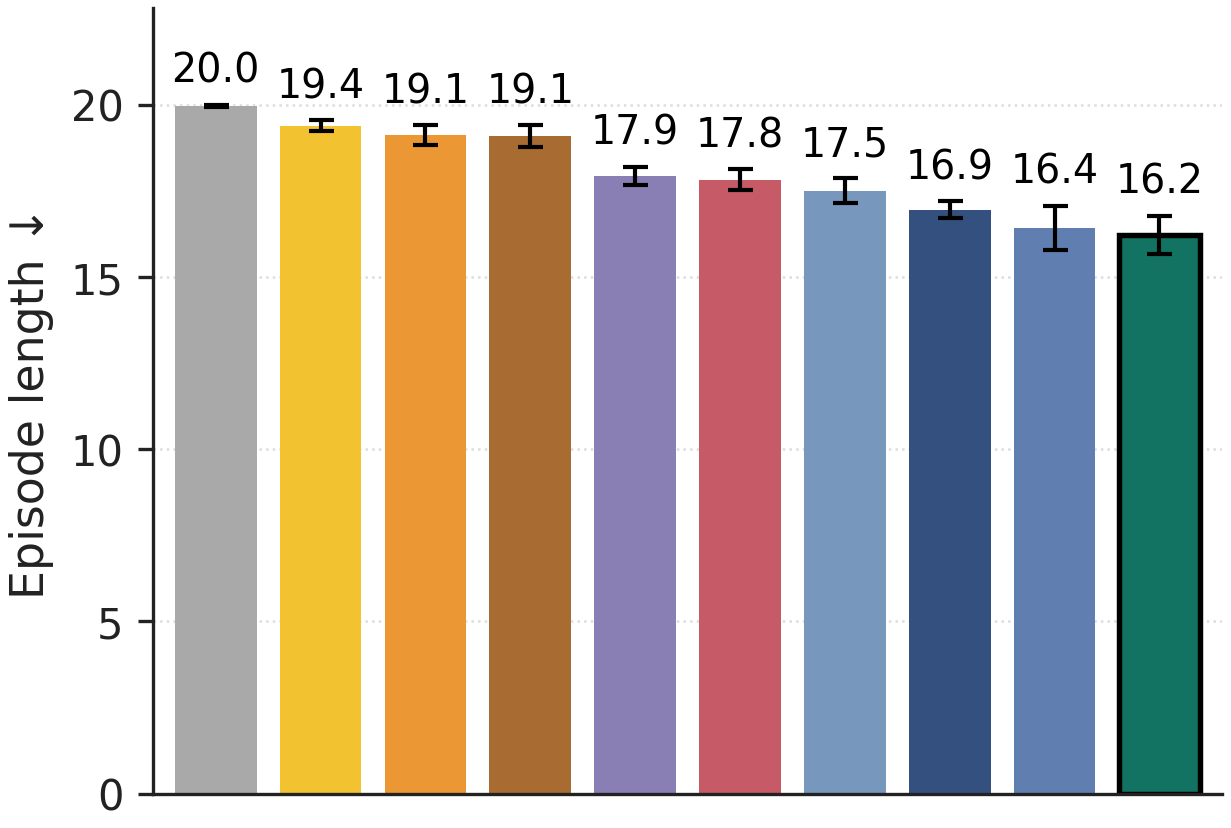}
\caption{Episode length $\downarrow$}
\label{fig:main_results:episode_length}
\end{subfigure}
\hfill
\begin{subfigure}[t]{0.24\textwidth}
\centering
\includegraphics[width=\linewidth]{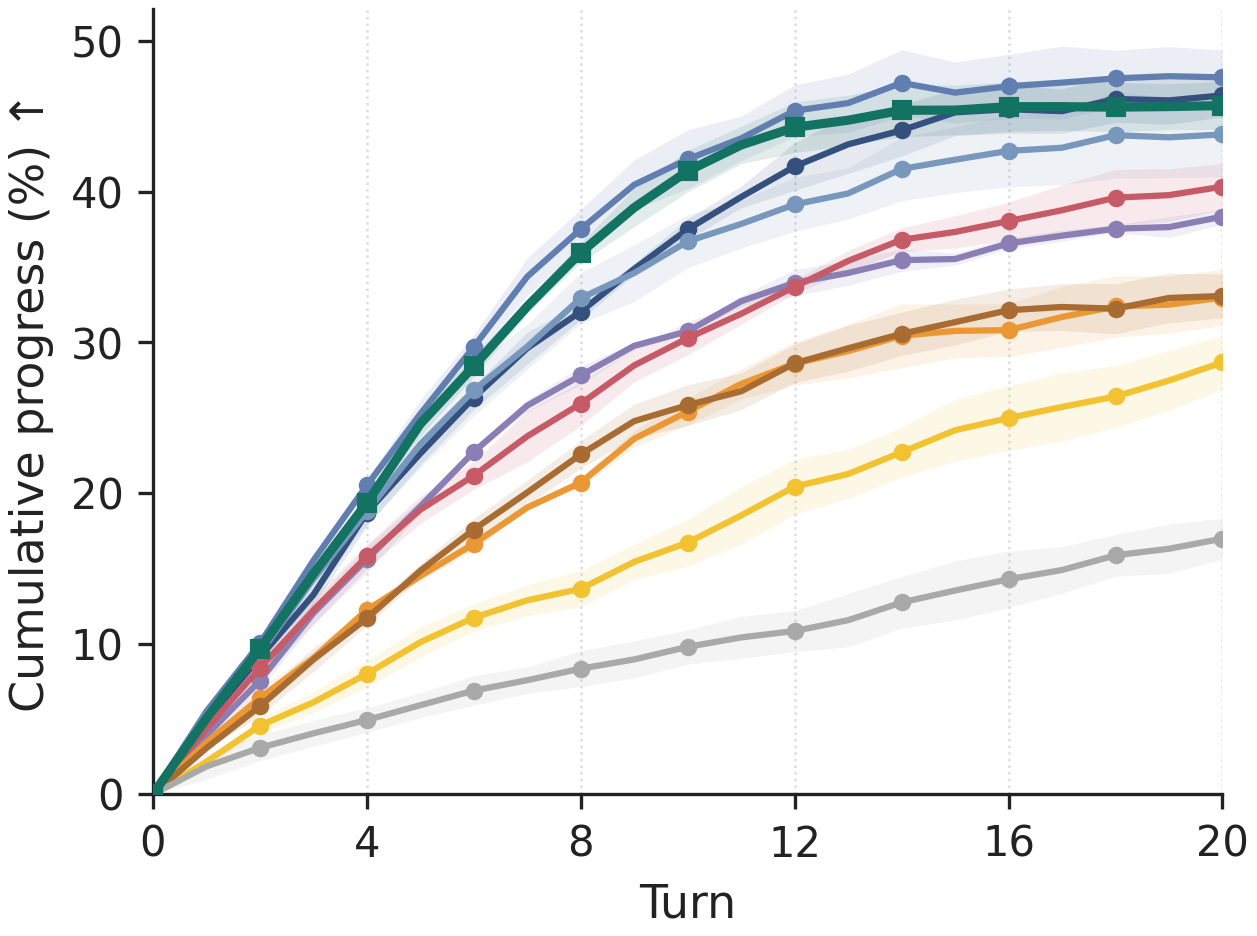}
\caption{Cumulative progress (\%) $\uparrow$}
\label{fig:main_results:progress_curve}
\end{subfigure}

\par
\includegraphics[width=0.85\textwidth]{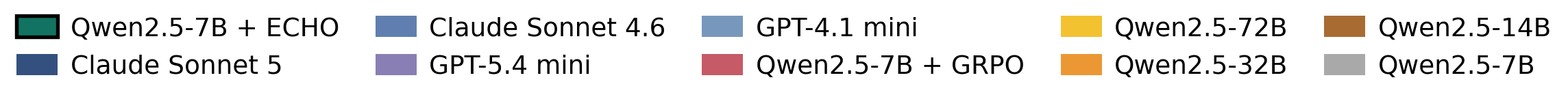}

\vspace*{-3mm} 
\caption{Results in the CRAFT environment. Bars: mean $\pm$ SEM over {\bf 3} independent runs over the evaluation set, sorted by value; Qwen2.5-7B + ECHO outlined. Curve: cumulative progress gained per turn, shaded $\pm$ SEM}
\vspace*{-3mm} 
\label{fig:main_results}
\end{figure*}

\begin{figure*}[ht!]
    \centering
    \includegraphics[width=\textwidth]{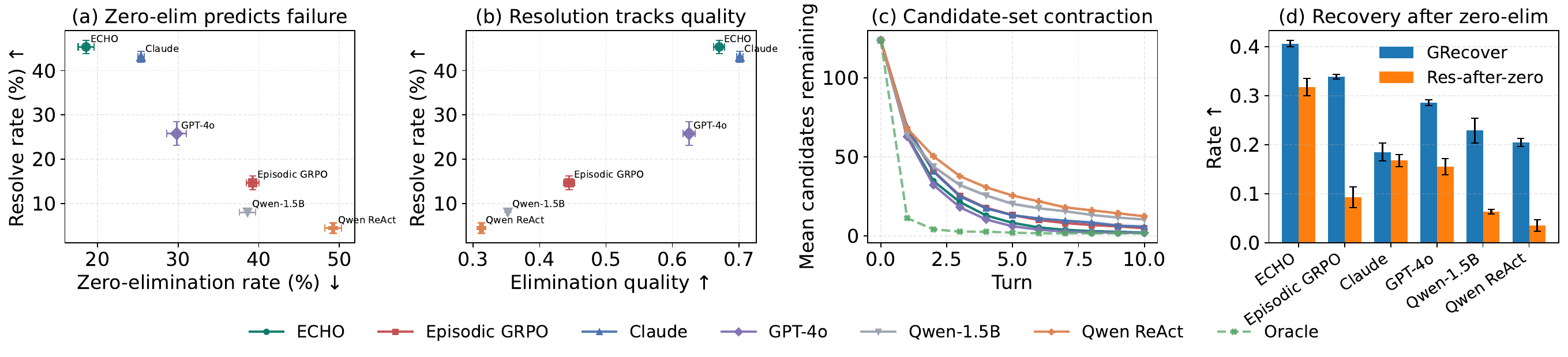}
    \vspace*{-6mm}
    \caption{\textbf{Behavioral diagnostics in CSG}. \textsc{ECHO} succeeds by reducing zero-elimination turns, maintaining high-quality belief contraction, shrinking the candidate set faster, and recovering more robustly after epistemic failures. Results on question type are shown in \Cref{tab:csg_hint_distribution_all} in Appendix~\ref{app:binary_search_echo}.}
    \vspace*{-3mm} 
    \label{fig:csg_behavior_diagnostics}
\end{figure*}

\vspace*{-2mm}
\paragraph{ECHO improves both resolution and epistemic behavior.}
Table~\ref{tab:combined_results} shows that \textsc{ECHO} achieves the strongest learned-policy performance in CSG. Its resolve rate of \(45.3\%\) exceeds the best frontier baseline, Claude Sonnet 4.6 (\(43.1\%\)), and substantially outperforms episodic-return GRPO (\(14.7\%\)) 
and all base/prompting baselines. In CRAFT, \textsc{ECHO} competes with Sonnet 4.6 \textit{and} 5 on Final Progress and beats them on Completion Rate, while globally outperforming all other baselines. More importantly, the improvement is reflected in process-level epistemic behavior: \textsc{ECHO} has the lowest CSG zero-elimination rate (\(18.6\%\)), indicating that its questions most often update the agent belief state, and lowest CRAFT invalid move rate, indicating substantially improved spatial reasoning. While Claude attains slightly higher CSG per-turn quality and groundedness, \textsc{ECHO} combines high elimination quality (\(0.670\)) and groundedness (\(0.718\)) with much stronger trajectory-level robustness: the highest grounded recovery (\(0.406\)) and resolution-after-zero rate (\(0.318\)). In CRAFT, Claude's much higher Oracle match rate but ECHO's lower mean episode length shows that ECHO succeeds not just by matching the Oracle but by actually eliciting information about high-value moves. This suggests that \textsc{ECHO}'s advantage is not merely asking strong individual questions, but that \textit{the \textsc{ECHO}-trained policy learns more robust information-gathering behavior across the trajectory in both domains: it asks well-grounded questions} and recovers from local epistemic failure. 
Additional results are reported in Appendix~\ref{app:training_dynamics} and \ref{app:binary_search_echo}. 

\vspace*{-2mm}
\paragraph{Behavioral diagnostics reveal why \textsc{ECHO} succeeds.}
To understand whether \textsc{ECHO}'s gains come from genuine belief-state adaptation,
we analyze CSG's process-level diagnostics in Fig.~\ref{fig:csg_behavior_diagnostics}. Results show that resolution is tightly linked to posterior contraction: models with fewer zero-elimination turns and higher elimination quality achieve higher solve rates, with \textsc{ECHO} occupying the favorable region of high resolution and low wasted updates. Fig.~\ref{fig:csg_behavior_diagnostics}(c) further shows that \textsc{ECHO} continues shrinking the candidate set across turns, whereas episodic-return GRPO and prompting baselines flatten earlier, indicating stalled information acquisition. Fig.~\ref{fig:csg_behavior_diagnostics}(d) shows that \textsc{ECHO} is more robust after local epistemic failures: following zero-elimination events, it produces more grounded recoveries and is substantially more likely to still resolve the episode. Together, these diagnostics support the central mechanism behind \textsc{ECHO}: it learns to ask questions that remain informative under the evolving belief state, rather than being merely plausible.


\vspace*{-2mm}
\paragraph{What epistemically adaptive behavior looks like.}
The qualitative trajectory in Fig.~\ref{fig:clue_overview} illustrates the same mechanism at the episode level. For the same secret, \textsc{ECHO} asks questions that remain useful under the changing candidate set, progressively contracting the belief state until only the target remains, while episodic-return GRPO asks questions that are initially plausible, but then non-informative, producing several zero-elimination turns. Thus, the failure mode is \emph{stale information seeking}: otherwise reasonable questions that are no longer informative after earlier observations.

\vspace*{-2mm}

\section{Conclusion}
\label{sec:conclusion}
\vspace*{-2mm}
We introduced Epistemic Decision Processes (EDPs), a belief-state formulation of multi-turn information seeking, and showed theoretically why trajectory-level success is insufficient for epistemic credit assignment: policies must receive credit at the decision points where actions change the posterior. Building on this view, we proposed \textsc{ECHO}, to train policies to be \textit{epistemically adaptive} by rewarding posterior-sensitive information gain. Across two tasks, \textsc{ECHO} trains a small (1.5B--7B) open-weight policy that competes with or beats strong proprietary frontier baselines maintaining high-quality belief contraction, and recovering more robustly after local epistemic failures. 

Our results suggest that adaptive language agents in information-seeking long-horizon tasks should be trained and evaluated not only by eventual success, but by how they gather evidence, update beliefs, and recover from uncertainty over time. Extending EDP-style credit assignment to richer tool-use and search is a natural next step for future work. 

\vspace*{-2mm}
\section*{Limitations}
 
CSG provides exact Bayesian posteriors by construction
(Lemma~\ref{lemma:exact_posterior}), but real information-seeking
tasks---medical diagnosis, web search, preference
elicitation---involve implicit, high-dimensional hypothesis spaces
where the belief state must be
approximated~\citep{choudhury2025bed,choudhury2025processrewardmodelsllm,lee2025prints}, and where information
sources may be noisy or
adversarial. 
CRAFT introduces a setting with potentially noisy information sources but one that remains well-controlled and where the Directors remain cooperative with the Builder agent. This points to the potential of ECHO in more open-ended settings but in such
settings, the posterior-sensitive turn reward that drives
\textsc{ECHO} could be challenging to compute analytically. 

Our experiments demonstrate \textsc{ECHO} at the 1.5B and 7B
scales; whether epistemic credit remains necessary at larger scales
where models may have stronger in-context belief
tracking~\citep{kossen2024incontext}, or provides complementary
gains, is an open question. Similarly, the tasks operate over a
relatively short horizon (e.g., 10 or 20 total turns); scaling epistemic credit assignment to longer horizons with relatively
open-ended actions including general-purpose tool calls and richer observation
spaces~\citep{abdulhai2023lmrl,jinsearch} like the web is unexplored.
Extending EDPs to settings with approximate beliefs, unreliable
information sources, and compositional latent structure is a
natural direction for future work.

\section*{Ethical Considerations}

Our policies are trained only against LLM simulators (GPT-4o-mini and GPT-4.1-mini) in synthetic tasks and may not transfer to, or behave appropriately with, human partners, for example by over-questioning. We also train without a KL penalty to the Instruct base model, and the released adapters have not been evaluated for safety, so they are intended for research use only.

All content in our code, trained adapters, evaluation data and episode logs is in English and used for non-commercial research, consistent with the licenses and intended use of underlying artifacts. We did not observe any identifiable information or offensive content. We will release all artifacts under MIT license for research use only.

\section*{Acknowledgments}

This material is based in part upon work supported by Other Transaction award 1AY2AX000062 from the U.S. Advanced Research Projects Agency for Health (ARPA-H) Platform Accelerating Rural Access to Distributed Integrated Medical Care (PARADIGM) program. Approved for public release, distribution unlimited. Views expressed herein do not reflect the policy or position of the Department of Health and Human Services, or the U.S. Government.



\bibliography{abhijnan_lit_survey,colm2026_conference, rlhf, rlhf_2, nips_deception, nk_cites, nk_hannah_cites}

\appendix

\section{Theory and Proofs}
\label{app:proofs}
In this section, we provide our main theoretical results and proofs. 
\begin{proposition}[Belief-state sufficiency]
\label{prop:belief-state-sufficiency}
Let \(\mathcal{E}=(\mathcal{Z},\rho_0,\mathcal{A},\Omega,O,u,T)\) be an Epistemic Decision Process. Let \(h_t=(a_0,o_0,\ldots,a_{t-1},o_{t-1})\) denote the interaction history, and let \(b_t(z)=\mathbb{P}(z^\star=z\mid h_t)\) be the posterior belief over the latent task variable. Suppose there exists a surface-state map \(s:\mathcal{H}\rightarrow\mathcal{S}\) such that the following hold:
\begin{enumerate}
    \item \textbf{Observation sufficiency:} the observation kernel depends on the history only through the latent variable, surface state, and action:
    \[
        O(o_t\mid z^\star,h_t,a_t)=O(o_t\mid z^\star,s(h_t),a_t).
    \]
    \item \textbf{Surface-state update sufficiency:} there exists a transition map \(F:\mathcal{S}\times\mathcal{A}\times\Omega\rightarrow\mathcal{S}\) such that
    \[
        s(h_{t+1})=F(s(h_t),a_t,o_t).
    \]
    \item \textbf{Utility sufficiency:} the per-turn utility can be written as a measurable function of the surface state, belief, action, and observation:
    \[
        u_t=u(s(h_t),b_t,a_t,o_t).
    \]
    The terminal reward can likewise be written as \(R_T=R_T(s(h_T),b_T)\), or equivalently as the posterior expectation of a latent terminal reward.
\end{enumerate}
Then there exists an optimal policy $\pi^\star$ for the EDP that depends on the history \(h_t\) only through the augmented epistemic state \((s(h_t),b_t,t)\).
\end{proposition}

\begin{proof}
We prove the result by reducing the EDP to a belief-state MDP~\citep{aastrom1965optimal, smallwood1973optimal}.

\paragraph{Step 1: The augmented epistemic state determines future observation distributions.}
Fix any history \(h_t\), action \(a_t\), surface state \(s_t=s(h_t)\), and belief \(b_t\). Since \(b_t(z)=\mathbb{P}(z^\star=z\mid h_t)\), the predictive distribution over the next observation is
\[
    \mathbb{P}(o_t\mid h_t,a_t)
    =
    \int_{\mathcal{Z}}
    O(o_t\mid z,s_t,a_t)b_t(z)\,dz.
\]
Thus, by the observation sufficiency assumption, the distribution of \(o_t\) depends on the full history \(h_t\) only through \((s_t,b_t)\) and the chosen action \(a_t\).

\paragraph{Step 2: The belief update is determined by the augmented epistemic state.}
After observing \(o_t\), the posterior is updated by Bayes' rule:
\begin{align}
    b_{t+1}(z)
    &\notag=
    \tau(b_t,a_t,o_t)(z)\\
    &\notag=
    \frac{O(o_t\mid z,s_t,a_t)b_t(z)}
    {\int_{\mathcal{Z}}O(o_t\mid z',s_t,a_t)b_t(z')\,dz'}.
\end{align}
Therefore, \(b_{t+1}\) is fully determined by \((s_t,b_t,a_t,o_t)\). By the surface-state update assumption, \(s_{t+1}=s(h_{t+1})=F(s_t,a_t,o_t)\). Hence the next augmented epistemic state \((s_{t+1},b_{t+1},t+1)\) depends on the history only through \((s_t,b_t,t)\), the action \(a_t\), and the observation \(o_t\).

\paragraph{Step 3: The reward is determined by the augmented epistemic state.}
By the utility sufficiency assumption, the per-turn utility is \(u(s_t,b_t,a_t,o_t)\). Thus the distribution of both the next augmented state and the immediate utility is determined by \((s_t,b_t,t)\) and \(a_t\). No additional information from the raw history \(h_t\) is needed for predicting future rewards.

\paragraph{Step 4: Bellman recursion on the belief-state MDP.}
Define the optimal value function on the augmented epistemic state by backward induction~\citep{sutton2018reinforcement}. At the terminal step,
\[
    V_T^\star(s,b)=R_T(s,b).
\]
For \(t<T\), define
\begin{align}
    V_t^\star(s,b)
    =&\notag
    \max_{a\in\mathcal{A}}
    \mathbb{E}_{z\sim b}
    \mathbb{E}_{o\sim O(\cdot\mid z,s,a)}
    \Big[\\
        &\notag u(s,b,a,o)
        +
        V_{t+1}^\star(F(s,a,o),\\
        &\tau(b,a,o))
    \Big].
\end{align}
This Bellman recursion is well-defined because the predictive observation distribution, the belief update, the surface-state update, and the utility all depend only on \((s,b,t)\) and \(a\).

\paragraph{Step 5: Constructing an optimal policy.}
For any history \(h_t\), let \(s_t=s(h_t)\) and \(b_t=\mathbb{P}(z^\star\mid h_t)\). Define a greedy policy
\begin{align}
    \pi^\star(h_t)
    \in&\notag
    \arg\max_{a\in\mathcal{A}}
    \mathbb{E}_{z\sim b_t}
    \mathbb{E}_{o\sim O(\cdot\mid z,s_t,a)}
    \Big[\\
        &\notag u(s_t,b_t,a,o)
        +
        V_{t+1}^\star(F(s_t,a,o),\\
        &\tau(b_t,a,o))
    \Big].
\end{align}
By construction, \(\pi^\star(h_t)\) depends on \(h_t\) only through \((s(h_t),b_t,t)\). Moreover, by the Bellman recursion~\citep{sutton2018reinforcement}, this policy attains \(V_t^\star(s(h_t),b_t)\) for every history \(h_t\). Therefore, \(\pi^\star\) is optimal, and an optimal policy exists that is Markovian~\citep{zhang2025markovianreflectiveexplorationbayesadaptive} in the augmented epistemic state \((s(h_t),b_t,t)\).
\end{proof}

\paragraph{Significance.}
Proposition~\ref{prop:belief-state-sufficiency} justifies the EDP formulation. Although the agent observes a growing text history \(h_t\), the full history is not the fundamental decision state. What matters for optimal action selection is the visible surface state together with the posterior belief induced by the history. This separates ordinary text-state conditioning from epistemic adaptivity: a policy that ignores \(b_t\) may condition on the prompt and turn index, but it cannot in general choose actions based on what the agent has learned from external observations. The result is the EDP analogue of the standard belief-MDP reduction in Bayesian RL and POMDPs, but applied to multi-turn LLM agents whose observations come from an external environment rather than from internal chain-of-thought continuation.

\paragraph{Connection to Bayesian RL for LLM reasoning.}
This proposition parallels the belief-sufficiency argument used in Bayesian RL accounts of LLM reasoning. In BARL~\citep{zhang2025markovianreflectiveexplorationbayesadaptive}, the hidden uncertainty is over candidate MDPs or reasoning hypotheses, and the belief is updated over those hypotheses inside a generated reasoning trace. In our EDP setting, the hidden uncertainty is over an external latent task variable \(z^\star\), and the belief \(b_t(z)=\mathbb{P}(z^\star=z\mid h_t)\) is updated through external observations. Thus, the proof technique is the same belief-MDP reduction, but our setting is explicitly multi-turn: the agent acts, receives evidence from the environment, updates its posterior, and then chooses the next action.

\begin{theorem}[Adaptive policies can strictly outperform belief-agnostic policies]
\label{thm:adaptive-gap}
For every depth $d \geq 1$, there exists an Epistemic Decision Process $\mathcal{E}_d$ 
with horizon $d$ such that:
\begin{enumerate}
    \item there exists an epistemically adaptive policy $\pi^\star$ with 
          $\mathbb{P}(\mathrm{success} \mid \pi^\star) = 1$;
    \item every belief-agnostic policy $\mu$ satisfies 
          $\mathbb{P}(\mathrm{success} \mid \mu) \leq 2^{1-d}$.
\end{enumerate}
\end{theorem}

\begin{proof}
We construct a family of EDPs in which success requires using evidence 
revealed in previous turns. The construction separates policies that 
condition on the epistemic state from policies that do not.

\paragraph{Step 1: Construction of $\mathcal{E}_d$.}
Let the latent task variable be a binary string 
$z^\star = (z^\star_1, \ldots, z^\star_d) \in \mathcal{Z} = \{0,1\}^d$, 
sampled uniformly from $\rho_0$. The action space is $\mathcal{A} = \{0,1\}$, 
and the observation space is $\Omega = \{0,1,\bot\}$. The surface state 
is only the turn index, $s(h_t) = t$, so a belief-agnostic policy may 
know the depth but not the posterior induced by observed evidence.

At turn $t = 0$, action $a_0 = 0$ always reveals the first latent bit: 
$o_0 = z^\star_1$. Action $a_0 = 1$ returns $o_0 = \bot$.
For turns $t \geq 1$, the environment reveals the \emph{next} latent bit 
if and only if the agent chooses the \emph{previously revealed} bit:
\[
    o_t =
    \begin{cases}
        z^\star_{t+1}, & \text{if } a_t = z^\star_t, \\
        \bot,          & \text{otherwise.}
    \end{cases}
\]
The terminal reward is $R_T = 1$ if the agent has observed all $d$ latent 
bits by the end of the episode, and $R_T = 0$ otherwise. Since at most one 
bit is revealed per turn, any $\bot$ observation prevents success.

The key design principle is that the correct action at turn $t$ is 
determined by the observation $o_{t-1} = z^\star_t$ received at the 
previous turn. This observation is part of the history $h_t$ and therefore 
encoded in the belief $b_t$, but it is not a function of the surface state 
$s(h_t) = t$ alone.

\paragraph{Step 2: An epistemically adaptive policy succeeds with 
probability one.}
Define $\pi^\star$ as follows. At $t = 0$, choose $a_0 = 0$, which reveals 
$z^\star_1$. At each later turn $t \geq 1$, choose $a_t = z^\star_t$, 
the bit revealed as observation $o_{t-1}$ at the previous turn.

This action is accessible to $\pi^\star$: the value $z^\star_t = o_{t-1}$ 
is contained in the history $h_t = (a_0, o_0, \ldots, a_{t-1}, o_{t-1})$ 
and hence in the posterior belief $b_t$. An epistemically adaptive policy 
can therefore compute $a_t = z^\star_t$ from $b_t$.

By induction, after turn $t$, the agent has observed 
$z^\star_1, \ldots, z^\star_{t+1}$. After $d - 1$ turns of this form, 
the agent has observed the full latent string $z^\star$, so 
$\mathbb{P}(\mathrm{success} \mid \pi^\star) = 1$.

\paragraph{Step 3: Belief-agnostic policies cannot use the revealed bit.}
Let $\mu$ be any belief-agnostic policy. Since $s(h_t) = t$, the policy 
conditions only on the turn index: 
$\mu(a_t \mid h_t) = \mu_t(a_t)$ for some distribution 
$\mu_t \in \Delta(\{0,1\})$, independent of the observations received.

At $t = 0$, the best belief-agnostic policy chooses $a_0 = 0$, so we 
upper bound its probability of passing the first turn by $1$.

For any later turn $t \geq 1$, success requires $a_t = z^\star_t$. 
The value $z^\star_t$ was revealed as observation $o_{t-1}$, but $\mu$ 
does not condition on $o_{t-1}$. Conditional on success at all previous 
turns---which is required for $z^\star_t$ to have been revealed---the 
revealed bit $z^\star_t$ is uniformly distributed on $\{0,1\}$ by symmetry 
of the prior. Therefore, for any fixed action distribution $\mu_t$:
\begin{align}
    &\notag\mathbb{P}(a_t = z^\star_t \mid \mathrm{success\ up\ to\ } t-1,\, \mu)
    \\=&
    \sum_{a \in \{0,1\}} \mu_t(a)\,\mathbb{P}(z^\star_t = a)
    =
    \frac{\mu_t(0) + \mu_t(1)}{2}
    = \frac{1}{2}.
\end{align}
The equality holds regardless of the mixing weights $\mu_t(0), \mu_t(1)$ 
because the bit $z^\star_t$ is symmetric.

\paragraph{Step 4: Chaining the per-turn bounds.}
For the terminal reward $R_T = 1$, the agent must choose the correct 
action at every turn $t = 1, \ldots, d-1$ (turn $t = 0$ succeeds 
with probability $1$ under the best belief-agnostic policy). The 
events $\{a_t = z^\star_t\}$ are conditionally independent across turns 
given the latent string $z^\star$, because $\mu_t$ is a fixed distribution 
and $z^\star_1, \ldots, z^\star_d$ are independent bits. Therefore:
\[
    \mathbb{P}(\mathrm{success} \mid \mu)
    \leq
    1 \cdot \prod_{t=1}^{d-1} \frac{1}{2}
    =
    2^{1-d}.
\]
Combined with Part~1, this gives $\mathbb{P}(\mathrm{success} \mid \pi^\star) 
= 1$ versus $\mathbb{P}(\mathrm{success} \mid \mu) \leq 2^{1-d}$ for any 
belief-agnostic $\mu$, establishing the claimed separation.

\paragraph{Remark.}
This construction is related to the binary tree construction in 
\citep{zhang2025markovianreflectiveexplorationbayesadaptive} (Theorem~4.3), 
where an epistemically adaptive policy succeeds with probability~$1$ in a 
depth-$d$ tree while any Markovian policy succeeds with probability 
$1/2^{d-1}$. The key difference is the proof technique: their argument 
uses a visit-frequency induction on tree nodes under Markovian policies, 
while our argument uses the direct symmetry of the per-turn bit under 
belief-agnostic action selection. Both constructions achieve an 
exponential gap of $2^{\Omega(d)}$, but ours is more direct for 
the EDP setting because it does not require the full tree structure—only a sequential bit-reveal chain.
\end{proof}



\begin{theorem}[Spectral aliasing of trajectory-level credit]
\label{thm:spectral-aliasing}
Fix a policy \(\pi\) in an EDP and a turn \(t\). Let \(\omega\) denote a full sampled trajectory, let \(G(\omega)=\sum_{s=0}^{T-1}r_s(\omega)\) be its total return, and let
\[
    A(\omega)
    =
    A_{\mathrm{B}}^\pi(h_t(\omega),a_t(\omega))
\]
be the Bayesian advantage of the action taken at turn \(t\). Consider the class of trajectory-level estimators \(\widehat A_t^{\mathrm{traj}}=\phi(G)\), where \(\phi\) is any measurable function of the total return.

Let \(\mathcal{H}=L^2(\Omega,\mathbb{P}_\pi)\) and let \(\mathcal{H}_G=\{\phi(G):\phi\in L^2\}\) be the closed subspace of return-measurable functions. Let \(\Pi_G:\mathcal{H}\rightarrow \mathcal{H}_G\) be the orthogonal projection onto this subspace, i.e. \(\Pi_G A=\mathbb{E}[A\mid G]\). Then
\begin{align}
    \inf_{\phi}
    \mathbb{E}\left[
        \left(
        \phi(G)-A
        \right)^2
    \right]
    &\notag=
    \left\|
        (I-\Pi_G)A
    \right\|_2^2\\
    &\notag=
    \mathbb{E}\left[
        \operatorname{Var}(A\mid G)
    \right].
\end{align}
Consequently, if \(A_{\mathrm{B}}^\pi(h_t,a_t)\) is not measurable with respect to the total return \(G\), then any trajectory-level estimator has strictly positive irreducible statewise error.
\end{theorem}

\begin{proof}
We view advantage estimation as an approximation problem in the Hilbert space \(\mathcal{H}=L^2(\Omega,\mathbb{P}_\pi)\), equipped with inner product \(\langle f,g\rangle=\mathbb{E}[fg]\). The class of trajectory-level estimators \(\phi(G)\) is exactly the subspace \(\mathcal{H}_G\) of functions measurable with respect to the sigma-algebra generated by \(G\).

The conditional expectation \(\Pi_G A=\mathbb{E}[A\mid G]\) is the orthogonal projection of \(A\) onto \(\mathcal{H}_G\). Therefore, for any \(\phi(G)\in \mathcal{H}_G\), we have the orthogonal decomposition
\[
    \phi(G)-A
    =
    \underbrace{\phi(G)-\Pi_G A}_{\in \mathcal{H}_G}
    -
    \underbrace{(I-\Pi_G)A}_{\perp \mathcal{H}_G}.
\]
By Pythagoras,
\[
    \|\phi(G)-A\|_2^2
    =
    \|\phi(G)-\Pi_G A\|_2^2
    +
    \|(I-\Pi_G)A\|_2^2.
\]
The first term is minimized by choosing \(\phi^\star(G)=\Pi_G A=\mathbb{E}[A\mid G]\), leaving the irreducible residual \(\|(I-\Pi_G)A\|_2^2\). Since conditional expectation gives the best \(L^2\) predictor, this residual is
\begin{align}
    \|(I-\Pi_G)A\|_2^2
    &\notag=
    \mathbb{E}\left[
        \left(A-\mathbb{E}[A\mid G]\right)^2
    \right]\\
    &\notag=
    \mathbb{E}\left[
        \operatorname{Var}(A\mid G)
    \right].
\end{align}
If \(A\) is not measurable with respect to \(G\), then \(A\neq \Pi_GA\) on a set of positive probability, so \(\|(I-\Pi_G)A\|_2^2>0\). This proves the claim.
\end{proof}

\begin{corollary}[Finite-dimensional view of trajectory-level aliasing]
\label{cor:finite-spectral-aliasing}
Assume the trajectory space is finite, \(\Omega=\{\omega_1,\ldots,\omega_n\}\), with probabilities \(p_i=\mathbb{P}_\pi(\omega_i)>0\). Let \(A_i=A(\omega_i)\), where \(A(\omega)=A_{\mathrm{B}}^\pi(h_t(\omega),a_t(\omega))\), and let \(G_i=G(\omega_i)\) be the total return of trajectory \(\omega_i\). Let \(\Pi_G\) denote the orthogonal projection onto the subspace of functions measurable with respect to \(G\), under the weighted inner product \(\langle f,g\rangle=\sum_{i=1}^n p_i f_i g_i\). Then
\[
    (\Pi_G A)_i
    =
    \mathbb{E}[A\mid G=G_i],
\]
i.e., \(\Pi_G\) replaces each trajectory-level advantage by the average Bayesian advantage among trajectories with the same total return. Moreover, \(\Pi_G\) is self-adjoint and idempotent, so its spectrum is contained in \(\{0,1\}\). The eigenspace with eigenvalue \(1\) consists of functions that are constant on each return level set, while the eigenspace with eigenvalue \(0\) consists of within-level-set contrasts whose conditional mean given \(G\) is zero.

Consequently, the irreducible error of any trajectory-level estimator \(\phi(G)\) is exactly the squared norm of the component of \(A\) lying outside the return-measurable subspace:
\begin{align}
    \inf_{\phi}
    \|\phi(G)-A\|_2^2
    &\notag=
    \|(I-\Pi_G)A\|_2^2\\
    &\notag=
    \sum_g
    \mathbb{P}(G=g)\operatorname{Var}(A\mid G=g).
\end{align}
Thus, if two positive-probability trajectories have the same total return but require different Bayesian credit at turn \(t\), then trajectory-level estimators cannot recover that within-return difference.
\end{corollary}

\begin{proof}
Since \(\Pi_G\) is the conditional expectation operator \(\mathbb{E}[\cdot\mid G]\), it maps each value \(A_i\) to the average of \(A\) over all trajectories with the same total return \(G_i\). Thus \((\Pi_G A)_i=\mathbb{E}[A\mid G=G_i]\).

The operator \(\Pi_G\) is an orthogonal projection in \(L^2(\Omega,\mathbb{P}_\pi)\), so \(\Pi_G^\ast=\Pi_G\) and \(\Pi_G^2=\Pi_G\). Therefore every eigenvalue \(\lambda\) satisfies \(\lambda^2=\lambda\), implying \(\lambda\in\{0,1\}\). The \(1\)-eigenspace is the range of \(\Pi_G\), namely the set of functions measurable with respect to \(G\), which are exactly the functions that are constant on each return level set. The \(0\)-eigenspace is the nullspace of \(\Pi_G\), namely functions \(f\) satisfying \(\mathbb{E}[f\mid G]=0\), which are precisely within-level-set contrasts.

By Theorem~\ref{thm:spectral-aliasing}, the best trajectory-level estimator is \(\Pi_G A=\mathbb{E}[A\mid G]\), and its residual is \((I-\Pi_G)A\). Therefore
\[
    \inf_{\phi}\|\phi(G)-A\|_2^2
    =
    \|(I-\Pi_G)A\|_2^2
    =
    \mathbb{E}[\operatorname{Var}(A\mid G)].
\]
In the finite case, this conditional variance expands as
\[
    \mathbb{E}[\operatorname{Var}(A\mid G)]
    =
    \sum_g
    \mathbb{P}(G=g)\operatorname{Var}(A\mid G=g).
\]
If \(A\) varies among trajectories in the same return level set, then at least one conditional variance term is positive, so the residual is strictly positive.
\end{proof}

\textbf{Remark.}~\textsc{ECHO} retains the group-relative clipped 
trust-region form of standard policy gradient methods~\citep{schulman2015trust,sutton2018reinforcement,williams92}, but replaces 
trajectory-level success with turn-level epistemic progress as the 
credit signal.

\paragraph{Why this matters for epistemic credit}
Corollary~\ref{cor:finite-spectral-aliasing} makes the limitation of trajectory-level credit explicit. Total return partitions trajectories into return level sets, and any estimator \(\phi(G)\) can only assign credit that is constant within each level set. However, in an EDP, two trajectories can have the same final return while requiring different credit for an intermediate action because the action was taken under a different belief state. This within-return variation is exactly the component \((I-\Pi_G)A\), which lies outside the return-measurable subspace and is therefore invisible to any trajectory-level estimator. The failure is not finite-sample noise or poor optimization; it is an aliasing\footnote{We use the term ``aliasing'' by analogy to signal processing, where distinct high-frequency components can become indistinguishable under a coarse sampling operator~\citep{oppenheim1999discrete,shannon2006communication}. In our setting, the coarse observation is the aggregate return \(G\): trajectories with different belief-conditioned advantages can collapse to the same return level set.} problem caused by projecting belief-dependent credit onto a coarser signal.

\paragraph{Connection to policy-gradient baselines}
This result is complementary to classical policy-gradient variance-reduction analyses such as Greensmith et al.~\citep{greensmith2004variance}. Those analyses study how baselines reduce variance when they condition on useful state information. Our result identifies the opposite obstruction in EDPs: if the estimator conditions only on aggregate return, then the projection subspace is too small to recover belief-dependent advantage components. In other words, the issue is not merely choosing a better scalar baseline; the estimator must expose information about the epistemic state.

\begin{lemma}[Exact posterior in CSG]
\label{lemma:exact_posterior}
Let $z^\star$ be drawn uniformly from $\{1,\ldots,N\}$ and let
$\mathcal{C}_0 = \{1,\ldots,N\}$. Suppose the oracle responds
truthfully at every turn, i.e.,
$p(o_t \mid z^\star\!=\!z,\, a_t) = \mathbf{1}[z \in \mathcal{C}_{t+1}]$,
where $\mathcal{C}_{t+1} \subseteq \mathcal{C}_t$ is the set of
candidates consistent with observation $o_t$. Then for all
$t \geq 0$, the posterior is exactly uniform over the remaining
candidate set:
\[
  b_t(z) \;=\; \mathbb{P}(z^\star = z \mid h_t)
  \;=\;
  \frac{\mathbf{1}[z \in \mathcal{C}_t]}{|\mathcal{C}_t|}.
\]
\end{lemma}

\begin{proof} We prove this by induction. Notice that the base case holds by construction:
$b_0(z) = 1/N = 1/|\mathcal{C}_0|$ for all $z$.

For the inductive step, assume
$b_t(z) = \mathbf{1}[z \in \mathcal{C}_t] / |\mathcal{C}_t|$.
After action $a_t$ and observation $o_t$, Bayes' rule gives
\[
  b_{t+1}(z)
  = \frac{b_t(z)\; p(o_t \mid z,\, a_t)}
         {\sum_{z'} b_t(z')\; p(o_t \mid z',\, a_t)}.
\]
Substituting the inductive hypothesis and the binary likelihood:
\begin{align}
  b_{t+1}(z)
  &\notag= \frac{
      \frac{\mathbf{1}[z \in \mathcal{C}_t]}{|\mathcal{C}_t|}
      \cdot \mathbf{1}[z \in \mathcal{C}_{t+1}]
    }{
      \sum_{z' \in \mathcal{C}_t}
      \frac{1}{|\mathcal{C}_t|}
      \cdot \mathbf{1}[z' \in \mathcal{C}_{t+1}]
    }\\
  &\notag= \frac{\mathbf{1}[z \in \mathcal{C}_{t+1}]}
         {|\mathcal{C}_{t+1}|},
\end{align}
where the numerator uses
$\mathcal{C}_{t+1} \subseteq \mathcal{C}_t$, so
$\mathbf{1}[z \in \mathcal{C}_t] \cdot
 \mathbf{1}[z \in \mathcal{C}_{t+1}]
 = \mathbf{1}[z \in \mathcal{C}_{t+1}]$,
and the denominator sums to
$|\mathcal{C}_{t+1}|/|\mathcal{C}_t|$, which cancels
$|\mathcal{C}_t|$.
\end{proof}

\paragraph{Remark.}
The oracle in CSG is implemented via GPT-4o mini, so the observation
kernel $O(o_t \mid z^\star, h_t, a_t)$ is stochastic\footnote{\href{https://community.openai.com/t/why-does-openai-api-behave-randomly-with-temperature-0-and-top-p-0/934104}{https://community.openai.com/t/why-does-openai-api-behave-randomly-with-temperature-0-and-top-p-0/934104}} rather than
deterministic—even at temperature set to zero. However, the candidate-set update removes
inconsistent candidates given the \emph{realized} observation, so
$b_t$ is exact conditional on $o_t$; the stochasticity enters
through the EDP environment dynamics
(\Cref{sec:edp_theory}), not through belief-approximation error.
This makes the turn reward
$r_t = (|\mathcal{C}_t| - |\mathcal{C}_{t+1}|)/|\mathcal{C}_t|$
a direct measure of posterior contraction, computable in closed
form without learned reward models or Monte Carlo EIG
estimation~\citep{choudhury2025bed}.
 
\paragraph{Connection to BAMDPs and epistemic planning.}
    EDPs are best viewed as a task-level specialization of belief-state decision processes, rather than as a representationally separate alternative to POMDPs or BAMDPs. The distinction is one of modeling focus, but one that leverages the expressivity of LLMs that naturally  operate in exponentially complex language-based state-action spaces. BAMDPs typically maintain beliefs over unknown transition or reward parameters, or over latent MDP identities that are computationally expensive to track with LLMs, and use interaction to improve control under environment uncertainty. Epistemic planning methods such as Reflect-then-Plan similarly reason over model uncertainty to improve offline-to-online planning~\citep{jeong2025reflectthenplan}. In contrast, EDPs target information-seeking language agents operating through a fixed interaction interface: the environment need not be unknown in its dynamics, but the current task variable \(z^\star\) is hidden (see \Cref{ssec:epistemically_adaptive_policies}). The relevant belief is therefore a posterior over task hypotheses, \(p(z^\star \mid h_t)\), and the central question is whether an action acquires evidence that changes this posterior. Prior work~\citep{zhang2025markovianreflectiveexplorationbayesadaptive} is closest in spirit, using Bayes-adaptive RL to induce uncertainty-adaptive reflection over reasoning chains or MDP hypotheses~\citep{zhang2025markovianreflectiveexplorationbayesadaptive}; EDPs instead focus on external\footnote{In contrast to \citet{zhang2025markovianreflectiveexplorationbayesadaptive}, where adaptivity emerges “in-context” aka in reasoning chains within a turn, our focus in on adaptivity in an external sense i.e., across interaction turns.} query actions and observable belief updates \textit{across} interaction turns. Thus, the contribution of EDPs is not a new belief-state formalism in the abstract, but a lens that makes the credit signal for information-seeking actions explicit.

\section{Task Explicit EDP Interpretations}
\label{app:edp-interp}
\subsection{Clue Selector Game (CSG)}
CSG's latent task variable is the secret number \(z^\star=n^\star\). The history \(h_t\) consists of the previous arm selections, questions, and oracle observations. The belief state \(b_t\) is the posterior over the secret number induced by this history. Because the secret is sampled uniformly and oracle observations only eliminate inconsistent candidates, the posterior is uniform over the remaining candidate set:
\[
    b_t(n)=
    \begin{cases}
    1/|\mathcal{C}_t|, & n\in\mathcal{C}_t,\\
    0, & n\notin\mathcal{C}_t.
    \end{cases}
\]
For example, if the initial candidates are \(\{1,\ldots,10\}\) and the agent asks whether the number is even, a truthful ``yes'' response updates the candidate set to \(\{2,4,6,8,10\}\), making the posterior uniform over those five numbers. Thus, in CSG, the candidate set is not merely a textual state variable; it is a direct representation of the agent's epistemic state.

Turn-level utility measures epistemic progress. The dominant reward term is fractional candidate elimination, \(r_{\mathrm{info}}=(|\mathcal{C}_t|-|\mathcal{C}_{t+1}|)/|\mathcal{C}_t|\), which measures how much the action contracts the current posterior. Training also includes small shaping terms for valid arm--question matching and query diversity, a redundancy penalty for repetition, and sparse bonuses for resolving the task efficiently. \an{The action value is therefore belief-dependent: two histories may share the same belief-independent task context---for example, the same turn index, action schema, and remaining budget---while inducing different posteriors over \(z^\star\), and therefore requiring different useful actions. }
The posterior-sensitive reward signal is used by \textsc{ECHO} and other RL baselines for training; it assigns credit based on the epistemic progress made at the turn where the action was chosen, rather than broadcasting only the final episode outcome. CSG is therefore a clean testbed for our theory. 

\subsection{Distributed Partial Information Puzzle (DPIP)}

In CRAFT, the latent task variable is the goal structure $z^\star=S^*$. History $h_t$ consists of the dialogue history between Directors and Builder and Builder moves or clarification questions. The belief state $b_t$ is the distribution over possible target goal structures at time $t$---in other words everything the Builder cannot rule out about the goal structure. As in CSG, the candidate set of possible remaining goal structures is a direct representation of the Builder's epistemic state, reified through actions on the board. The dominant reward term---change in task progress for an executed candidate move---is a direct analogue to CSG's fractional candidate elimination term, as greater change in progress directly indexes how many possible remaining moves can be made while maintaining monotonic progress toward the goal. Importantly, due to {\tt PLACE} and {\tt REMOVE} moves, progress in CRAFT is not strictly monotonic, which is an important source of its complexity.

\section{ECHO Training Dynamics}
\label{app:training_dynamics}

\begin{figure*}[t]
    \centering

    \begin{subfigure}[t]{0.235\linewidth}
        \centering
        \includegraphics[width=\linewidth]{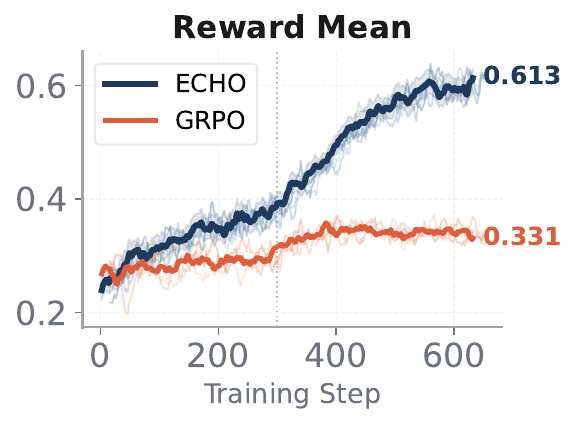}
      
        \label{fig:csg_reward}
    \end{subfigure}
    \hfill
    \begin{subfigure}[t]{0.235\linewidth}
        \centering
        \includegraphics[width=\linewidth]{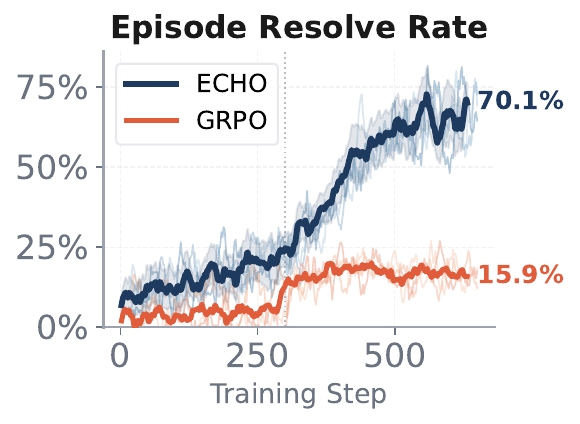}
 
        \label{fig:csg_episode_resolve_rate}
    \end{subfigure}
    \hfill
    \begin{subfigure}[t]{0.235\linewidth}
        \centering
        \includegraphics[width=\linewidth]{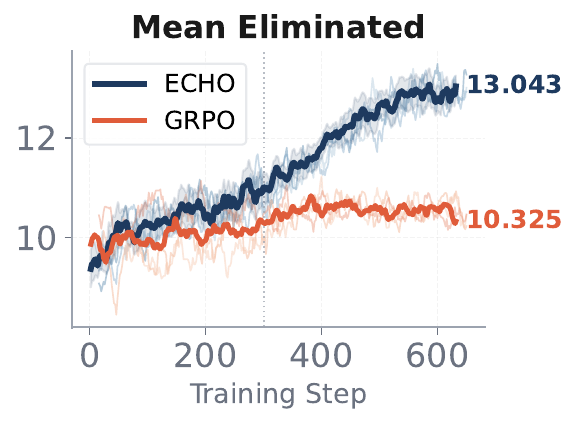}
 
        \label{fig:csg_redundancy_rate}
    \end{subfigure}
    \hfill
    \begin{subfigure}[t]{0.235\linewidth}
        \centering
        \includegraphics[width=\linewidth]{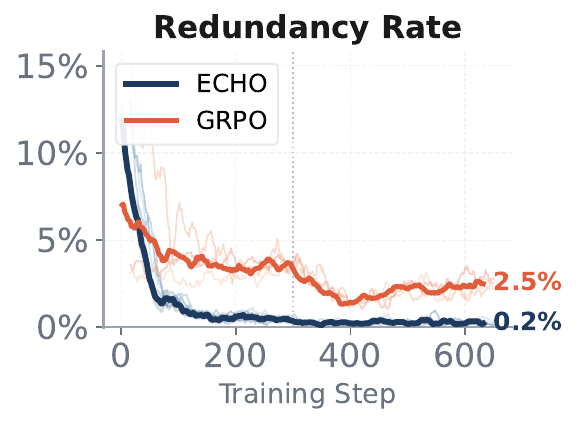}
 
        \label{fig:csg_arm_match_rate}
    \end{subfigure}

    \vspace*{-2mm} 
    \caption{Performance and behavioral trends during \textbf{RL training} in the Clue Selector Game (CSG) environment. Plots show metrics SEM across 3 independent training runs. Additional training dynamics are shown in Fig.~\ref{fig:clue_steps_additional_metrics} in Appendix~\ref{app:training_dynamics}.}
    \vspace*{-2mm} 
    \label{fig:clue_steps_four_panel_training_main}
\end{figure*}

\vspace*{-2mm}
\paragraph{Explicit reasoning does not substitute for epistemic training.}
Explicit reasoning traces alone do not produce epistemic adaptability in long-horizon EDPs. Full ReAct forces models to reason on nearly every turn, yet GPT-4o + Full ReAct resolves only \(25.3\%\) of episodes. Qwen2.5-1.5B + Full ReAct resolves only \(4.4\%\). By contrast, \textsc{ECHO} resolves \(45.3\%\) of episodes and produces explicit reasoning text on only \(0.9\%\) of turns. This suggests that the central bottleneck for epistemic adaptivity is \textit{not} verbalizing intermediate reasoning, but learning a policy conditioned on the evolving belief state. This is “silent exploration” (Sec.~\ref{sec:edp_theory}) manifest: adaptive information seeking expressed through query choices and posterior-sensitive feedback rather than long test-time reasoning traces. \textsc{ECHO} trains adaptable models using belief-sensitive reward structures and provides strong performance with higher token cost-efficiency—a trend consistent with prior works~\citep{ma2025reasoningmodelseffectivethinking} in LLM reasoning.

\vspace*{-2mm}
\paragraph{Training dynamics support turn-level epistemic credit assignment.}
Fig.~\ref{fig:clue_steps_four_panel_training_main} 
shows that \textsc{ECHO}'s gains emerge as a turn-level epistemic learning effect rather than as a generic reward increase. By the end of training, \textsc{ECHO} reaches much higher episode resolution than trajectory-level GRPO (\(70.1\%\) vs.\ \(15.9\%\)), eliminates more candidates per turn (\(13.04\) vs.\ \(10.33\)), and reduces redundancy almost completely (\(0.2\%\) vs.\ \(2.5\%\)). It also resolves more often at the turn level (\(9.3\%\) vs.\ \(1.8\%\)) and shortens average episodes (\(7.67\) vs.\ \(9.29\) turns). These trends match our theoretical analysis of EDPs (Sec.~\ref{sec:edp_theory}): when credit is assigned at the decision depth where an action changes the posterior, the policy learns to ask questions that contract the current candidate set rather than merely optimize a sparse final outcome in a multiturn task. The optimization diagnostics suggest specialization rather than mode collapse~\citep{janus2022mysteries,zhang2025verbalized}: \textsc{ECHO}'s entropy decreases while query diversity remains near saturation, indicating that
the policy becomes more selective without repeating the same questions. In contrast, GRPO maintains higher entropy but achieves lower information gain and far lower resolution, consistent with trajectory-level credit failing to identify which intermediate actions were useful under the current belief. Overall, these results indicate that \textbf{aggregate trajectory returns can
fail to identify the per-turn Bayesian advantage needed for epistemic credit.} This pattern matches the EDP perspective: when credit is assigned at the decision depth where an action changes the posterior, the policy learns to choose questions that contract the current belief state rather than relying on sparse trajectory-level outcomes.


\paragraph{Agent Prompts} \Cref{fig:csg-prompts} (A) shows prompts used for training of LLM agents in CSG environment as well as evaluation.  \Cref{fig:csg-prompts} (B) shows the prompt used to instantiate the clue (or hint) provider model with GPT-4o-mini with near-deterministic inference with temperature ($T=0$).  The CSG uses two prompts with distinct roles.
The selector prompt presents the agent with its current belief state ---
the remaining candidate list $C_t$ and the full arm query history ---
and requires structured JSON output specifying an arm index and a property question.
No chain-of-thought or reasoning scaffold is provided; the agent must select
informative actions directly from the belief state encoding.
The oracle prompt is given to a fixed \texttt{gpt-4o-mini} clue holder
at temperature zero, which receives the secret number and the agent's question
and returns a single deterministic sentence answering only the queried property.
The oracle is explicitly instructed not to volunteer information beyond what was asked
and to deflect vague questions, enforcing strict information separation across
the $K{=}5$ arms.
Repeat questions to the same arm incur a penalty of $r_t = -0.1$
rather than an oracle call, discouraging redundant queries.
This design ensures that all information gain flows exclusively through
the structured arm--question interface, making the per-turn reward signal
a clean measure of epistemic progress.

\begin{figure*}[t]
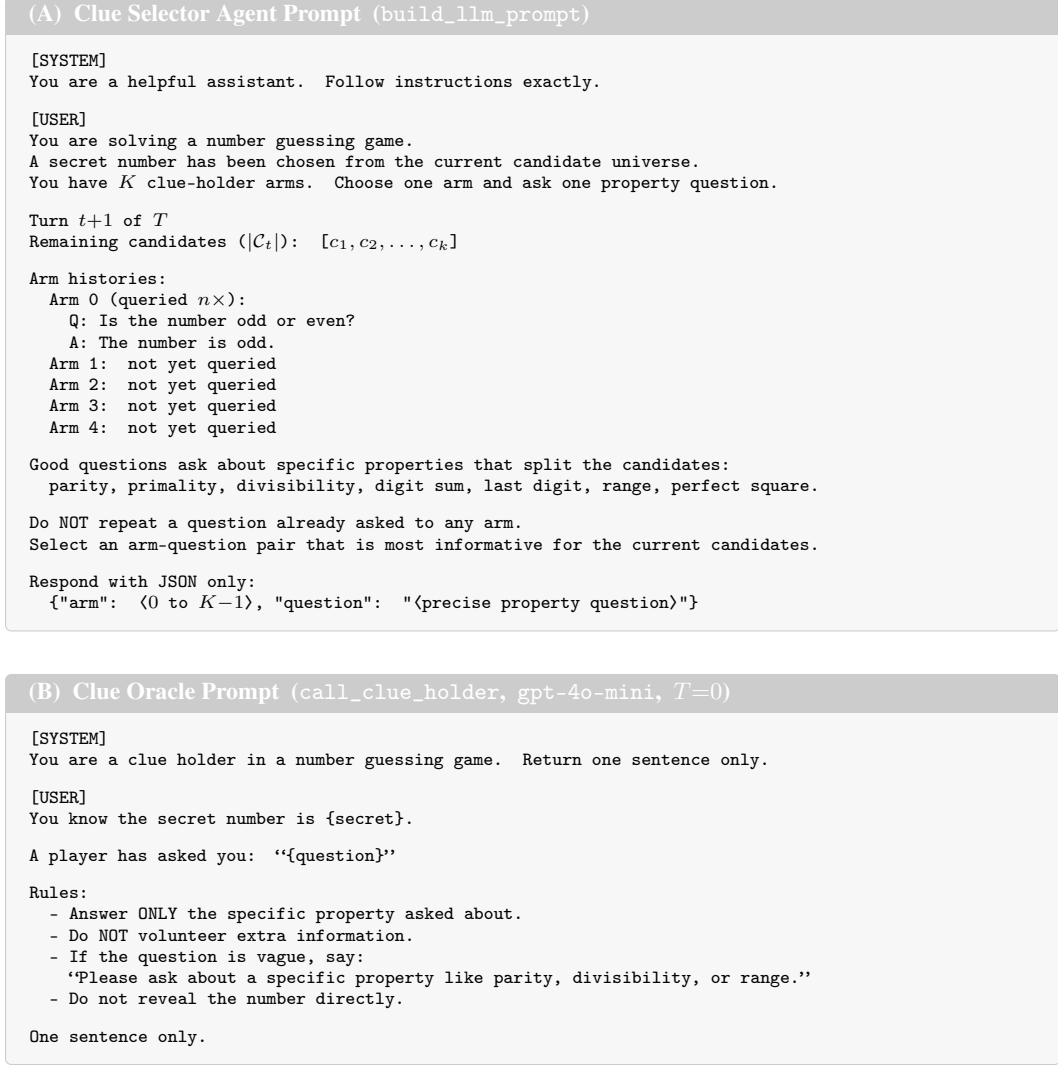

\centering
\small

\begin{tcolorbox}[
  title={\textbf{(A)}\enspace Clue Selector Agent Prompt\enspace(\texttt{build\_llm\_prompt})},
  fonttitle=\small\bfseries,
  colback=gray!6, colframe=gray!40,
  boxrule=0.4pt, arc=2pt,
  left=6pt, right=6pt, top=4pt, bottom=4pt,
  width=\linewidth
]
\ttfamily\scriptsize
\textbf{[SYSTEM]}\\
You are a helpful assistant. Follow instructions exactly.

\medskip
\textbf{[USER]}\\
You are solving a number guessing game.\\
A secret number has been chosen from the current candidate universe.\\
You have $K$ clue-holder arms. Choose one arm and ask one property question.

\medskip
Turn $t{+}1$ of $T$\\
Remaining candidates ($|\mathcal{C}_t|$): [$c_1, c_2, \ldots, c_k$]

\medskip
Arm histories:\\
\hspace*{1em}Arm 0 (queried $n\times$):\\
\hspace*{2em}Q: Is the number odd or even?\\
\hspace*{2em}A: The number is odd.\\
\hspace*{1em}Arm 1: not yet queried\\
\hspace*{1em}Arm 2: not yet queried\\
\hspace*{1em}Arm 3: not yet queried\\
\hspace*{1em}Arm 4: not yet queried

\medskip
Good questions ask about specific properties that split the candidates:\\
\hspace*{1em}parity, primality, divisibility, digit sum, last digit, range, perfect square.

\medskip
Do NOT repeat a question already asked to any arm.\\
Select an arm--question pair that is most informative for the current candidates.

\medskip
Respond with JSON only:\\
\hspace*{1em}\{"arm": \textlangle$0$ to $K{-}1$\textrangle, "question": "\textlangle precise property question\textrangle"\}
\end{tcolorbox}

\vspace{6pt}

\begin{tcolorbox}[
  title={\textbf{(B)}\enspace Clue Oracle Prompt\enspace(\texttt{call\_clue\_holder},\enspace\texttt{gpt-4o-mini},\enspace$T{=}0$)},
  fonttitle=\small\bfseries,
  colback=gray!6, colframe=gray!40,
  boxrule=0.4pt, arc=2pt,
  left=6pt, right=6pt, top=4pt, bottom=4pt,
  width=\linewidth
]
\ttfamily\scriptsize
\textbf{[SYSTEM]}\\
You are a clue holder in a number guessing game. Return one sentence only.

\medskip
\textbf{[USER]}\\
You know the secret number is \{secret\}.

\medskip
A player has asked you: ``\{question\}''

\medskip
Rules:\\
\hspace*{1em}-- Answer ONLY the specific property asked about.\\
\hspace*{1em}-- Do NOT volunteer extra information.\\
\hspace*{1em}-- If the question is vague, say:\\
\hspace*{2em}``Please ask about a specific property like parity, divisibility, or range.''\\
\hspace*{1em}-- Do not reveal the number directly.

\medskip
One sentence only.
\end{tcolorbox}

\caption{\small
Prompts used in the Clue Selector Game (CSG).
\textbf{(A)} The selector prompt shown to the \textbf{information-seeking agent} at each turn \(t\). The current candidate set \(\mathcal{C}_t\), remaining turn budget, and arm-specific query histories encode the agent's observable belief state. The agent must return JSON specifying an arm index and a property question, without a chain-of-thought scaffold.
\textbf{(B)} The \textbf{oracle agent} prompt used by a fixed \texttt{gpt-4o-mini} clue holder at temperature \(T{=}0\). The oracle knows the secret but may answer only the specific queried property and must \textit{not} reveal additional information and does not get the conversation history as input, motivated by prior work~\citep{zhang-etal-2024-probing}. The information reward is computed from the resulting candidate-set contraction,
\(r^{\mathrm{info}}_t=(|\mathcal{C}_t|-|\mathcal{C}_{t+1}|)/|\mathcal{C}_t|\).
The prompt format is held fixed across training and evaluation, with the candidate universe parameterized by the task setting.
}
\label{fig:csg-prompts}
\end{figure*}
\paragraph{Training dynamics and mechanism analysis.}
Figure~\ref{fig:clue_steps_additional_metrics} along with main metrics in \Cref{fig:clue_steps_four_panel_training_main} (\Cref{sec:results}) decomposes training into \textbf{outcome, mechanism, and optimization-health metrics}. The main pattern is that \textsc{ECHO} improves not only final task success, but also the intermediate epistemic behavior predicted by the EDP analysis. Episode resolve rate rises sharply, while mean eliminated candidates per turn increases, indicating that the policy learns to ask questions that contract the current candidate set more effectively. At the same time, redundancy falls to near zero, showing that the policy is not simply repeating high-probability questions, but is using the evolving history to avoid already-exhausted information. This supports the central motivation for \textbf{turn-level epistemic credit}: when the reward signal is assigned at the decision depth where an action changes the posterior, the policy receives direct supervision about which questions are useful under the current belief state.

\paragraph{Specialization without collapse.}
The diagnostic metrics further suggest that \textsc{ECHO}'s gains reflect healthy specialization rather than mode collapse. Entropy decreases during training, showing that the policy becomes more decisive, but diversity remains near saturation, indicating that it does not collapse onto a small set of repeated questions. Mechanistically, the policy appears to learn strong \textbf{within-step information gathering} before fully mastering higher-level arm selection: candidate elimination and efficiency improve substantially, while arm-match remains comparatively weaker. This is consistent with the task structure, where the policy must learn a joint action consisting of both an arm choice and a natural-language question. In EDP terms, \textsc{ECHO} first improves posterior-sensitive local action quality---asking useful questions for the current candidate set---while leaving additional headroom for learning better long-horizon coordination over which information source to query.

\begin{figure*}[t]
    \centering

    \includegraphics[width=0.24\linewidth]{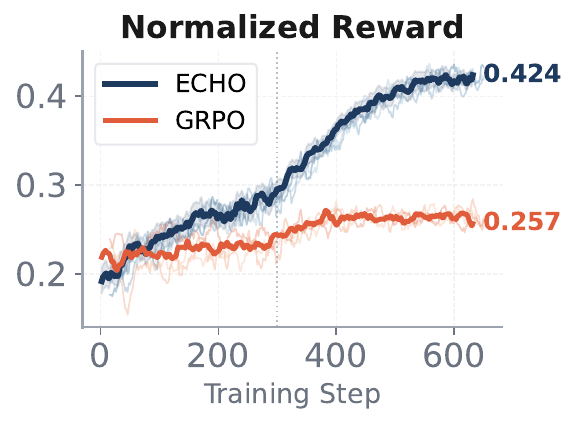}
    \hfill
    \includegraphics[width=0.24\linewidth]{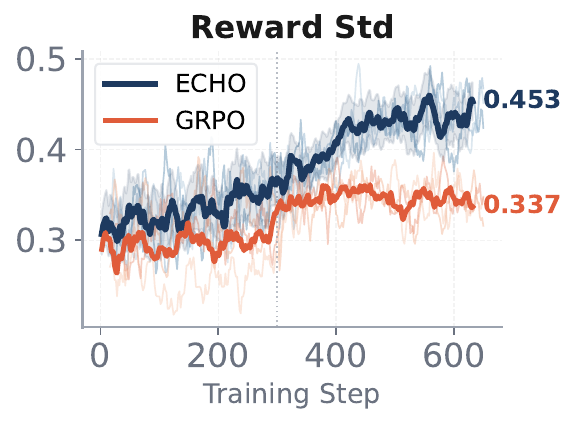}
    \hfill
    \includegraphics[width=0.24\linewidth]{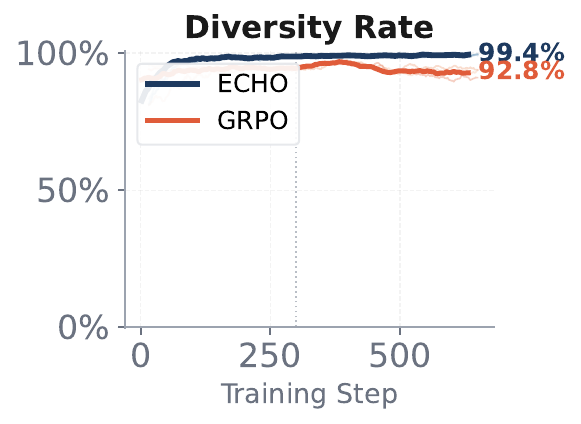}
    \hfill
    \includegraphics[width=0.24\linewidth]{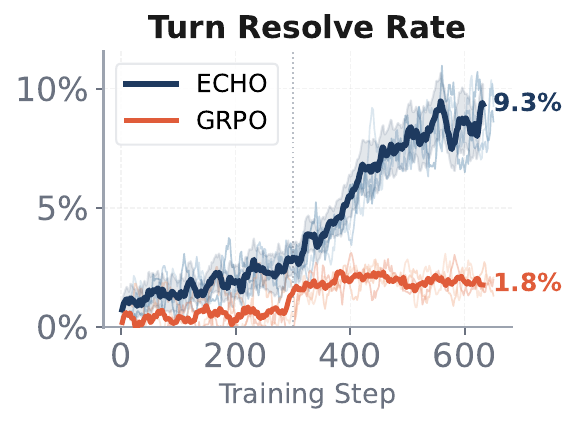}

    \vspace{0.8em}

    \includegraphics[width=0.24\linewidth]{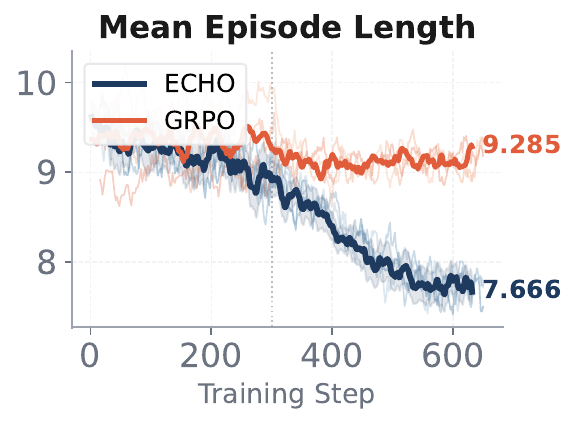}
    \hfill
    \includegraphics[width=0.24\linewidth]{figures/rl_per_step/clue_eliminated_mean.pdf}
    \hfill
    \includegraphics[width=0.24\linewidth]{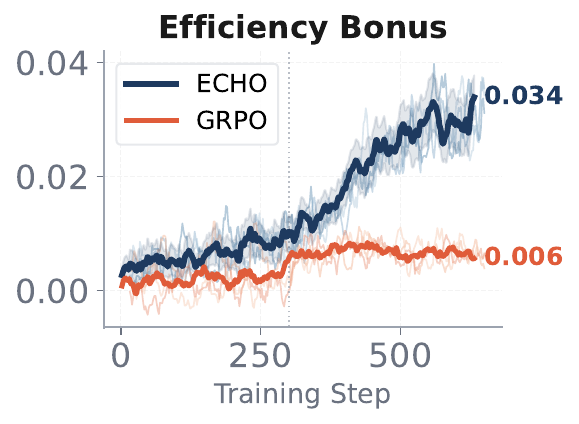}
    \hfill
    \includegraphics[width=0.24\linewidth]{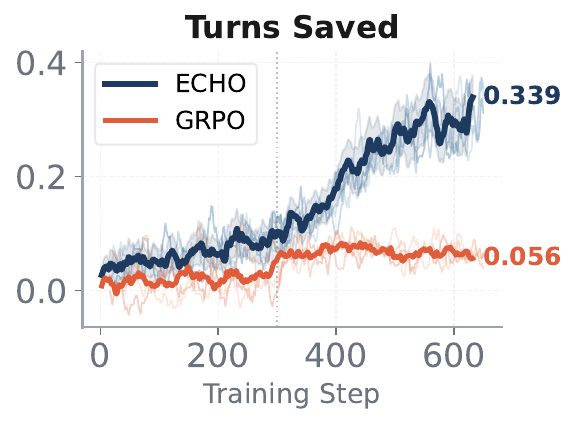}

    \vspace{0.8em}

    \includegraphics[width=0.31\linewidth]{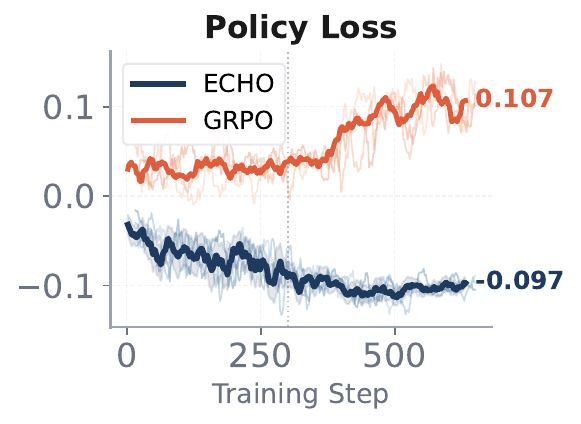}
    \hfill
    \includegraphics[width=0.31\linewidth]{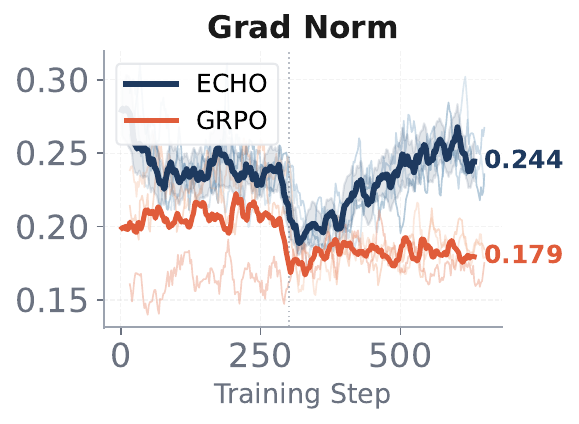}
    \hfill
    \includegraphics[width=0.31\linewidth]{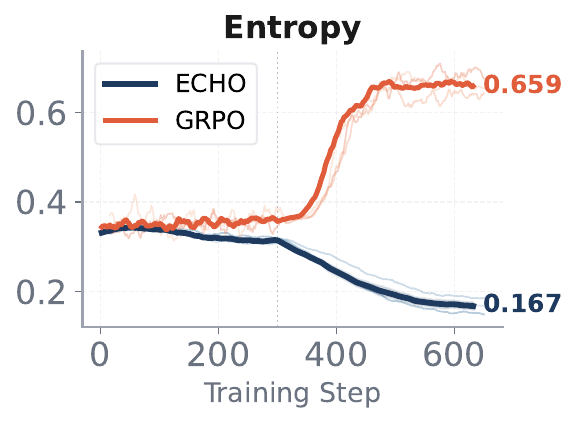}

    \caption{Additional learning dynamics of RL baselines and diagnostic metrics across training steps in the Clue Selector Game (CSG) environment.}
    \label{fig:clue_steps_additional_metrics}
\end{figure*}

\begin{table*}[t]
\centering
\scriptsize
\setlength{\tabcolsep}{6pt}
\renewcommand{\arraystretch}{1.12}
\begin{tabular}{lrrrr}
\toprule
\textbf{Metric} & \textbf{\textsc{ECHO}} & \textbf{GRPO} & \textbf{$\Delta$} & \textbf{Paired $t$-test $p$} \\
\midrule
Reward mean & 0.598 & 0.342 & +0.257 & $7.55{\times}10^{-4}$ \\
Normalized reward & 0.419 & 0.264 & +0.155 & $8.43{\times}10^{-4}$ \\
Episode resolve rate & 66.1\% & 17.1\% & +48.9 pts & $2.42{\times}10^{-4}$ \\
Mean eliminated candidates & 12.924 & 10.539 & +2.385 & $1.29{\times}10^{-3}$ \\
Redundancy rate $\downarrow$ & 0.26\% & 2.31\% & $-$2.04 pts & $1.55{\times}10^{-2}$ \\
Mean episode length $\downarrow$ & 7.724 & 9.136 & $-$1.412 & $9.88{\times}10^{-4}$ \\
Turn resolve rate & 8.66\% & 1.94\% & +6.72 pts & $3.69{\times}10^{-4}$ \\
Turns saved & 0.304 & 0.066 & +0.238 & $4.08{\times}10^{-4}$ \\
Entropy $\downarrow$ & 0.171 & 0.663 & $-$0.492 & $2.44{\times}10^{-3}$ \\
\bottomrule
\end{tabular}
\caption{
Final-window comparison between \textsc{ECHO} and episodic return baseline (GRPO). Values are seed-level means over the last 100 training steps, averaged across three seeds. \(p\)-values are from paired seed-level \(t\)-tests. Lower is better for redundancy rate, mean episode length, and entropy.
}
\label{tab:echo_grpo_final_window_stats}
\end{table*}

\paragraph{Seed-level significance.}
We compute final-window means over the last 100 training steps for each seed and compare \textsc{ECHO} against episodic GRPO using paired seed-level tests. \Cref{tab:echo_grpo_final_window_stats} shows the results. Across all three seeds, \textsc{ECHO} improves the primary outcome and mechanism metrics: episode resolve rate increases from \(17.1\%\) to \(66.1\%\), mean eliminated candidates rises from \(10.54\) to \(12.92\), turn resolve rate increases from \(1.94\%\) to \(8.66\%\), and mean episode length decreases from \(9.14\) to \(7.72\). Paired \(t\)-tests are significant for these metrics despite the small number of seeds, reflecting consistent improvements across runs. Since \(n=3\), we interpret these results primarily through effect size and cross-seed consistency rather than relying solely on significance tests.

\section{Additional Experiments}
\label{app:binary_search_echo}

\paragraph{How does standard RL credit assignment behave in CSG?} Because CSG is an environment with stochastic transitions, its feedback
supplied by an LLM oracle (\Cref{ssec:clue_task}), it is not obvious which components of standard RL credit assignment
carry over. We test them directly with a targeted ablation.
Specifically, \textsc{ECHO} makes two choices: it groups rollouts at matched turn depth, and
it centers each turn's reward on the group mean. We vary each in isolation,
dropping the grouping in episode-return GRPO and keeping it in
\textsc{ECHO}-LOO while substituting a leave-one-out mean. ECHO-LOO is an ablated version of \textsc{ECHO} that replaces the
group mean used for centering with a leave-one-out mean over the same
depth-matched group, in the spirit of the leave-one-out baseline of
\citet{ahmadian-etal-2024-back}, which is defined there in the bandit setting. Our results in Table~\ref{tab:ablation} show that the two are not
comparable in effect: discarding depth-matched
grouping costs roughly a third of the resolve rate at every checkpoint, whereas
changing the centering statistic has no measurable effect. Credit assignment in CSG must
therefore condition on turn depth, but is insensitive to how it is centered
once it does.

\begin{table}[t]
\centering
\small
\setlength{\tabcolsep}{4.5pt}
\begin{tabular}{lcccccc}
\toprule
& \multicolumn{3}{c}{Resolve (\%)} & \multicolumn{3}{c}{Reward} \\
\cmidrule(lr){2-4}\cmidrule(lr){5-7}
Method & 100 & 200 & 300 & 100 & 200 & 300 \\
\midrule
\textsc{ECHO}       & 10.4 & 14.3 & 19.3 & 0.305 & 0.342 & 0.371 \\
\textsc{ECHO}-LOO   & 10.7 & 14.3 & 19.3 & 0.320 & 0.341 & 0.380 \\
Episode-return      &  8.0 &  7.7 & 12.5 & 0.267 & 0.282 & 0.298 \\
\bottomrule
\end{tabular}
\caption{Resolve rate and reward on CSG for three advantage estimators; columns
index training steps. Removing depth-matched grouping (episode-return GRPO)
costs roughly a \textit{third} of the resolve rate as training progresses,
while replacing the group mean with a leave-one-out mean (\textsc{ECHO}-LOO)
leaves performance unchanged under identical conditions.}
\label{tab:ablation}
\end{table}


\paragraph{Additional Baselines}
\textbf{Offline SFT+RL} is a single-turn offline
baseline trained on 100 pre-collected Claude Sonnet 4.6 chain-of-thought
trajectories (867 turns, $17\%$ resolve rate), representing training on fixed
demonstrations from a stronger model rather than online interaction; this
baseline therefore additionally benefits from distillation of a frontier
demonstrator. \textbf{Base Qwen 2.5} at size 0.5B provides a further low-end
scaling reference. Further \textbf{Prompting variants} include explicit reasoning prompts on Qwen-2.5-7B-Instruct: chain-of-thought (CoT), full ReAct~\citep{yao2022react} using a complete Thought/Action/Observation loop with growing history context, and a simple version of ReAct using a Thought/Action format in the system prompt. These baselines test whether eliciting explicit reasoning can substitute for RL training on the same base model. The full results are shown in Table~\ref{tab:csg_main_results}.

\begin{table*}[t]
\centering
\scriptsize
\setlength{\tabcolsep}{3.6pt}
\renewcommand{\arraystretch}{1.08}
\resizebox{\linewidth}{!}{
\begin{tabular}{lccccccc}
\toprule
\textbf{Model} &
\textbf{Resolve$\uparrow$} &
\textbf{Zero$\downarrow$} &
\textbf{Qual.$\uparrow$} &
\textbf{Ground$\uparrow$} &
\textbf{GRecover$\uparrow$} &
\textbf{ResAfterZero$\uparrow$} &
\textbf{Reason\%} \\
\midrule

\multicolumn{8}{l}{\textit{Frontier baselines}} \\
\hdashline
Claude Sonnet 4.6
& \ms{43.1}{1.2}
& \ms{25.4}{0.4}
& \textbf{\ms{0.701}{0.005}}
& \textbf{\ms{0.738}{0.005}}
& \underline{\ms{0.185}{0.018}}
& \ms{0.168}{0.012}
& \ms{52.3}{0.5} \\

o3-mini
& \ms{27.1}{4.9}
& \ms{29.8}{1.9}
& \ms{0.610}{0.078}
& \ms{0.636}{0.082}
& \ms{0.248}{0.009}
& \ms{0.151}{0.027}
& \ms{0.3}{0.3} \\

GPT-4.1-mini
& \ms{31.6}{2.4}
& \ms{28.7}{0.6}
& \ms{0.615}{0.002}
& \ms{0.640}{0.003}
& \ms{0.285}{0.003}
& \ms{0.159}{0.012}
& \ms{0.0}{0.0} \\

GPT-4o
& \ms{25.8}{2.7}
& \ms{29.8}{1.2}
& \ms{0.625}{0.009}
& \ms{0.662}{0.007}
& \ms{0.286}{0.006}
& \ms{0.156}{0.017}
& \ms{0.0}{0.0} \\

GPT-4o + Full ReAct
& \ms{25.3}{1.3}
& \ms{28.3}{0.9}
& \ms{0.618}{0.002}
& \ms{0.625}{0.001}
& \ms{0.231}{0.014}
& \ms{0.087}{0.010}
& \ms{100.0}{0.0} \\

GPT-4o-mini
& \ms{22.7}{1.5}
& \ms{30.2}{0.4}
& \ms{0.563}{0.006}
& \ms{0.588}{0.006}
& \ms{0.352}{0.025}
& \ms{0.128}{0.010}
& \ms{0.0}{0.0} \\

GPT-4.1-nano
& \ms{10.2}{1.8}
& \ms{35.6}{0.6}
& \ms{0.496}{0.005}
& \ms{0.528}{0.010}
& \ms{0.327}{0.017}
& \ms{0.083}{0.014}
& \ms{0.0}{0.0} \\

\midrule
\multicolumn{8}{l}{\textit{Base and prompting baselines}} \\
\hdashline
Qwen2.5-14B
& \ms{22.7}{2.0}
& \ms{27.3}{0.7}
& \ms{0.535}{0.007}
& \ms{0.537}{0.008}
& \ms{0.290}{0.014}
& \ms{0.125}{0.011}
& \ms{1.8}{0.2} \\

Qwen2.5-7B
& \ms{20.9}{0.9}
& \ms{34.0}{0.5}
& \ms{0.502}{0.002}
& \ms{0.509}{0.003}
& \ms{0.200}{0.010}
& \ms{0.105}{0.024}
& \ms{2.2}{0.6} \\

Qwen2.5-3B
& \ms{8.4}{1.2}
& \ms{36.0}{1.0}
& \ms{0.454}{0.006}
& \ms{0.454}{0.007}
& \ms{0.280}{0.008}
& \ms{0.072}{0.017}
& \ms{0.8}{0.1} \\

Qwen2.5-1.5B
& \ms{8.0}{0.0}
& \ms{38.6}{1.0}
& \ms{0.352}{0.002}
& \ms{0.340}{0.005}
& \ms{0.229}{0.025}
& \ms{0.064}{0.004}
& \ms{1.1}{0.2} \\

Qwen2.5-1.5B + CoT
& \ms{6.7}{1.3}
& \ms{38.4}{0.6}
& \ms{0.362}{0.005}
& \ms{0.351}{0.007}
& \ms{0.223}{0.017}
& \ms{0.058}{0.012}
& \ms{1.0}{0.1} \\

Qwen2.5-1.5B + ReAct
& \ms{6.2}{1.2}
& \ms{40.7}{0.9}
& \ms{0.348}{0.007}
& \ms{0.339}{0.004}
& \ms{0.243}{0.020}
& \ms{0.058}{0.009}
& \ms{14.8}{0.7} \\

Qwen2.5-1.5B + Full ReAct
& \ms{4.4}{1.2}
& \ms{49.2}{1.1}
& \ms{0.312}{0.003}
& \ms{0.316}{0.002}
& \ms{0.205}{0.008}
& \ms{0.036}{0.012}
& \ms{100.0}{0.0} \\

Qwen2.5-0.5B
& \underline{\ms{3.6}{0.4}}
& \underline{\ms{50.2}{2.1}}
& \underline{\ms{0.281}{0.013}}
& \underline{\ms{0.268}{0.014}}
& \ms{0.189}{0.008}
& \underline{\ms{0.027}{0.008}}
& \ms{1.5}{0.3} \\

\midrule
\multicolumn{8}{l}{\textit{RL-trained Qwen2.5-1.5B baselines}} \\
\hdashline
\rowcolor{grposhade}
Episodic Return (GRPO)
& \ms{14.7}{1.5}
& \ms{39.2}{0.8}
& \ms{0.445}{0.008}
& \ms{0.446}{0.008}
& \ms{0.339}{0.004}
& \ms{0.093}{0.021}
& \ms{16.0}{1.5} \\


Offline SFT+RL
& \ms{13.3}{1.3}
& \ms{41.2}{0.8}
& \ms{0.392}{0.003}
& \ms{0.372}{0.005}
& \ms{0.206}{0.012}
& \ms{0.089}{0.012}
& \ms{1.0}{0.1} \\

\midrule
\rowcolor{echoshade}
\textbf{\textsc{ECHO} (Ours)}
& \textbf{\ms{45.3}{1.5}}
& \textbf{\ms{18.6}{1.0}}
& \ms{0.670}{0.008}
& \ms{0.718}{0.006}
& \textbf{\ms{0.406}{0.006}}
& \textbf{\ms{0.318}{0.018}}
& \ms{0.9}{0.0} \\

\midrule
\rowcolor{oracleshade}
Greedy oracle (reference)
& \ms{86.2}{0.4}
& \ms{29.6}{0.6}
& \ms{0.700}{0.006}
& \ms{0.704}{0.006}
& \ms{0.031}{0.001}
& \ms{0.088}{0.003}
& \ms{0.0}{0.0} \\

\bottomrule
\end{tabular}
}
\caption{
\textbf{Full CSG results}. Values are mean with SEM shown in subscript across \textbf{3} independent runs on the evaluation set. Resolve, Zero, and Reason\% are reported in percentage points; all other metrics are rates or normalized scores. Bold marks the best learned/model result by metric direction, and underlining marks the worst learned/model result. 
}
\vspace*{-4mm}
\label{tab:csg_main_results}
\end{table*}

\subsection{Main Evaluation Metric Definitions}
\label{ssec:main_task_epistemic_metrics_full}

\subsubsection{Clue Selector Game}
We compute task-level and process-level metrics from CSG episode logs over 10 maximum turns. Unless otherwise stated, turn-level metrics are averaged over turns within a run, and model-level values report mean and standard error across runs. Additional metrics that we consider and those that CSG RL training environment provides are noted in \Cref{tab:metric_definitions} (for training-related metrics) and \Cref{tab:arm_metric_definitions} (for arm-routing specific metrics). 

\begin{itemize}
    \item \textbf{Resolve} measures final task success: the fraction of episodes in which the agent identifies the secret within the turn budget.

    \item \textbf{Zero} measures zero-elimination behavior: the fraction of non-redundant turns that eliminate no candidates from the current candidate set. Redundant turns are excluded so that Zero captures new but uninformative questions.

    \item \textbf{Qual.} measures candidate-elimination quality relative to the current belief state. It normalizes the number of eliminated candidates by an ideal binary split of the current candidate set and caps the value at \(1\). Thus, high Qual. indicates questions that efficiently partition the remaining candidate set.

    \item \textbf{Ground} measures grounded information seeking: the fraction of turns that are both non-redundant and sufficiently informative under the current candidate set. In our implementation, a turn is grounded when its elimination quality is at least \(0.5\).

    \item \textbf{LateCalib} measures late-stage redundancy in narrow belief states. It is computed on turns with small remaining candidate sets late in the episode and records whether the action was redundant. Lower values indicate fewer wasted turns after the belief state has already narrowed.

    \item \textbf{Reason\%} measures visible reasoning text outside the required structured action format. It is used only as a diagnostic for whether a policy relies on explicit natural-language rationales rather than direct action selection.

    \item \textbf{ZeroEvents} counts zero-elimination non-redundant turns that have a following turn, and serves as the event denominator for recovery metrics.
 
\item \textbf{Rec1} measures weak immediate recovery after a Zero event: the fraction of Zero events after which the next turn is non-redundant and eliminates at least one candidate. It captures whether the agent resumes gaining any information at all.

\item \textbf{GRecover} measures grounded immediate recovery after a Zero event: the fraction of Zero events after which the next turn is grounded, meaning non-redundant and sufficiently informative under the current candidate set. This is stricter than Rec1 because the next turn must not merely eliminate something; it must pass the groundedness threshold.

    \item \textbf{RecQ} measures the elimination quality of the immediate next turn after a Zero event.

    \item \textbf{NextZero} measures whether a Zero event is followed by another zero-elimination non-redundant turn. Lower is better.

    \item \textbf{NextBad} measures whether a Zero event is followed by either another Zero event or a redundant turn. Lower is better.

    \item \textbf{TTR} measures the number of turns until the next grounded turn after a Zero event, averaged over cases where grounded recovery occurs. Lower is better.

    \item \textbf{TTRsucc} measures whether a grounded recovery occurs at any later turn after a Zero event before the episode ends.

    \item \textbf{ResAfterZero} measures episode-level robustness after local failure: among episodes containing at least one Zero event, it reports the fraction that still resolve.
\end{itemize}

\subsubsection{CRAFT/DPIP}

\begin{itemize}
    \item \textbf{Completion Rate} measures the percentage of games in which the goal structure was completed within the fixed $T=20$ turn budget.

    \item \textbf{Overall Progress} is a scalar summary used for termination
    and trend analysis, and is the unweighted mean of the next three metrics:

    \begin{itemize}
        \item \textbf{Completion Percentage} measures layer-exact
        correctness --- a block at position $c$ and layer $k$ counts as correct
        only if it matches $S^*(c)[k]$ for pre-specified goal state $S^*$.
        
        \item \textbf{Intersection over Union (IoU)} aggregates overlap across all positions. This metric is insensitive to block order within a stack and rewards
        partial position matches.
        
        \item \textbf{Position Accuracy} is coarser metric that rewards positions
        where the set of blocks matches exactly, regardless of layer order.
    \end{itemize}

    \item \textbf{Oracle Match} measures the percentage of moves that match on of the CRAFT oracle recommendations.

    \item \textbf{Invalid Move Rate} is the percentage of attempted moves that are illegal given the current board (such as placing a block at an occupied cell).

    \item \textbf{Parse Failure Rate} measures the percentage of Builder-generated moves whose form could not be parsed by the simulation engine.

    \item \textbf{Progress Gained} measures the difference between final progress and initial progress.

    \item \textbf{Episode Length} measures the mean number of episodes taken by each game in a sample. Games terminate after 20 turns even if not finshed, so lower mean episode length means more games were successfully completed before the 20-turn budget was exhausted.
\end{itemize}

\subsection{Progress Tracking and Delta Computation}
\label{subsec:tracking}

After each move, we compute a delta
$\Delta_t = \text{OP}(S_t, S^*) - \text{OP}(S_{t-1}, S^*)$
relative to the previous turn, for each metric. A recent trend estimate is computed
over a sliding window of size 3:

\[
\tau = \frac{1}{\min(t, 3)} \sum_{i=t-2}^{t} \Delta_i
\]

The game is considered \emph{improving} if $\tau > -0.05$, allowing
for small temporary regressions. An estimated turns remaining is
computed as $\lceil (1 - \text{OP}_t) / \bar{\Delta} \rceil$ where
$\bar{\Delta}$ is the average per-turn progress, though this estimate
degrades in stagnation conditions where $\bar{\Delta} \approx 0$.

\subsection{From Credit Signal to Strategy: The Emergence of Binary Search}

\paragraph{Oracle response distributions reveal emergent strategy.}
To characterize policy behavior beyond resolve rate, we analyze the distribution of oracle hint types received across evaluation episodes in CSG. This is shown in Table~\ref{tab:csg_hint_distribution_all}. For each turn, we classify the oracle response into coarse categories such as range, parity, divisibility, primality, square, digit properties, between, uninformative, and redundant, using keyword matching against the hint string. We also report the \emph{zero-elimination rate}: the fraction of non-redundant turns where the selected question eliminated no candidates. This captures cases where the agent asked a valid question, but the question was unhelpful under the current belief state. Table~\ref{tab:csg_hint_distribution_all} shows that \textsc{ECHO} has the lowest zero-elimination rate among evaluated policies (\(18.6\%\)), substantially below trajectory-level GRPO (Episodic Return) (\(39.8\%\)) and prompting baselines such as Qwen2.5-1.5B + CoT (\(38.4\%\)) and full ReAct (\(49.2\%\)). Its hint distribution is also highly concentrated on \textit{range-based} clues (\(60.7\%\)), with low redundancy (\(0.6\%\)) and little reliance on digit-property or miscellaneous hints. This suggests that \textsc{ECHO} does not merely improve final resolution; it changes the mechanism of interaction toward consistently informative, belief-sensitive queries.

\paragraph{Bisection as a special case of the EDP objective.}
The emergence of binary search in \textsc{ECHO} is a consequence of the CSG environment, not a universal property of epistemically adaptive policies. In CSG, the initial prior is uniform, \(\rho_0=\mathrm{Unif}(\{1,\ldots,N\})\), and deterministic oracle filtering preserves a uniform posterior over the remaining candidates, \(b_t=\mathrm{Unif}(\mathcal{C}_t)\). For a binary question that partitions \(\mathcal{C}_t\) into fractions \(p\) and \(1-p\), the expected fractional elimination is
\[
    \mathbb{E}_{z^\star\sim b_t}[u(b_t,a_t,o_t)] = 2p(1-p),
\]
which is maximized at \(p=1/2\). Thus, under uniform beliefs, deterministic binary observations, and equal action costs, the EDP objective recovers bisection as the information-theoretically optimal strategy.

In general EDPs, these conditions need not hold. If the belief is non-uniform, as in diagnosis or retrieval settings, the optimal query should account for posterior mass rather than split the hypothesis set evenly. If observations are graded, noisy, or multi-valued, the optimal action depends on the task-specific observation model and belief geometry. The Bayesian value \(Q^\pi_{\mathrm{B}}(h_t,a_t)=\mathbb{E}_{z^\star\sim b_t}[Q^\pi(h_t,a_t\mid z^\star)]\) remains the relevant object, but its maximizer is not necessarily a bisection question. Thus, bisection is a special solution for the uniform-prior, deterministic-observation regime of CSGG, while the EDP framework and \textsc{ECHO}'s posterior-sensitive turn-level credit are more general.

\paragraph{Binary-search-like strategy in CSG} In our experiments in CSG, the hint distribution clearly explains why \textsc{ECHO} converges toward a binary-search-like strategy rather than the digit-property-heavy strategy used by the greedy oracle. In the CSG instantiation of the EDP (\Cref{sec:experiments}), the dominant evidential utility is fractional candidate elimination,
\(u(b_t,a_t,o_t)=|\mathcal{C}_t\setminus\mathcal{C}_{t+1}|/|\mathcal{C}_t|\), which measures posterior contraction relative to the current belief size. Since the secret is sampled uniformly and oracle observations only eliminate inconsistent candidates, \(b_t=\mathrm{Unif}(\mathcal{C}_t)\). For a binary question whose ``yes'' set occupies fraction \(p\) of the current candidate set, the expected fractional elimination under this uniform belief is \(2p(1-p)\), maximized at \(p=1/2\). Range-bisection questions can therefore achieve near-maximal expected utility by splitting \(\mathcal{C}_t\) roughly in half. By contrast, digit-property questions such as last digit or digit sum can be highly discriminative for some secrets, but often correspond to imbalanced predicates under the uniform posterior, giving lower expected utility on average. Thus, \textbf{\textsc{ECHO}} learns the policy that is optimal for an agent uncertain over \(z^\star\): ask questions that maximize expected posterior collapse. The \textbf{greedy oracle} baseline, in contrast, effectively conditions on the true latent variable and can exploit digit properties with high realized elimination for that specific secret, explaining its \(50.7\%\) combined use of last-digit and digit-sum hints versus \textsc{ECHO}'s \(60.7\%\) use of range clues.

\begin{table*}[t]
\centering
\scriptsize
\setlength{\tabcolsep}{3.5pt}
\renewcommand{\arraystretch}{1.08}
\resizebox{\linewidth}{!}{
\begin{tabular}{lrrrrrrrrrrrrr}
\toprule
\textbf{Model} &
\textbf{Turns} &
\textbf{ZeroElim} &
\textbf{Range} &
\textbf{Divis.} &
\textbf{Parity} &
\textbf{Other} &
\textbf{Redund.} &
\textbf{Uninf.} &
\textbf{Prime} &
\textbf{Square} &
\textbf{Between} &
\textbf{LastDigit} &
\textbf{DigitSum} \\
\midrule

\multicolumn{14}{l}{\textit{Frontier baselines}} \\
\hdashline
Claude Sonnet 4.6
& 2005 & 509 {\scriptsize (25.4)}
& 969 {\scriptsize (48.3)}
& 417 {\scriptsize (20.8)}
& 230 {\scriptsize (11.5)}
& 286 {\scriptsize (14.3)}
& 14 {\scriptsize (0.7)}
& 16 {\scriptsize (0.8)}
& 35 {\scriptsize (1.7)}
& 19 {\scriptsize (0.9)}
& 3 {\scriptsize (0.1)}
& 12 {\scriptsize (0.6)}
& 4 {\scriptsize (0.2)} \\

o3-mini
& 2098 & 626 {\scriptsize (29.8)}
& 503 {\scriptsize (24.0)}
& 747 {\scriptsize (35.6)}
& 452 {\scriptsize (21.5)}
& 114 {\scriptsize (5.4)}
& 135 {\scriptsize (6.4)}
& --
& 68 {\scriptsize (3.2)}
& 21 {\scriptsize (1.0)}
& 1 {\scriptsize (0.0)}
& 20 {\scriptsize (1.0)}
& 37 {\scriptsize (1.8)} \\

GPT-4.1-mini
& 2051 & 588 {\scriptsize (28.7)}
& 590 {\scriptsize (28.8)}
& 1069 {\scriptsize (52.1)}
& 218 {\scriptsize (10.6)}
& 58 {\scriptsize (2.8)}
& 11 {\scriptsize (0.5)}
& 7 {\scriptsize (0.3)}
& 77 {\scriptsize (3.8)}
& 4 {\scriptsize (0.2)}
& 4 {\scriptsize (0.2)}
& 13 {\scriptsize (0.6)}
& -- \\

GPT-4o
& 6354 & 1873 {\scriptsize (29.5)}
& 1873 {\scriptsize (29.5)}
& 2737 {\scriptsize (43.1)}
& 686 {\scriptsize (10.8)}
& 350 {\scriptsize (5.5)}
& 19 {\scriptsize (0.3)}
& 51 {\scriptsize (0.8)}
& 262 {\scriptsize (4.1)}
& 187 {\scriptsize (2.9)}
& 6 {\scriptsize (0.1)}
& 111 {\scriptsize (1.7)}
& 72 {\scriptsize (1.1)} \\

GPT-4o-mini
& 2146 & 649 {\scriptsize (30.2)}
& 519 {\scriptsize (24.2)}
& 1180 {\scriptsize (55.0)}
& 224 {\scriptsize (10.4)}
& 20 {\scriptsize (0.9)}
& 5 {\scriptsize (0.2)}
& --
& 151 {\scriptsize (7.0)}
& 44 {\scriptsize (2.1)}
& 1 {\scriptsize (0.0)}
& 2 {\scriptsize (0.1)}
& -- \\

GPT-4.1-nano
& 2205 & 785 {\scriptsize (35.6)}
& 468 {\scriptsize (21.2)}
& 1187 {\scriptsize (53.8)}
& 200 {\scriptsize (9.1)}
& 113 {\scriptsize (5.1)}
& 96 {\scriptsize (4.4)}
& 1 {\scriptsize (0.0)}
& 85 {\scriptsize (3.9)}
& 14 {\scriptsize (0.6)}
& 23 {\scriptsize (1.0)}
& 17 {\scriptsize (0.8)}
& 1 {\scriptsize (0.0)} \\

GPT-4o + ReAct
& 2082 & 590 {\scriptsize (28.3)}
& 607 {\scriptsize (29.2)}
& 779 {\scriptsize (37.4)}
& 242 {\scriptsize (11.6)}
& 185 {\scriptsize (8.9)}
& 2 {\scriptsize (0.1)}
& 22 {\scriptsize (1.1)}
& 48 {\scriptsize (2.3)}
& 68 {\scriptsize (3.3)}
& 1 {\scriptsize (0.0)}
& 71 {\scriptsize (3.4)}
& 57 {\scriptsize (2.7)} \\

\midrule
\multicolumn{14}{l}{\textit{RL-trained baselines}} \\
\hdashline
Episodic Return (GRPO)
& 4333 & 1723 {\scriptsize (39.8)}
& 1001 {\scriptsize (23.1)}
& 1532 {\scriptsize (35.4)}
& 812 {\scriptsize (18.7)}
& 583 {\scriptsize (13.5)}
& 118 {\scriptsize (2.7)}
& 12 {\scriptsize (0.3)}
& 152 {\scriptsize (3.5)}
& 106 {\scriptsize (2.4)}
& 10 {\scriptsize (0.2)}
& 6 {\scriptsize (0.1)}
& 1 {\scriptsize (0.0)} \\


Offline SFT+RL
& 2152 & 887 {\scriptsize (41.2)}
& 34 {\scriptsize (1.6)}
& 1689 {\scriptsize (78.5)}
& 235 {\scriptsize (10.9)}
& 104 {\scriptsize (4.8)}
& 52 {\scriptsize (2.4)}
& 3 {\scriptsize (0.1)}
& 28 {\scriptsize (1.3)}
& 4 {\scriptsize (0.2)}
& --
& 3 {\scriptsize (0.1)}
& -- \\

\midrule
\multicolumn{14}{l}{\textit{Base and prompting baselines}} \\
\hdashline
Qwen2.5-0.5B
& 2219 & 1115 {\scriptsize (50.2)}
& 397 {\scriptsize (17.9)}
& 878 {\scriptsize (39.6)}
& 385 {\scriptsize (17.4)}
& 185 {\scriptsize (8.3)}
& 222 {\scriptsize (10.0)}
& 41 {\scriptsize (1.8)}
& 47 {\scriptsize (2.1)}
& 14 {\scriptsize (0.6)}
& 3 {\scriptsize (0.1)}
& 42 {\scriptsize (1.9)}
& 5 {\scriptsize (0.2)} \\

Qwen2.5-1.5B
& 8856 & 3698 {\scriptsize (41.8)}
& 1125 {\scriptsize (12.7)}
& 2878 {\scriptsize (32.5)}
& 1556 {\scriptsize (17.6)}
& 1308 {\scriptsize (14.8)}
& 983 {\scriptsize (11.1)}
& 216 {\scriptsize (2.4)}
& 379 {\scriptsize (4.3)}
& 186 {\scriptsize (2.1)}
& 88 {\scriptsize (1.0)}
& 97 {\scriptsize (1.1)}
& 40 {\scriptsize (0.5)} \\

Qwen2.5-3B
& 2204 & 794 {\scriptsize (36.0)}
& 133 {\scriptsize (6.0)}
& 1559 {\scriptsize (70.7)}
& 242 {\scriptsize (11.0)}
& 96 {\scriptsize (4.4)}
& 133 {\scriptsize (6.0)}
& --
& 20 {\scriptsize (0.9)}
& 21 {\scriptsize (1.0)}
& --
& --
& -- \\

Qwen2.5-7B
& 2108 & 716 {\scriptsize (34.0)}
& 213 {\scriptsize (10.1)}
& 1419 {\scriptsize (67.3)}
& 228 {\scriptsize (10.8)}
& 125 {\scriptsize (5.9)}
& 38 {\scriptsize (1.8)}
& 14 {\scriptsize (0.7)}
& 54 {\scriptsize (2.6)}
& 16 {\scriptsize (0.8)}
& --
& --
& 1 {\scriptsize (0.0)} \\

Qwen2.5-14B
& 2110 & 576 {\scriptsize (27.3)}
& 85 {\scriptsize (4.0)}
& 1648 {\scriptsize (78.1)}
& 232 {\scriptsize (11.0)}
& 39 {\scriptsize (1.8)}
& 36 {\scriptsize (1.7)}
& 2 {\scriptsize (0.1)}
& 59 {\scriptsize (2.8)}
& 7 {\scriptsize (0.3)}
& --
& 2 {\scriptsize (0.1)}
& -- \\

Qwen2.5-1.5B + CoT
& 2209 & 849 {\scriptsize (38.4)}
& 293 {\scriptsize (13.3)}
& 795 {\scriptsize (36.0)}
& 410 {\scriptsize (18.6)}
& 312 {\scriptsize (14.1)}
& 243 {\scriptsize (11.0)}
& 10 {\scriptsize (0.5)}
& 81 {\scriptsize (3.7)}
& 62 {\scriptsize (2.8)}
& 2 {\scriptsize (0.1)}
& 1 {\scriptsize (0.0)}
& -- \\

Qwen2.5-1.5B + ReAct
& 4448 & 2000 {\scriptsize (45.0)}
& 591 {\scriptsize (13.3)}
& 1289 {\scriptsize (29.0)}
& 727 {\scriptsize (16.3)}
& 683 {\scriptsize (15.4)}
& 480 {\scriptsize (10.8)}
& 191 {\scriptsize (4.3)}
& 206 {\scriptsize (4.6)}
& 61 {\scriptsize (1.4)}
& 86 {\scriptsize (1.9)}
& 95 {\scriptsize (2.1)}
& 39 {\scriptsize (0.9)} \\

Qwen2.5-1.5B + ReAct$^\dagger$
& 2224 & 1095 {\scriptsize (49.2)}
& 302 {\scriptsize (13.6)}
& 544 {\scriptsize (24.5)}
& 316 {\scriptsize (14.2)}
& 383 {\scriptsize (17.2)}
& 213 {\scriptsize (9.6)}
& 132 {\scriptsize (5.9)}
& 110 {\scriptsize (4.9)}
& 7 {\scriptsize (0.3)}
& 85 {\scriptsize (3.8)}
& 95 {\scriptsize (4.3)}
& 37 {\scriptsize (1.7)} \\

\midrule
\multicolumn{14}{l}{\textit{Ours}} \\
\hdashline
\rowcolor{echoshade}
\textbf{\textsc{ECHO}}
& 1998 & \textbf{371 {\scriptsize (18.6)}}
& \textbf{1213 {\scriptsize (60.7)}}
& 440 {\scriptsize (22.0)}
& 252 {\scriptsize (12.6)}
& 37 {\scriptsize (1.9)}
& 11 {\scriptsize (0.6)}
& --
& 16 {\scriptsize (0.8)}
& 2 {\scriptsize (0.1)}
& 27 {\scriptsize (1.4)}
& --
& -- \\

\midrule
\multicolumn{14}{l}{\textit{Oracle}} \\
\hdashline
\rowcolor{oracleshade}
Greedy oracle
& 763 & 226 {\scriptsize (29.6)}
& 64 {\scriptsize (8.4)}
& 134 {\scriptsize (17.6)}
& 43 {\scriptsize (5.6)}
& --
& --
& --
& 43 {\scriptsize (5.6)}
& 40 {\scriptsize (5.2)}
& 52 {\scriptsize (6.8)}
& \textbf{168 {\scriptsize (22.0)}}
& \textbf{219 {\scriptsize (28.7)}} \\

\bottomrule
\end{tabular}
}
\caption{\textbf{Hint-type distribution} from the Oracle (GPT-4o-mini) across all evaluated policies in the Clue Selector Game evaluation set. Each cell reports the number of turns assigned to a hint category, with the percentage of total turns for that policy in parentheses. ZeroElim denotes non-redundant turns that eliminated no candidates. \textsc{ECHO} is highlighted and shows the lowest zero-elimination rate among learned policies while relying more strongly on range-based clues.}
\label{tab:csg_hint_distribution_all}
\end{table*}

\paragraph{Arm routing versus belief-conditioned information seeking.}

\begin{table}[t]
\centering
\scriptsize
\setlength{\tabcolsep}{4pt}
\renewcommand{\arraystretch}{1.08}
\begin{tabular}{p{0.18\linewidth}p{0.74\linewidth}}
\toprule
\textbf{Metric} & \textbf{Meaning} \\
\midrule




\textbf{ArmEnt} & Normalized entropy of arm choices, measuring how evenly the policy distributes actions across the available arms. \\

\textbf{DomArm} & Dominant-arm share: the fraction of turns assigned to the policy's most frequently selected arm. \\

\textbf{StaticComp} & Static arm compatibility: whether the question type is semantically compatible with the fixed property family assigned to the selected arm. \\

\textbf{ArmMatch} & Logged arm-match rate: agreement between the selected arm and the environment's target arm label. \\

\textbf{SameArm} & Fraction of turns in which the policy selects the same arm as in the previous turn. \\

\textbf{ZeroSame} & Zero-elimination rate conditioned on selecting the same arm as in the previous turn. \\

\textbf{ZeroSwitch} & Zero-elimination rate conditioned on switching to a different arm from the previous turn. \\

\textbf{QualSame} & Mean elimination quality conditioned on selecting the same arm as in the previous turn. \\

\textbf{QualSwitch} & Mean elimination quality conditioned on switching to a different arm from the previous turn. \\
\bottomrule
\end{tabular}
\caption{Definitions of arm-routing and belief-state diagnostic metrics.}
\label{tab:arm_metric_definitions}
\end{table}
Table~\ref{tab:arm_routing_belief_diagnostics} tests whether \textsc{ECHO}'s gains are explained by raw arm exploration, static arm matching, or belief-conditioned information gain. The results rule out a simple exploration account: \textsc{ECHO} has lower arm entropy than several weaker baselines such as GRPO and Qwen2.5-1.5B base model, yet achieves substantially higher resolve rate. Instead, \textsc{ECHO} exhibits the strongest learned-policy combination of low zero-elimination rate, high elimination quality, high groundedness, and high semantic compatibility between selected arms and question types. Raw ArmMatch alone is not sufficient to explain performance: \textsc{ECHO}'s ArmMatch is modest, but its selected arm-question pairs eliminate candidates much more reliably than GRPO or base-model baselines. The same-arm diagnostics further show that arm repetition is not inherently harmful. While frontier models often suffer when reusing an arm, \textsc{ECHO}'s same-arm turns remain highly informative, with lower zero-elimination and higher elimination quality than its switched-arm turns. This suggests that \textsc{ECHO} learns to reuse and route arms in a belief-conditioned manner rather than merely maximizing static arm labels.

\begin{table*}[t]
\centering
\scriptsize
\setlength{\tabcolsep}{4.2pt}
\renewcommand{\arraystretch}{1.08}
\resizebox{\linewidth}{!}{
\begin{tabular}{lccccccccccccc}
\toprule
\textbf{Model} &
\textbf{Res. $\uparrow$} &
\textbf{Zero $\downarrow$} &
\textbf{Qual. $\uparrow$} &
\textbf{Ground $\uparrow$} &
\textbf{ArmEnt} &
\textbf{DomArm} &
\textbf{StaticComp $\uparrow$} &
\textbf{ArmMatch $\uparrow$} &
\textbf{SameArm} &
\textbf{ZeroSame $\downarrow$} &
\textbf{ZeroSwitch $\downarrow$} &
\textbf{QualSame $\uparrow$} &
\textbf{QualSwitch $\uparrow$} \\
\midrule

\rowcolor{oracleshade}
Greedy oracle (reference)
& 0.862 & 0.296 & 0.700 & 0.704 & 0.839 & 0.507 & 1.000 & 1.000 & 0.468 & 0.393 & 0.423 & 0.603 & 0.570 \\

\midrule
\rowcolor{echoshade}
\textbf{\textsc{ECHO} (Ours)}
& \textbf{0.453} & \textbf{0.186} & \textbf{0.670} & \textbf{0.718} & 0.770 & 0.536 & \textbf{0.581} & 0.183 & 0.565 & \textbf{0.172} & \textbf{0.257} & \textbf{0.670} & \textbf{0.582} \\

\midrule
\multicolumn{14}{l}{\textit{Frontier baselines}} \\
\hdashline
Claude Sonnet 4.6
& 0.431 & 0.254 & 0.701 & 0.738 & 0.976 & 0.309 & 0.345 & 0.146 & 0.112 & 0.645 & 0.239 & 0.313 & 0.721 \\
GPT-4o
& 0.258 & 0.298 & 0.625 & 0.662 & 0.996 & 0.230 & 0.391 & 0.172 & 0.215 & 0.596 & 0.261 & 0.318 & 0.654 \\

\midrule
\multicolumn{14}{l}{\textit{RL-trained Qwen2.5-1.5B baselines}} \\
\hdashline
Episodic Return (GRPO)
& 0.118 & 0.398 & 0.423 & 0.419 & 0.944 & 0.292 & 0.288 & 0.141 & 0.378 & 0.393 & 0.469 & 0.409 & 0.354 \\

\midrule
\multicolumn{14}{l}{\textit{Base and prompting baselines}} \\
\hdashline
Qwen2.5-1.5B
& 0.080 & 0.386 & 0.352 & 0.340 & 0.944 & 0.319 & 0.238 & 0.138 & 0.391 & 0.401 & 0.446 & 0.300 & 0.287 \\
Qwen2.5-1.5B + ReAct
& 0.044 & 0.492 & 0.312 & 0.316 & 0.826 & 0.537 & 0.313 & 0.189 & 0.552 & 0.525 & 0.564 & 0.237 & 0.265 \\

\bottomrule
\end{tabular}
}
\caption{
\textbf{Arm-routing and belief-state diagnostics in the Clue Selector Game}. See \Cref{tab:arm_metric_definitions} for of arm-routing and belief-state diagnostic metrics in CSG.}
\label{tab:arm_routing_belief_diagnostics}
\end{table*}

\begin{table*}[t]
\centering
\scriptsize
\setlength{\tabcolsep}{4.2pt}
\renewcommand{\arraystretch}{1.05}
\resizebox{\linewidth}{!}{
\begin{tabular}{lccccccccccc}
\toprule
\textbf{Model} & \textbf{Resolve} & \textbf{N} & \textbf{Rec1 $\uparrow$} & \textbf{GRec1 $\uparrow$} & \textbf{RecQ $\uparrow$} & \textbf{Gain $\uparrow$} & \textbf{TTR $\downarrow$} & \textbf{TTRsucc $\uparrow$} & \textbf{Sust2 $\uparrow$} & \textbf{NextRed $\downarrow$} & \textbf{EpRes $\uparrow$} \\
\midrule
\multicolumn{12}{l}{\textit{Frontier baselines}} \\
\hdashline
Claude Sonnet 4.6 & 43.1\% & 509 & 0.145 & 0.145 & 0.139 & 0.078 & 2.23 & 0.303 & 0.102 & 0.240 & 0.169 \\
o3-mini & 32.0\% & 384 & 0.198 & 0.190 & 0.179 & 0.096 & 1.34 & 0.258 & 0.202 & 0.255 & 0.171 \\
GPT-4.1-mini & 31.6\% & 588 & 0.223 & 0.219 & 0.201 & 0.109 & \textbf{1.32} & 0.298 & 0.154 & 0.243 & 0.156 \\
GPT-4o & 25.8\% & 634 & 0.229 & 0.219 & 0.195 & 0.107 & 1.38 & 0.314 & 0.226 & 0.249 & 0.157 \\
GPT-4o + ReAct & 25.3\% & 590 & 0.193 & 0.173 & 0.166 & 0.097 & 1.50 & 0.276 & 0.122 & 0.254 & 0.087 \\
GPT-4o-mini & 22.7\% & 649 & 0.337 & 0.284 & 0.269 & 0.145 & 1.45 & 0.413 & 0.162 & 0.203 & 0.126 \\
GPT-4.1-nano & 10.2\% & 785 & 0.317 & 0.271 & 0.258 & 0.146 & 1.60 & 0.469 & 0.122 & 0.242 & 0.082 \\

\midrule
\multicolumn{12}{l}{\textit{Base and prompting baselines}} \\
\hdashline
Qwen2.5-14B-Instruct & 22.7\% & 576 & 0.311 & 0.227 & 0.217 & 0.124 & 1.41 & 0.339 & 0.132 & 0.248 & 0.121 \\
Qwen2.5-7B-Instruct & 20.9\% & 716 & 0.189 & 0.159 & 0.152 & 0.085 & 1.59 & 0.254 & 0.051 & 0.236 & 0.106 \\
Qwen2.5-3B-Instruct & 8.4\% & 794 & 0.316 & 0.234 & 0.239 & 0.136 & 1.64 & 0.390 & 0.162 & 0.243 & 0.072 \\
Qwen2.5-1.5B-Instruct & 8.0\% & 849 & 0.329 & 0.200 & 0.212 & 0.127 & 2.16 & 0.476 & 0.064 & 0.269 & 0.064 \\
Qwen2.5-1.5B + CoT & 6.7\% & 849 & 0.304 & 0.194 & 0.203 & 0.122 & 2.09 & 0.429 & 0.068 & 0.254 & 0.058 \\
Qwen2.5-1.5B + ReAct & 6.2\% & 905 & 0.326 & 0.213 & 0.224 & 0.134 & 2.22 & \textbf{0.492} & 0.068 & 0.261 & 0.058 \\
Qwen2.5-1.5B + ReAct$^\dagger$ & 4.4\% & 1095 & 0.256 & 0.181 & 0.179 & 0.111 & 2.27 & 0.453 & 0.051 & 0.215 & 0.036 \\
Qwen2.5-0.5B-Instruct & 3.6\% & 1115 & 0.274 & 0.167 & 0.182 & 0.106 & 2.37 & 0.421 & 0.034 & 0.225 & 0.027 \\

\midrule
\multicolumn{12}{l}{\textit{RL-trained Qwen2.5-1.5B baselines}} \\
\hdashline
Episodic Return (GRPO) & 11.8\% & 2014 & \textbf{0.364} & 0.259 & 0.259 & \textbf{0.153} & 1.82 & 0.488 & 0.096 & 0.202 & 0.089 \\
Offline SFT+RL & 13.3\% & 887 & 0.262 & 0.170 & 0.182 & 0.116 & 1.95 & 0.363 & 0.042 & 0.200 & 0.089 \\

\midrule
\rowcolor{echoshade}
\textbf{\textsc{ECHO} (Ours)} & \textbf{45.3\%} & 371 & 0.342 & \textbf{0.313} & \textbf{0.284} & 0.151 & 1.37 & 0.437 & \textbf{0.261} & 0.243 & \textbf{0.305} \\
\bottomrule
\end{tabular}
}
\caption{Recovery metrics conditioned on zero-information events in the Clue Selector Game. Metrics summarize whether a model recovers after an uninformative step, how quickly it recovers, whether recovery is sustained, and whether the episode ultimately resolves. Arrows indicate preferred direction; best values are bolded.}
\label{tab:csg_recovery_zero}
\end{table*}
\begin{table*}[t]
\centering
\scriptsize
\setlength{\tabcolsep}{4.2pt}
\renewcommand{\arraystretch}{1.05}
\resizebox{\linewidth}{!}{
\begin{tabular}{lccccccccccc}
\toprule
\textbf{Model} & \textbf{Resolve} & \textbf{N} & \textbf{Rec1 $\uparrow$} & \textbf{GRec1 $\uparrow$} & \textbf{RecQ $\uparrow$} & \textbf{Gain $\uparrow$} & \textbf{TTR $\downarrow$} & \textbf{TTRsucc $\uparrow$} & \textbf{Sust2 $\uparrow$} & \textbf{NextRed $\downarrow$} & \textbf{EpRes $\uparrow$} \\
\midrule
\multicolumn{12}{l}{\textit{Frontier baselines}} \\
\hdashline
Claude Sonnet 4.6 & 43.1\% & 742 & 0.193 & 0.193 & 0.188 & 0.019 & 1.86 & 0.311 & 0.075 & 0.282 & 0.369 \\
o3-mini & 32.0\% & 524 & 0.208 & 0.202 & 0.193 & 0.024 & 1.32 & 0.263 & 0.150 & 0.267 & 0.297 \\
GPT-4.1-mini & 31.6\% & 867 & 0.300 & 0.254 & 0.241 & 0.062 & 1.40 & 0.363 & 0.140 & 0.234 & 0.267 \\
GPT-4o & 25.8\% & 880 & 0.266 & 0.257 & 0.232 & 0.054 & \textbf{1.31} & 0.340 & 0.220 & 0.247 & 0.244 \\
GPT-4o + ReAct & 25.3\% & 797 & 0.237 & 0.212 & 0.205 & 0.051 & 1.46 & 0.320 & 0.124 & 0.240 & 0.200 \\
GPT-4o-mini & 22.7\% & 942 & 0.394 & \textbf{0.311} & \textbf{0.300} & 0.116 & 1.50 & 0.461 & 0.185 & 0.206 & 0.220 \\
GPT-4.1-nano & 10.2\% & 1024 & 0.302 & 0.250 & 0.240 & 0.110 & 1.59 & 0.422 & 0.114 & 0.275 & 0.094 \\

\midrule
\multicolumn{12}{l}{\textit{Base and prompting baselines}} \\
\hdashline
Qwen2.5-14B-Instruct & 22.7\% & 963 & \textbf{0.415} & 0.255 & 0.256 & 0.096 & 1.63 & 0.437 & 0.116 & 0.227 & 0.205 \\
Qwen2.5-7B-Instruct & 20.9\% & 965 & 0.230 & 0.184 & 0.177 & 0.058 & 1.57 & 0.289 & 0.058 & 0.244 & 0.195 \\
Qwen2.5-3B-Instruct & 8.4\% & 1088 & 0.310 & 0.222 & 0.224 & 0.106 & 1.63 & 0.369 & 0.142 & 0.276 & 0.072 \\
Qwen2.5-1.5B-Instruct & 8.0\% & 1257 & 0.317 & 0.202 & 0.212 & 0.113 & 2.12 & 0.457 & 0.068 & 0.294 & 0.076 \\
Qwen2.5-1.5B + CoT & 6.7\% & 1244 & 0.295 & 0.184 & 0.193 & 0.100 & 2.04 & 0.391 & 0.063 & 0.295 & 0.062 \\
Qwen2.5-1.5B + ReAct & 6.2\% & 1315 & 0.313 & 0.208 & 0.215 & 0.115 & 2.20 & \textbf{0.479} & 0.060 & 0.287 & 0.062 \\
Qwen2.5-1.5B + ReAct$^\dagger$ & 4.4\% & 1418 & 0.261 & 0.183 & 0.183 & 0.107 & 2.21 & 0.437 & 0.054 & 0.243 & 0.036 \\
Qwen2.5-0.5B-Instruct & 3.6\% & 1422 & 0.259 & 0.160 & 0.171 & 0.093 & 2.32 & 0.398 & 0.036 & 0.245 & 0.031 \\

\midrule
\multicolumn{12}{l}{\textit{RL-trained Qwen2.5-1.5B baselines}} \\
\hdashline
Episodic Return (GRPO) & 11.8\% & 2520 & 0.357 & 0.252 & 0.252 & \textbf{0.125} & 1.81 & 0.474 & 0.096 & 0.200 & 0.097 \\
Offline SFT+RL & 13.3\% & 1110 & 0.280 & 0.183 & 0.191 & 0.097 & 1.93 & 0.380 & 0.048 & 0.200 & 0.106 \\

\midrule
\rowcolor{echoshade}
\textbf{\textsc{ECHO} (Ours)} & \textbf{45.3\%} & 621 & 0.337 & 0.304 & 0.277 & 0.072 & 1.39 & 0.432 & \textbf{0.239} & 0.298 & \textbf{0.422} \\
\bottomrule
\end{tabular}
}
\caption{Recovery metrics conditioned on bad events in the Clue Selector Game. Metrics evaluate how models respond after poor intermediate choices, including immediate recovery, grounded recovery, recovery quality, speed of recovery, sustained improvement, redundancy avoidance, and final episode resolution. Arrows indicate preferred direction; best values are bolded.}
\label{tab:csg_recovery_bad}
\end{table*}
\begin{table*}[t]
\centering
\scriptsize
\setlength{\tabcolsep}{4.2pt}
\renewcommand{\arraystretch}{1.05}
\resizebox{\linewidth}{!}{
\begin{tabular}{lccccccccccc}
\toprule
\textbf{Model} & \textbf{Resolve} & \textbf{N} & \textbf{Rec1 $\uparrow$} & \textbf{GRec1 $\uparrow$} & \textbf{RecQ $\uparrow$} & \textbf{Gain $\uparrow$} & \textbf{TTR $\downarrow$} & \textbf{TTRsucc $\uparrow$} & \textbf{Sust2 $\uparrow$} & \textbf{NextRed $\downarrow$} & \textbf{EpRes $\uparrow$} \\
\midrule
\multicolumn{12}{l}{\textit{Frontier baselines}} \\
\hdashline
Claude Sonnet 4.6 & 43.1\% & 728 & 0.196 & 0.196 & 0.192 & 0.019 & 1.86 & 0.317 & 0.075 & 0.268 & 0.369 \\
o3-mini & 32.0\% & 524 & 0.208 & 0.202 & 0.193 & 0.024 & 1.32 & 0.263 & 0.150 & 0.267 & 0.297 \\
GPT-4.1-mini & 31.6\% & 856 & 0.303 & 0.256 & 0.243 & 0.062 & 1.40 & 0.367 & 0.141 & 0.225 & 0.267 \\
GPT-4o & 25.8\% & 868 & 0.270 & 0.260 & 0.236 & 0.055 & \textbf{1.31} & 0.344 & 0.220 & 0.240 & 0.244 \\
GPT-4o + ReAct & 25.3\% & 795 & 0.238 & 0.213 & 0.205 & 0.051 & 1.46 & 0.321 & 0.125 & 0.239 & 0.200 \\
GPT-4o-mini & 22.7\% & 937 & 0.395 & \textbf{0.313} & \textbf{0.301} & 0.116 & 1.50 & 0.463 & 0.186 & 0.204 & 0.220 \\
GPT-4.1-nano & 10.2\% & 928 & 0.321 & 0.267 & 0.256 & 0.117 & 1.61 & 0.457 & 0.115 & 0.234 & 0.094 \\

\midrule
\multicolumn{12}{l}{\textit{Base and prompting baselines}} \\
\hdashline
Qwen2.5-14B-Instruct & 22.7\% & 927 & \textbf{0.428} & 0.262 & 0.263 & 0.098 & 1.64 & 0.451 & 0.118 & 0.206 & 0.205 \\
Qwen2.5-7B-Instruct & 20.9\% & 927 & 0.239 & 0.192 & 0.184 & 0.060 & 1.57 & 0.300 & 0.059 & 0.227 & 0.195 \\
Qwen2.5-3B-Instruct & 8.4\% & 955 & 0.323 & 0.228 & 0.231 & 0.107 & 1.63 & 0.380 & 0.144 & 0.245 & 0.072 \\
Qwen2.5-1.5B-Instruct & 8.0\% & 997 & 0.333 & 0.213 & 0.223 & 0.117 & 2.12 & 0.484 & 0.067 & 0.267 & 0.076 \\
Qwen2.5-1.5B + CoT & 6.7\% & 1001 & 0.303 & 0.195 & 0.204 & 0.103 & 2.07 & 0.418 & 0.070 & 0.257 & 0.062 \\
Qwen2.5-1.5B + ReAct & 6.2\% & 1048 & 0.320 & 0.211 & 0.220 & 0.116 & 2.24 & \textbf{0.489} & 0.063 & 0.263 & 0.058 \\
Qwen2.5-1.5B + ReAct$^\dagger$ & 4.4\% & 1205 & 0.271 & 0.190 & 0.188 & 0.108 & 2.25 & 0.465 & 0.055 & 0.206 & 0.036 \\
Qwen2.5-0.5B-Instruct & 3.6\% & 1200 & 0.268 & 0.165 & 0.178 & 0.096 & 2.36 & 0.415 & 0.036 & 0.228 & 0.031 \\

\midrule
\multicolumn{12}{l}{\textit{RL-trained Qwen2.5-1.5B baselines}} \\
\hdashline
Episodic Return (GRPO) & 11.8\% & 2374 & 0.356 & 0.254 & 0.253 & \textbf{0.124} & 1.82 & 0.477 & 0.098 & 0.200 & 0.096 \\
Offline SFT+RL & 13.3\% & 1058 & 0.276 & 0.181 & 0.190 & 0.095 & 1.93 & 0.375 & 0.049 & 0.194 & 0.101 \\

\midrule
\rowcolor{echoshade}
\textbf{\textsc{ECHO} (Ours)} & \textbf{45.3\%} & 610 & 0.336 & 0.305 & 0.277 & 0.071 & 1.39 & 0.433 & \textbf{0.234} & 0.293 & \textbf{0.422} \\
\bottomrule
\end{tabular}
}
\caption{Recovery metrics conditioned on low-quality events in the Clue Selector Game. Metrics summarize immediate recovery, grounded recovery, recovery quality, recovery latency, sustained improvement, redundancy after failure, and final episode resolution. Arrows indicate preferred direction; best values are bolded.}
\label{tab:csg_recovery_low}
\end{table*}
\begin{table*}[t]
\centering
\scriptsize
\setlength{\tabcolsep}{4.2pt}
\renewcommand{\arraystretch}{1.05}
\resizebox{\linewidth}{!}{
\begin{tabular}{lccccccccccc}
\toprule
\textbf{Model} & \textbf{Resolve} & \textbf{N} & \textbf{Rec1 $\uparrow$} & \textbf{GRec1 $\uparrow$} & \textbf{RecQ $\uparrow$} & \textbf{Gain $\uparrow$} & \textbf{TTR $\downarrow$} & \textbf{TTRsucc $\uparrow$} & \textbf{Sust2 $\uparrow$} & \textbf{NextRed $\downarrow$} & \textbf{EpRes $\uparrow$} \\
\midrule
\multicolumn{12}{l}{\textit{Frontier baselines}} \\
\hdashline
Claude Sonnet 4.6 & 43.1\% & 665 & 0.191 & 0.191 & 0.187 & 0.009 & \textbf{1.29} & 0.250 & 0.064 & 0.314 & 0.369 \\
o3-mini & 32.0\% & 498 & 0.175 & 0.171 & 0.161 & 0.002 & 1.34 & 0.225 & 0.091 & 0.281 & 0.292 \\
GPT-4.1-mini & 31.6\% & 780 & 0.231 & 0.214 & 0.200 & 0.036 & 1.34 & 0.292 & 0.089 & 0.260 & 0.263 \\
GPT-4o & 25.8\% & 809 & 0.214 & 0.205 & 0.187 & 0.024 & 1.34 & 0.282 & 0.140 & 0.268 & 0.234 \\
GPT-4o + ReAct & 25.3\% & 747 & 0.213 & 0.189 & 0.181 & 0.033 & 1.41 & 0.277 & 0.092 & 0.256 & 0.200 \\
GPT-4o-mini & 22.7\% & 810 & 0.320 & 0.259 & 0.247 & 0.084 & 1.51 & 0.390 & 0.132 & 0.240 & 0.198 \\
GPT-4.1-nano & 10.2\% & 876 & 0.231 & 0.192 & 0.183 & 0.075 & 1.60 & 0.332 & 0.059 & 0.322 & 0.090 \\

\midrule
\multicolumn{12}{l}{\textit{Base and prompting baselines}} \\
\hdashline
Qwen2.5-14B-Instruct & 22.7\% & 862 & \textbf{0.349} & 0.233 & 0.231 & 0.075 & 1.52 & 0.381 & 0.105 & 0.254 & 0.191 \\
Qwen2.5-7B-Instruct & 20.9\% & 918 & 0.193 & 0.158 & 0.151 & 0.043 & 1.56 & 0.255 & 0.046 & 0.256 & 0.191 \\
Qwen2.5-3B-Instruct & 8.4\% & 919 & 0.232 & 0.155 & 0.154 & 0.064 & 1.67 & 0.287 & 0.065 & 0.319 & 0.068 \\
Qwen2.5-1.5B-Instruct & 8.0\% & 917 & 0.231 & 0.158 & 0.162 & 0.081 & 1.88 & 0.322 & 0.033 & 0.361 & 0.063 \\
Qwen2.5-1.5B + CoT & 6.7\% & 966 & 0.211 & 0.134 & 0.138 & 0.064 & 1.82 & 0.275 & 0.035 & 0.351 & 0.062 \\
Qwen2.5-1.5B + ReAct & 6.2\% & 949 & 0.229 & 0.158 & 0.161 & 0.079 & 1.97 & 0.334 & 0.042 & 0.349 & 0.062 \\
Qwen2.5-1.5B + ReAct$^\dagger$ & 4.4\% & 960 & 0.164 & 0.123 & 0.121 & 0.071 & 1.93 & 0.253 & 0.021 & 0.329 & 0.032 \\
Qwen2.5-0.5B-Instruct & 3.6\% & 960 & 0.167 & 0.095 & 0.104 & 0.054 & 1.93 & 0.210 & 0.019 & 0.330 & 0.027 \\

\midrule
\multicolumn{12}{l}{\textit{RL-trained Qwen2.5-1.5B baselines}} \\
\hdashline
Episodic Return (GRPO) & 11.8\% & 2009 & 0.268 & 0.191 & 0.190 & 0.083 & 1.76 & 0.358 & 0.054 & 0.243 & 0.090 \\
Offline SFT+RL & 13.3\% & 978 & 0.217 & 0.154 & 0.155 & 0.068 & 1.90 & 0.322 & 0.030 & 0.224 & 0.089 \\

\midrule
\rowcolor{echoshade}
\textbf{\textsc{ECHO} (Ours)} & \textbf{45.3\%} & 593 & 0.312 & \textbf{0.285} & \textbf{0.261} & 0.060 & 1.39 & \textbf{0.405} & \textbf{0.192} & 0.312 & \textbf{0.421} \\
\bottomrule
\end{tabular}
}
\caption{Recovery metrics conditioned on late bad events in the Clue Selector Game. Metrics evaluate whether models can recover from harmful or low-quality decisions occurring later in an episode, where fewer turns remain for correction. Arrows indicate preferred direction; best values are bolded.}
\label{tab:csg_recovery_late_bad}
\end{table*}
\subsection{ RL Results}

\paragraph{CSG Reward Composition} The reward  $r_t$ is computed as follows. See \Cref{sec:experiments} for motivation on why this reward design is appropriate for our theoretical and experimental setup.  Specifically, if the policy has already asked the same question to the selected arm earlier in the episode, the turn returns immediately with the redundancy penalty $r_t = -0.1$ and the oracle is not consulted. Otherwise the oracle returns a hint, the candidate set is filtered, and $r_t$ is the sum of up to five components: the fractional information gain $r_{\text{info}} = (\text{candidates\_before} - \text{candidates\_after}) / \text{candidates\_before}$, which lies in $[0,1]$ and dominates the dense per-turn signal; an arm-match bonus of $+0.05$ when the question matches a property the selected arm covers (keyword match against \texttt{ARM\_PROPERTIES}); a diversity bonus of $+0.05$ when the question has not been asked to any arm earlier in the episode; a sparse resolution bonus of $+1.0$ when the post-oracle candidate set equals $\{\text{secret}\}$; and an efficiency bonus of $0.1 \cdot (T_{\max} - t)$ on resolution, with $T_{\max} = 10$ turns.

\paragraph{CSG Training Hyperparameters} 

\begin{table}[t]
\centering
\small
\setlength{\tabcolsep}{5pt}
\begin{tabular}{@{}ll@{}}
\toprule
\textbf{Parameter} & \textbf{Value} \\
\midrule
\multicolumn{2}{@{}l}{\textit{Policy and optimization}} \\
Policy model                & Qwen2.5-1.5B-Instruct \\
Per-device train batch size & 4 \\
Generations per prompt $G$  & 4 \\
Generation batch size       & 16 ($4\times4$) \\
Steps per generation        & 4 \\
Learning rate               & $5\mathrm{e}{-6}$ \\
KL coefficient $\beta$      & 0.0 \\
Max gradient norm           & 0.5 \\
Precision                   & bf16 \\
Sampling temperature        & 1.0 \\
Training steps              & 650 \\
Seeds                       & 42, 123, 7 \\
\midrule
\multicolumn{2}{@{}l}{\textit{LoRA adapter}} \\
Rank $r$                    & 16 \\
$\alpha$ (scaling $\alpha/r$) & 32 (2.0) \\
Dropout                     & 0.05 \\
Bias                        & none \\
Target modules              & q, k, v, o\_proj \\
Trainable parameters        & 4.4M (0.28\%) \\
\bottomrule
\end{tabular}
\caption{RL training hyperparameters used for ECHO and Episodic GRPO baselines for our main experiments/results in \Cref{tab:combined_results}.}
\label{tab:rl_hyperparams_echo}
\end{table}
Training parameters for policy alignment for ECHO and episodic GRPO are provided in \Cref{tab:rl_hyperparams_echo}. Inspired by prior works~\citep{hu2025openreasonerzeroopensourceapproach}, we set the KL coefficient $\beta_{\text{KL}} = 0$ to remove the reference-model anchor and permit the policies to explore in the CSG environment. Gradients are clipped at
$\|\nabla\| \le 0.5$, and the PPO clip range uses the TRL default\footnote{\url{https://huggingface.co/docs/trl/en/grpo_trainer}} of
$\varepsilon = 0.2$, applied symmetrically.


\paragraph{CRAFT Reward Composition}
The Builder's reward $r_t$ is computed as follows. See \Cref{sec:experiments} for motivation on why this reward design is appropriate for our theoretical and experimental setup. An output that cannot be parsed into a \texttt{PLACE}, \texttt{REMOVE} or \texttt{CLARIFY} command receives $r_t = -0.10$; a \texttt{CLARIFY} question receives $r_t = -0.05$; a non-clarify move rejected by the game as physically invalid receives $r_t = -0.15$; and a valid move that is not among the oracle's candidate moves (matched on action, cell and layer) receives $r_t = -0.10$. Otherwise the move is executed and $r_t$ is the change in task progress $\Delta P_t$, which dominates the dense per-turn signal. Task progress $P \in [0,1]$ is the mean of three board-level scores: per-cell block-type IoU, the fraction of target blocks at the correct cell and layer, and the fraction of cells whose block set matches the target. $\Delta P_t$ is negative when a move takes the board away from the target, and both signs are weighted by $1.0$. When the move completes the structure ($P \ge 0.95$), the Builder additionally receives a sparse completion bonus of $+0.3$ and an efficiency bonus of $0.1 \cdot (T_{\max} - t)$, with $T_{\max} = 20$ turns.

\paragraph{CRAFT Training Hyperparameters}

\begin{table}[t]
\centering
\small
\setlength{\tabcolsep}{5pt}
\begin{tabular}{@{}ll@{}}
\toprule
\textbf{Parameter} & \textbf{Value} \\
\midrule
\multicolumn{2}{@{}l}{\textit{Policy and optimization}} \\
Policy model                & Qwen2.5-7B-Instruct \\
Per-device train batch size & 2 \\
Generations per prompt $G$  & 3 \\
Generation batch size       & 6 ($2\times3$) \\
Steps per generation        & 3 \\
Learning rate               & $5\mathrm{e}{-6}$, linear decay \\
LR schedule horizon         & 500 steps \\
KL coefficient $\beta$      & 0.0 \\
Max gradient norm           & 0.5 \\
Precision                   & bf16 \\
Sampling temperature        & 1.0 (top-$p$ 0.9) \\
Max prompt / completion     & 3584 / 220 tokens \\
Training steps              & 350 \\
Seed                        & 42 \\
\midrule
\multicolumn{2}{@{}l}{\textit{Environment}} \\
Turn budget $T_{\max}$      & 20 \\
Oracle candidate moves      & up to 20 \\
Directors                   & GPT-4.1-mini ($T=0.7$) \\
Training structures         & 500 \\
\midrule
\multicolumn{2}{@{}l}{\textit{LoRA adapter}} \\
Rank $r$                    & 16 \\
$\alpha$ (scaling $\alpha/r$) & 32 (2.0) \\
Dropout                     & 0.05 \\
Bias                        & none \\
Target modules              & q, k, v, o\_proj \\
Trainable parameters        & 10.1M (0.13\%) \\
\bottomrule
\end{tabular}
\caption{RL training hyperparameters used for ECHO and Episodic GRPO baselines in CRAFT for the results in \Cref{tab:combined_results}.}
\label{tab:rl_hyperparams_craft}
\end{table}
Training parameters for ECHO and episodic GRPO in CRAFT are provided in \Cref{tab:rl_hyperparams_craft}. We keep the CSG optimization settings ($\beta_{\text{KL}} = 0$, $\|\nabla\| \le 0.5$, symmetric clip range $\varepsilon = 0.2$) and the same LoRA configuration, and train Qwen2.5-7B-Instruct on a pool of 500 generated structures disjoint from the 20 evaluation structures. Each CRAFT rollout requires up to 60 Director calls, so we use $G = 3$ rollouts per structure: each generation phase plays 18 episodes on 2 structures, split into 3 optimizer steps of 6 episodes (2 groups) each. Starting boards are drawn uniformly from six configurations (empty, first layer, first two layers, or one Director's wall). Both methods share the same seed and a 500-step linear learning-rate schedule, and are evaluated at checkpoint-350, after roughly 2{,}100 training episodes. We did not perform hyperparameter search due to the well-defined and controlled nature of the training setup, and instead used a fixed set of hyperparameters across experiments

\paragraph{Computing Infrastructure and Packages}
All training runs use GPUs with bf16 precision and gradient checkpointing. Rollouts are generated with Hugging Face \texttt{generate}. In CSG, training Qwen2.5-1.5B-Instruct for 650 steps takes approximately 16 hours for \textsc{ECHO} and episodic GRPO, each on one NVIDIA RTX Pro 6000 Blackwell. In CRAFT, training Qwen2.5-7B-Instruct for 350 steps takes approximately 90 GPU-hours for \textsc{ECHO} and episodic GRPO, each on one NVIDIA L40S.

All RL training builds on TRL's \texttt{GRPOTrainer}, which we subclass to run multi-turn rollouts and to replace the advantage computation, while keeping the TRL clipped objective with token-level loss aggregation (\texttt{loss\_type="dapo"}). Models are loaded with Hugging Face Transformers and LoRA adapters are implemented with PEFT. All experiments use PyTorch 2.10.0, Transformers 4.57.3, TRL 0.24.0.dev0, PEFT 0.17.1, Accelerate 1.12.0 and bitsandbytes 0.49.1 (4-bit NF4 with double quantization and bf16 compute, used only to evaluate Qwen2.5-32B and 72B), on Python 3.14 and CUDA 13.0. The CSG oracle, CRAFT Directors and API baselines are called through the official OpenAI and Anthropic Python SDKs. Means, standard errors and permutation tests are computed with NumPy.

\begin{table*}[t]
\centering
\scriptsize
\setlength{\tabcolsep}{5pt}
\renewcommand{\arraystretch}{1.15}
\begin{tabularx}{\linewidth}{p{3.0cm}p{4.6cm}X}
\toprule
\textbf{Metric} & \textbf{Definition} & \textbf{Diagnostic Role in Clue Selector Training} \\
\midrule
\multicolumn{3}{l}{\textit{Outcome metrics --- did the policy solve the task?}} \\

\midrule
Episode Resolve Rate & Fraction of episodes (capped at 10 turns) where the policy identifies the secret number. & Primary outcome. The only metric that directly answers whether the policy solves the task; bounded above by 1, so the trajectory bends as performance approaches the ceiling. \\
Turn Resolve Rate & Fraction of individual turns that reduce the candidate set to a single element. & Per-turn analogue of episode resolve. Distinguishes ``the policy resolved this episode'' from ``the policy resolved on this specific turn,'' surfacing whether resolution comes from a single decisive action or from gradual narrowing. \\
Reward & Mean scalar reward across all turns in the batch; primary GRPO signal. & Coarse summary of all reward components. Useful as a sanity check but conflates information gain, resolution bonus, and efficiency bonus; should always be read alongside the component metrics. \\
Normalized Reward & Reward divided by the maximum achievable per episode. & Scale-free measure of proximity to optimal play. Lets cross-configuration comparisons abstract away changes in reward weighting. \\
\midrule
\multicolumn{3}{l}{\textit{Mechanism metrics --- how is the policy solving (or failing to solve) the task?}} \\
\midrule
Eliminated Mean & Average number of candidates removed per turn after applying the oracle's response. & Direct measure of information gain per question. With candidate sets near 25--30, an information-maximizing question eliminates roughly half; values approaching this threshold indicate the policy has discovered near-optimal binary partitioning at the question level. \\
Candidates Before / After & Average candidate set size at the start of a turn / after applying the oracle's response. & Pre- and post-action uncertainty. The gap is per-turn information gain; the ratio is the fractional reduction. A falling \emph{Candidates Before} across the episode indicates earlier turns are doing useful work that compresses later turns' starting state. \\
Arm Match Rate & Fraction of turns where the selected arm matches the oracle's preferred arm. & Strategic alignment metric, distinct from question quality. With $K$ arms, the random baseline is $1/K$. Values near baseline while \emph{Episode Resolve Rate} climbs indicate the policy is winning through within-arm question quality rather than informed arm selection --- a decomposition that flags untapped headroom in higher-order strategy. \\
Mean Turn & Average turn index at which resolution occurs, conditioned on successful episodes. & Efficiency on solved episodes. Low values indicate the policy resolves early in the budget; the gap between this and \emph{Mean Episode Length} reflects the cost of unresolved episodes running to the cap. \\
Mean Episode Length & Average total turns per episode, including unresolved runs. & Aggregate efficiency. Drops toward \emph{Mean Turn} as \emph{Episode Resolve Rate} climbs, because fewer episodes are exhausting the turn cap without resolving. \\
Turns Saved & Average turns saved relative to a baseline episode length on resolved episodes. & Efficiency expressed in turn units. Couples resolution with promptness; nonzero values indicate the policy is not just resolving but resolving \emph{faster} than the baseline expects. \\
Efficiency Bonus & Additional reward for resolving episodes earlier than baseline. & Reward-side counterpart to \emph{Turns Saved}. Grows non-linearly with both arm match rate and resolve rate, since both must improve for early resolution to become routine. \\
\midrule
\multicolumn{3}{l}{\textit{Environment-correctness metrics --- is the policy seeing the reward signal we intended?}} \\
\midrule
Redundancy Rate & Fraction of turns where the oracle's response eliminates zero candidates. & Diagnostic of the filter rather than the policy. High values indicate that questions intended to be informative are being scored as wasted, gating the reward signal upstream of any learning. Near-zero values are a precondition for the mechanism metrics above to behave as designed. \\
Diversity Rate & Fraction of turns where the question was not asked previously in the same episode. & Guards against trivial repetition. Once near 1, it functions as a passive monitor: a drop is the first sign of mode collapse onto a small set of high-reward questions. \\
\midrule
\bottomrule
\end{tabularx}
\caption{Metrics used to analyze policy performance during training, training dynamics over turns, and environment correctness in the Clue selector task. Metrics are grouped by diagnostic role: outcome metrics report task success; mechanism metrics decompose \emph{how} success is achieved (information gain per question vs.\ strategic arm selection vs.\ efficiency); environment-correctness metrics confirm the reward signal is reaching the policy as intended}

\label{tab:metric_definitions}
\end{table*}

\subsection{Silent Exploration Behavior}

\begin{theorybox}

\paragraph{Silent exploration.}
When the policy adapts to the posterior without emitting an explicit natural-language rationale, \nk{we call this behavior \emph{silent exploration}}. A silently exploratory policy changes what it does in response to evidence, even if its output is only a compact structured action. This matters because visible reasoning is neither necessary nor sufficient for epistemic adaptivity: as our experimental results suggest, a model may produce long explanations without acting on the current belief, while another may output no rationale and still act adaptively.

\end{theorybox}

\textsc{ECHO} also achieves its gains with almost no explicit reasoning text (\(0.9\%\)) exposing its silently exploratory nature, whereas ReAct variants reason on every turn but do not improve resolution, indicating that explicit reasoning traces alone do not substitute for epistemically adaptive training. \an{Arm-level diagnostics further show that 
\textsc{ECHO}'s arm--question choices remain informative even under same-arm reuse, suggesting belief-conditioned routing rather than superficial arm switching (full results in \Cref{tab:arm_routing_belief_diagnostics} in Appendix~\ref{app:binary_search_echo}).}

\paragraph{ECHO achieves frontier-level epistemic quality without explicit reasoning.}

Without leveraging a significant amount of reasoning tokens, our approach is able to outperform models 100X larger in the frontier range, as well as base models larger in the parameter scale. At just 1.5B scale, ECHO-trained models outperform strong frontier models including reasoning models like O3-mini as well as strong baselines like CoT~\citep{wei2023chainofthought} and React~\citep{yao2022react} with GPT-4o. This shows the adaptability in multiturn information-seeking tasks within an EDP setting requires “adaptability” as well as “recovery” from bad decisions—a behavior we see emerge in models trained with our approach. Notably, this boost is overall performance is seen even without requiring the costly amount of reasoning tokens before the final answer. This suggests that ECHO trains models to be adaptable using belief-sensitive reward structures and can provide strong performance with higher token cost-efficiency—a trend consistent with prior works~\citep{ma2025reasoningmodelseffectivethinking}.

\section{CSG Episode Dataset and Statistics}

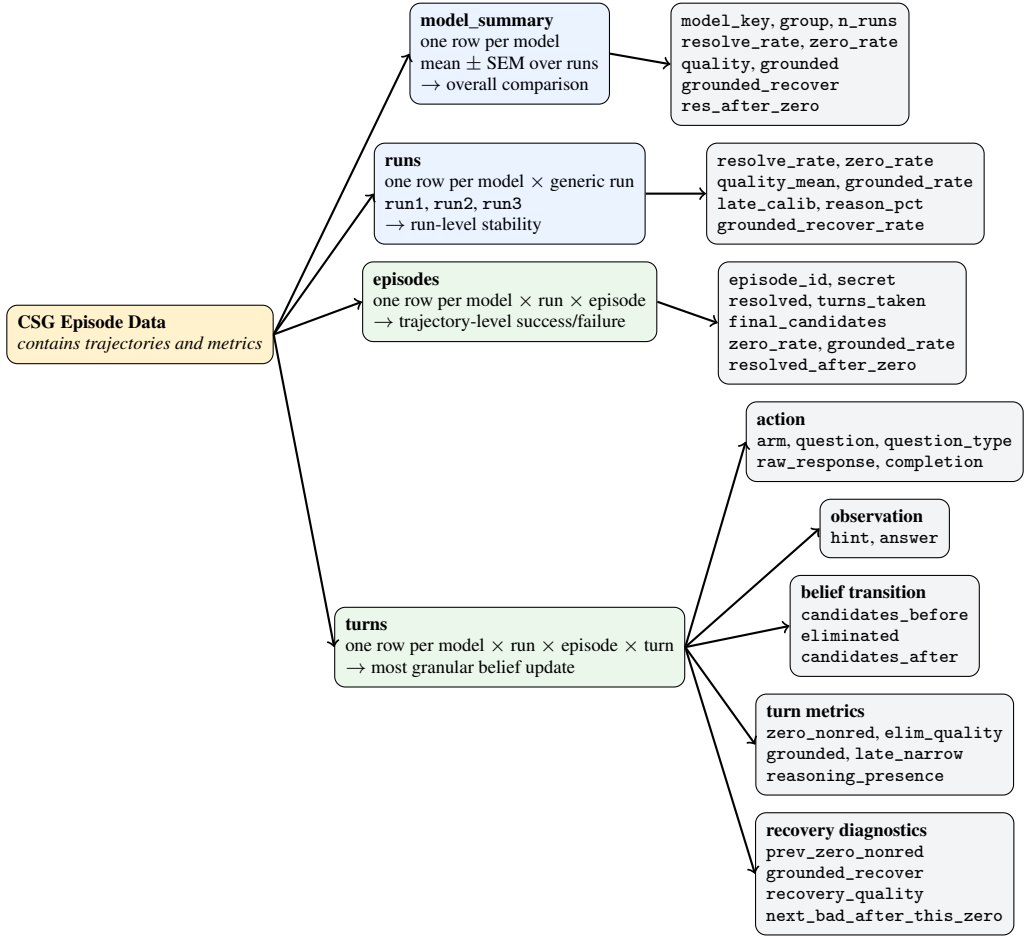
\begin{figure*}[t] \centering \scriptsize \begin{forest} for tree={ grow'=0, draw, rounded corners, align=left, parent anchor=east, child anchor=west, edge={->, thick}, l sep=8mm, s sep=2.2mm, inner xsep=4pt, inner ysep=3pt, font=\scriptsize, } [{\textbf{CSG Episode Data}\\ \emph{contains trajectories and metrics}}, fill=echoshade [{\textbf{model\_summary}\\ one row per model\\ mean \(\pm\) SEM over runs\\ \(\rightarrow\) overall comparison}, fill=lightblue [{\texttt{model\_key}, \texttt{group}, \texttt{n\_runs}\\ \texttt{resolve\_rate}, \texttt{zero\_rate}\\ \texttt{quality}, \texttt{grounded}\\ \texttt{grounded\_recover}\\ \texttt{res\_after\_zero}}, fill=lightgraybox] ] [{\textbf{runs}\\ one row per model \(\times\) generic run\\ \texttt{run1}, \texttt{run2}, \texttt{run3}\\ \(\rightarrow\) run-level stability}, fill=lightblue [{\texttt{resolve\_rate}, \texttt{zero\_rate}\\ \texttt{quality\_mean}, \texttt{grounded\_rate}\\ \texttt{late\_calib}, \texttt{reason\_pct}\\ \texttt{grounded\_recover\_rate}}, fill=lightgraybox] ] [{\textbf{episodes}\\ one row per model \(\times\) run \(\times\) episode\\ \(\rightarrow\) trajectory-level success/failure}, fill=lightgreen [{\texttt{episode\_id}, \texttt{secret}\\ \texttt{resolved}, \texttt{turns\_taken}\\ \texttt{final\_candidates}\\ \texttt{zero\_rate}, \texttt{grounded\_rate}\\ \texttt{resolved\_after\_zero}}, fill=lightgraybox] ] [{\textbf{turns}\\ one row per model \(\times\) run \(\times\) episode \(\times\) turn\\ \(\rightarrow\) most granular belief update}, fill=lightgreen [{\textbf{action}\\ \texttt{arm}, \texttt{question}, \texttt{question\_type}\\ \texttt{raw\_response}, \texttt{completion}}, fill=lightgraybox] [{\textbf{observation}\\ \texttt{hint}, \texttt{answer}}, fill=lightgraybox] [{\textbf{belief transition}\\ \texttt{candidates\_before}\\ \texttt{eliminated}\\ \texttt{candidates\_after}}, fill=lightgraybox] [{\textbf{turn metrics}\\ \texttt{zero\_nonred}, \texttt{elim\_quality}\\ \texttt{grounded}, \texttt{late\_narrow}\\ \texttt{reasoning\_presence}}, fill=lightgraybox] [{\textbf{recovery diagnostics}\\ \texttt{prev\_zero\_nonred}\\ \texttt{grounded\_recover}\\ \texttt{recovery\_quality}\\ \texttt{next\_bad\_after\_this\_zero}}, fill=lightgraybox] ] ] \end{forest} \caption{  CSG dataset structure. The dataset is organized from model-level summaries to run-level aggregates, episode-level trajectories, and turn-level belief-state transitions. File paths and raw run filenames are removed; run identities are exposed as generic identifiers. } \label{fig:csg_dataset_structure} \end{figure*}

\end{document}